\documentclass[11pt]{article}
\usepackage{amsmath,amssymb}
\usepackage{fullpage}
\usepackage[backref=page]{hyperref}
\usepackage{color}
\usepackage{tikz}
\usepackage{subfig}
\usetikzlibrary{matrix,arrows,decorations.pathreplacing}
\usepackage{algorithm}
\usepackage[noend]{algpseudocode}
\usepackage{xspace}
\usepackage{microtype}
\allowdisplaybreaks

\newenvironment{proof}{\noindent{\bf Proof : \ }}{\hfill$\Box$\par\medskip}

\newtheorem{theorem}{Theorem}[section]
\newtheorem{corollary}[theorem]{Corollary}
\newtheorem{lemma}[theorem]{Lemma}

\newtheorem{definition}[theorem]{Definition}

\newtheorem{claim}[theorem]{Claim}

\newtheorem{observation}[theorem]{Observation}

\newenvironment{proofof}[1]{\begin{trivlist} \item {\bf Proof
#1:~~}}
  {\qed\end{trivlist}}

\newcommand{\namedref}[2]{\hyperref[#2]{#1~\ref*{#2}}}
\newcommand{\thmlab}[1]{\label{thm:#1}}
\newcommand{\thmref}[1]{\namedref{Theorem}{thm:#1}}
\newcommand{\lemlab}[1]{\label{lem:#1}}
\newcommand{\lemref}[1]{\namedref{Lemma}{lem:#1}}
\newcommand{\claimlab}[1]{\label{claim:#1}}
\newcommand{\claimref}[1]{\namedref{Claim}{claim:#1}}
\newcommand{\corlab}[1]{\label{cor:#1}}
\newcommand{\corref}[1]{\namedref{Corollary}{cor:#1}}
\newcommand{\seclab}[1]{\label{sec:#1}}
\newcommand{\secref}[1]{\namedref{Section}{sec:#1}}
\newcommand{\applab}[1]{\label{app:#1}}
\newcommand{\appref}[1]{\namedref{Appendix}{app:#1}}

\newcommand{\figlab}[1]{\label{fig:#1}}
\newcommand{\figref}[1]{\namedref{Figure}{fig:#1}}
\newcommand{\alglab}[1]{\label{alg:#1}}
\renewcommand{\algref}[1]{\namedref{Algorithm}{alg:#1}}

\newcommand{\eqnlab}[1]{\label{eq:#1}}
\newcommand{\eqnref}[1]{\namedref{Equation}{eq:#1}}

\newcommand{\obslab}[1]{\label{obs:#1}}
\newcommand{\obsref}[1]{\namedref{Observation}{obs:#1}}

\def \tvd    {\mdef{d_{\mathsf{TV}}}}
\def \DISJ    {\mdef{\mathsf{DISJ}}}
\def \alg    {\mdef{\mathcal{A}}}

\def \adaptive    {\mdef{\textsc{AdaptiveStream}}}
\def \adaptiveone    {\mdef{\textsc{AdaptDistStream}}}
\def \countsketch    {\mdef{\textsc{CountSketch-M}}}
\def \taileps    {\mdef{tail\left(\frac{2}{\eps^2}\right)}}
\def \ams    {\mdef{\textsc{AMS-M}}}
\def \estimator    {\mdef{\textsc{Estimator-M}}}

\def \A    {\mdef{\mathbf{A}}}
\def \B    {\mdef{\mathbf{B}}}
\def \D    {\mdef{\mathbf{D}}}
\def \E    {\mdef{\mathbf{E}}}
\def \G    {\mdef{\mathbf{G}}}
\def \H    {\mdef{\mathbf{H}}}
\def \I    {\mdef{\mathbb{I}}}
\def \M    {\mdef{\mathbf{M}}}
\def \P    {\mdef{\mathbf{P}}}
\def \R    {\mdef{\mathbf{R}}}
\def \S    {\mdef{\mathbf{S}}}
\def \T    {\mdef{\mathbf{T}}}
\def \Q    {\mdef{\mathbf{Q}}}
\def \V    {\mdef{\mathbf{V}}}
\def \U    {\mdef{\mathbf{U}}}
\def \X    {\mdef{\mathbf{X}}}
\def \Y    {\mdef{\mathbf{Y}}}
\def \Z    {\mdef{\mathbf{Z}}}
\def \W    {\mdef{\mathbf{W}}}
\def \b    {\mdef{\mathbf{b}}}
\def \e    {\mdef{\mathbf{e}}}
\def \f    {\mdef{\mathbf{f}}}
\def \m    {\mdef{\mathbf{m}}}
\def \p    {\mdef{\mathbf{p}}}
\def \q    {\mdef{\mathbf{q}}}
\def \r    {\mdef{\mathbf{r}}}
\def \u    {\mdef{\mathbf{u}}}
\def \w    {\mdef{\mathbf{w}}}
\def \v    {\mdef{\mathbf{v}}}
\def \x    {\mdef{\mathbf{x}}}
\def \y    {\mdef{\mathbf{y}}}
\def \z    {\mdef{\mathbf{z}}}

\def \bzero {\mdef{\mathbf{0}}}

\newcommand\norm[1]{\left\lVert#1\right\rVert}
\newcommand{\PPr}[1]{\ensuremath{\mathbf{Pr}\left[#1\right]}}
\newcommand{\PPPr}[2]{\ensuremath{\underset{#1}{\mathbf{Pr}}\left[#2\right]}}
\newcommand{\Ex}[1]{\ensuremath{\mathbb{E}\left[#1\right]}}

\renewcommand{\O}[1]{\ensuremath{\mathcal{O}\left(#1\right)}}
\newcommand{\tO}[1]{\ensuremath{\tilde{\mathcal{O}}\left(#1\right)}}
\newcommand{\eps}{\epsilon}

\newcommand{\mdef}[1]{{\ensuremath{#1}}\xspace}  % Math Def which can also be used in normal text.
\DeclareMathOperator*{\argmin}{argmin}
\DeclareMathOperator*{\argmax}{argmax}
\DeclareMathOperator*{\polylog}{polylog}
\DeclareMathOperator*{\poly}{poly}
\DeclareMathOperator*{\median}{median}
\DeclareMathOperator*{\Var}{Var}
\DeclareMathOperator*{\Vol}{Vol}
\DeclareMathOperator*{\rank}{rank}

\newcommand{\superscript}[1]{\ensuremath{^{\mbox{\tiny{\textit{#1}}}}}\xspace}
\def \th {\superscript{th}}     % 'The i-th entry it a list...' --> i\th
\newcommand{\abs}[1]{\mdef{\left|#1\right|}}         % Absolute value
\newcommand{\ignore}[1]{}

\newif\ifnotes\notestrue %set this to true if notes are visible and to false (next line) if they should be hidden
\ifnotes
\newcommand{\samson}[1]{\textcolor{purple}{{\bf (Samson:} {#1}{\bf ) }} \marginpar{\tiny\bf
             \begin{minipage}[t]{0.5in}
               \raggedright S:
            \end{minipage}}}            							
\else
\newcommand{\samson}[1]{}
\fi

\hypersetup{
   colorlinks=true,%
   urlcolor=[rgb]{0.25,0.0,0.0},
   linkcolor=[rgb]{0.5,0.0,0.0},%
   citecolor=[rgb]{0,0.2,0.445},
}

\makeatletter
\renewcommand*{\@fnsymbol}[1]{\textcolor{blue}{\ensuremath{\ifcase#1\or *\or \dagger\or \ddagger\or
 \mathsection\or \triangledown\or \mathparagraph\or \|\or **\or \dagger\dagger
   \or \ddagger\ddagger \else\@ctrerr\fi}}}
\makeatother

\providecommand{\email}[1]{\href{mailto:#1}{\nolinkurl{#1}\xspace}}

\title{Non-Adaptive Adaptive Sampling on Turnstile Streams}
\date{\today}
\author{
Sepideh Mahabadi\thanks{Toyota Technological Institute at Chicago (TTIC). 
E-mail: \email{mahabadi@ttic.edu}}
\and
Ilya Razenshteyn\thanks{Microsoft Research. 
E-mail: \email{ilyaraz@microsoft.com}}
\and
David P. Woodruff\thanks{School of Computer Science, Carnegie Mellon University.
E-mail: \email{dwoodruf@cs.cmu.edu}}
\and
Samson Zhou\thanks{School of Computer Science, Carnegie Mellon University.  
E-mail: \email{samsonzhou@gmail.com}}
}

\begin{document}
\maketitle
\begin{abstract}
Adaptive sampling is a useful algorithmic tool for data summarization problems in the classical centralized setting, where the entire dataset is available to the single processor performing the computation. 
Adaptive sampling repeatedly selects rows of an underlying matrix $\A\in\mathbb{R}^{n\times d}$, where $n\gg d$, with probabilities proportional to their distances to the subspace of the previously selected rows. 
Intuitively, adaptive sampling seems to be limited to trivial multi-pass algorithms in the streaming model of computation due to its inherently sequential nature of assigning sampling probabilities to each row only after the previous iteration is completed. 
Surprisingly, we show this is not the case by giving the first one-pass algorithms for adaptive sampling on turnstile streams and using space $\poly(d,k,\log n)$, where $k$ is the number of adaptive sampling rounds to be performed. 

Our adaptive sampling procedure has a number of applications to various data summarization problems that either improve state-of-the-art or have only been previously studied in the more relaxed row-arrival model. 
We give the first relative-error algorithm for column subset selection on turnstile streams. 
We show our adaptive sampling algorithm also gives the first relative-error algorithm for subspace approximation on turnstile streams that returns $k$ noisy rows of $\A$. 
The quality of the output can be improved to a $(1+\eps)$-approximation at the tradeoff of a bicriteria algorithm that outputs a larger number of rows. 
We then give the first algorithm for projective clustering on turnstile streams that uses space sublinear in $n$. 
In fact, we use space $\poly\left(d,k,s,\frac{1}{\eps},\log n\right)$ to output a $(1+\eps)$-approximation, where $s$ is the number of $k$-dimensional subspaces. 
Our adaptive sampling primitive also provides the first algorithm for volume maximization on turnstile streams. 
We complement our volume maximization algorithmic results with lower bounds that are tight up to lower order terms, even for multi-pass algorithms. 
By a similar construction, we also obtain lower bounds for volume maximization in the row-arrival model, which we match with competitive upper bounds. 
\end{abstract}

\newpage

\section{Introduction}
Data summarization is a fundamental task in data mining, machine learning, statistics, and applied mathematics. 
The goal is to find a set $S$ of $k$ rows of a matrix $\A\in\mathbb{R}^{n\times d}$ that optimizes some predetermined function that quantifies how well $S$ represents $\A$. 
For example, row subset selection seeks $S$ to well-approximate $\A$ with respect to the spectral or Frobenius norm, subspace approximation asks to minimize the sum of the distances of the rows of $\A$ from $S$, while volume maximization wants to maximize the volume of the parallelepiped spanned by the rows of $S$. 
Due to their applications in data science, data summarization problems are particularly attractive to study for big data models. 

The streaming model of computation is an increasingly popular model for describing large datasets whose overwhelming size places restrictions on space available to algorithms. 
For turnstile streams that implicitly define $\A$, the matrix initially starts as the all zeros matrix and receives a large number of updates to its coordinates. 
Once the updates are processed, they cannot be accessed again and hence any information not stored is lost forever. 
The goal is then to perform some data summarization task after the stream is completed without storing $\A$ in its entirety. 

Adaptive sampling is a useful algorithmic paradigm that yields many data summarization algorithms in the centralized setting~\cite{DeshpandeV06, DeshpandeV07, DeshpandeRVW06}. 
The idea is that $S$ begins as the empty set and some row $\A_i$ of $\A$ is sampled with probability $\frac{\norm{\A_i}^p_2}{\norm{\A}_{p,2}^p}$, where $p\in\{1,2\}$ and $\norm{\A}_{p,q}=\left(\sum_{i=1}^n\left(\sum_{j=1}^d |A_{i,j}|^q\right)^{\frac{p}{q}}\right)^{\frac{1}{p}}$. 
As $S$ is populated, the algorithm \emph{adaptively samples} rows of $\A$, so that at each iteration, row $\A_i$ is sampled with probability proportional to the $p\th$ power of the distance of the row from $S$. 
That is, $\A_i$ is sampled with probability $\frac{\norm{\A_i(\I-\M^\dagger\M)}^p_2}{\norm{\A(\I-\M^\dagger\M)}_{p,2}^p}$, where $\M$ is the matrix formed by the rows in $S$. 
The procedure is repeated $k$ times until we obtain $k$ rows of $\A$, which then forms our summary of the matrix $\A$. 
Unfortunately, adaptive sampling seems like an inherently sequential procedure and thus the extent of its capabilities has not been explored in the streaming model. 

\subsection{Our Contributions}
In this paper, we show that although adaptive sampling seems like an iterative procedure, we do not need multiple passes over the stream to perform adaptive sampling. 
This is particularly surprising since \emph{any} row $\A_i$ of $\A$ can be made irrelevant, i.e., zero probability of being sampled, in future rounds if some row along the same direction of $\A_i$ is sampled in the present round. 
Yet somehow we must still output rows of $\A$ while storing a sublinear number of rows more or less non-adaptively. 
The challenge seems compounded by the turnstile model, since updates can be made to arbitrary elements of the matrix, but somehow we need to recover the rows with the largest norms. 
For example, if the last update in the stream is substantially larger than the previous updates, an adaptive sampler must return the entire row, even though this update could be an entry in any row of $\A$. 

To build our adaptive sampler, we first give an algorithm that performs a single round of sampling. 
Namely, given a matrix $\A\in\mathbb{R}^{n\times d}$ that is defined over a turnstile stream and post-processing query access to a matrix $\P\in\mathbb{R}^{d\times d}$, we first give $L_{p,q}$ samplers for $\A\P$. 
\begin{theorem}
Let $\eps>0$, $q=2$, and $p\in\{1,2\}$. 
There exists a one-pass streaming algorithm that takes rows of a matrix $\A\in\mathbb{R}^{n\times d}$ as a turnstile stream and post-processing query access to matrix $\P\in\mathbb{R}^{d\times d}$ after the stream, and with high probability, samples an index $i\in[n]$ with probability $\left(1\pm\O{\eps}\right)\frac{\norm{\A_i\P}_q^p}{\norm{\A\P}_{p,q}^p}+\frac{1}{\poly(n)}$. 
The algorithm uses $\poly\left(d,\frac{1}{\eps},\log n\right)$ bits of space. 
(See \thmref{thm:l2:sampling} and \thmref{thm:l12:sampling}.) 
\end{theorem}
We remark that our techniques can be extended to $p\in(0,2]$ but we only require $p\in\{1,2\}$ for the purposes of our applications. 
Now, suppose we want to perform adaptive sampling of a row of $\A\in\mathbb{R}^{n\times d}$ with probability proportional to its distance or squared distance from some subspace $\H\in\mathbb{R}^{i\times d}$, where $i$ is any integer. 
Then by taking $\P=\I-\H^\dagger\H$ and either $L_{1,2}$ or $L_{2,2}$ sampling, we select rows of $\A$ with probability roughly proportional to the distance or squared distance from $\H$. 
We can thus simulate $k$ rounds of adaptive sampling in a stream, despite its seemingly inherent sequential nature.
\begin{theorem}
Let $\A\in\mathbb{R}^{n\times d}$ be a matrix and $q_S$ be the probability of selecting a set $S\subset[n]$ of $k$ rows of $\A$ according to $k$ rounds of adaptive sampling with respect to either the distances to the selected subspace in each iteration or the squared distances to the selected subspace in each iteration.  
There exists an algorithm that takes inputs $\A$ through a turnstile stream and $\eps>0$, and outputs a set $S\subset[n]$ of $k$ indices such that if $p_S$ is the probability of the algorithm outputting $S$, then $\sum_S|p_S-q_S|\le\eps$. 
The algorithm uses $\poly\left(d,k,\frac{1}{\eps},\log n\right)$ bits of space. 
(See \thmref{thm:adaptive:sampler} and \thmref{thm:l12:adaptive:sampler}.) 
\end{theorem}
In other words, our output distribution is close in total variation distance to the desired adaptive sampling distribution. 

Our algorithm is the first to perform adaptive sampling on a stream; existing implementations require extended access to the matrix, such as in the centralized or distributed models, for subsequent rounds of sampling. 
Moreover, if the set $S$ of $k$ indices output by our algorithm is $s_1,\ldots,s_k$, then our algorithm also returns a set of rows $\r_1,\ldots,\r_k$ so that if $\R_0=\emptyset$ and $\R_i=\r_1\circ\ldots\circ\r_i$ for $i\in[k]$, then $\r_i=\u_{s_i}+\v_i$, where $\u_{s_i}=\A_{s_i}(\I-\R_i^\dagger\R_i)$ is the projection of the sampled row to the space orthogonal to the previously selected rows, and $\v_i$ is some small noisy vector formed by linear combinations of other rows in $\A$ such that $\norm{\v_i}_2\le\eps\norm{\u_{s_i}}_2$. 

Thus we do not return the true rows of $\A$ corresponding to the indices in $S$, but we output a small noisy perturbation to each of the rows, which we call noisy rows and suffices for a number of applications previously unexplored in turnstile streams. 
Crucially, the noisy perturbation in each of our output rows can be bounded in norm not only relative to the norm of the true row, but also relative to the residual. 
In fact, our previous example of a long stream of small updates followed a single arbitrarily large update shows that it is impossible to return the true rows of $\A$ in sublinear space. 
Since the arbitrarily large update can apply to \emph{any} entry of the matrix, the only way an algorithm can return the entire row containing the entry is if the entire matrix is stored. 

\paragraph{Column subset selection.}
In the row/column subset selection problem, the inputs are the matrix $\A\in\mathbb{R}^{n\times d}$ and an integer $k>0$, and the goal is to select $k$ rows/columns of $\A$ to form a matrix $\M$ to minimize $\norm{\A-\A\M^\dagger\M}_F$ or $\norm{\A-\M\M^\dagger\A}_F^2$. 
For the sake of presentation, we focus on the row subset selection problem for the remainder of this section. 
Since the matrix $\M$ has rank at most $k$, then $\norm{\A-\A\M^\dagger\M}_F\ge\norm{\A-\A^*_k}_F$, where $\A^*_k$ is the best rank $k$ approximation to $\A$. 
Hence, we would ideally like to obtain some guarantee for $\norm{\A-\A\M^\dagger\M}_F$ relative to $\norm{\A-\A^*_k}_F$. 
Such relative error algorithms were given in the centralized setting~\cite{DeshpandeRVW06,BoutsidisMD09,GuruswamiS12} and for row-arrival streams~\cite{CohenMM17,BravermanDMMUWZ18}, but no such results were previously known for turnstile streams. 
Our adaptive sampling framework thus provides the first algorithm on turnstile streams with relative error guarantees. 
\begin{theorem}
Given a matrix $\A\in\mathbb{R}^{n\times d}$ that arrives in a turnstile data stream, there exists a one-pass algorithm that outputs a set $\M$ of $k$ (noisy) rows of $\A$ such that 
\[\PPr{\norm{\A-\A\M^\dagger\M}_F^2\le16(k+1)!\norm{\A-\A^*_k}_F^2}\ge\frac{2}{3}.\]
The algorithm uses $\poly(d,k,\log n)$ bits of space. 
(See \thmref{thm:rss}.)
\end{theorem}

\paragraph{Subspace approximation.}
In the subspace approximation problem, the inputs are the matrix $\A\in\mathbb{R}^{n\times d}$ and an integer $k>0$ and the goal is to output a $k$-dimensional linear subspace $\H$ that minimizes $\left(\sum_{i=1}^n d(\A_i,\H)^p\right)^{\frac{1}{p}}$, where $p\in\{1,2\}$ and $d(\A_i,\H)=\norm{\A_i(\I-\H^\dagger\H)}_2$ is the distance from $\A_i$ to the subspace $\H$. 
A number of algorithms for the subspace approximation were given for the centralized setting~\cite{DeshpandeV07, FeldmanMSW10, ShyamalkumarV12, ClarksonW15} and more recently, \cite{LevinSW18} gave the first algorithm for subspace approximation on turnstile streams. 
The algorithm of \cite{LevinSW18} is based on sketching techniques and although it offers a superior $(1+\eps)$-approximation, their subspace has a larger number of rows and the rows may not originate from $\A$, whereas we select $k$ noisy rows of the matrix $\A$ to form the subspace. 
\begin{theorem}
Given $p\in\{1,2\}$ and a matrix $\A\in\mathbb{R}^{n\times d}$ that arrives in a turnstile data stream, there exists a one-pass algorithm that outputs a set $\Z$ of $k$ (noisy) rows of $\A$ such that 
\[\PPr{\left(\sum_{i=1}^n d(\A_i,\Z)^p\right)^{\frac{1}{p}}\le 4(k+1)!\left(\sum_{i=1}^n d(\A_i,\A^*_k)^p\right)^{\frac{1}{p}}}\ge\frac{2}{3},\]
where $\A^*_k$ is the best rank $k$ solution to the subspace approximation problem. 
The algorithm uses $\poly\left(d,k,\log n\right)$ bits of space. 
(See \thmref{thm:sap} and \thmref{thm:l12:sap}.)
\end{theorem}
Our adaptive sampling procedure also gives a bicriteria algorithm for a better approximation but allows the dimension of the subspace to be larger.
\begin{theorem}
Given $p\in\{1,2\}$, $\eps>0$, and a matrix $\A\in\mathbb{R}^{n\times d}$ that arrives in a turnstile data stream, there exists a one-pass algorithm that outputs a set $\Z$ of $\poly\left(k,\frac{1}{\eps},\log\frac{k}{\eps}\right)$ (noisy) rows of $\A$ such that 
\[\PPr{\left(\sum_{i=1}^n d(\A_i,\Z)^p\right)^{\frac{1}{p}}\le(1+\eps)\left(\sum_{i=1}^n d(\A_i,\A^*_k)^p\right)^{\frac{1}{p}}}\ge\frac{2}{3},\] 
where $\A^*_k$ is the best rank $k$ solution to the subspace approximation problem. 
The algorithm uses $\poly(d,k,\frac{1}{\eps},\log n)$ bits of space. 
(See \thmref{thm:sap:bicriteria}.)
\end{theorem}

\paragraph{Projective clustering.}
Projective clustering is an important problem for bioinformatics, computer vision, data mining, and unsupervised learning~\cite{Procopiuc17}. 
The projective clustering problem takes as inputs the matrix $\A\in\mathbb{R}^{n\times d}$ and integers $k>0$ for the target dimension of each subspace and $s>0$ for the number of subspaces, and the goal is to output $s$ $k$-dimensional linear subspaces $\H_1,\ldots,\H_s$ that minimizes $\left(\sum_{i=1}^n d(\A_i,\H)^p\right)^{\frac{1}{p}}$, where $p\in\{1,2\}$, $\H=\H_1\cup\ldots\cup\H_s$, and $d(\A_i,\H)$ is the distance from $\A_i$ to union $\H$ of $s$ subspaces $\H_1,\ldots,\H_s$. 
A number of streaming algorithms for projective clustering~\cite{BadoiuHI02, Har-PeledM04, Chen09, FeldmanMSW10} are based on the notion of core-sets, which are small numbers of weighted representative points. 
These results require a stream of (possibly high dimensional) points, which is equivalent to the row-arrival model and thus do not extend to turnstile streams. 
\cite{KerberR15} gives a turnstile algorithm based on random projections, but the algorithm requires space linear in the number of points. 
Thus our adaptive sampling procedure gives the first turnstile algorithm for projective clustering that uses space sublinear in the number of points. 
\begin{theorem}
Given $p\in\{1,2\}$, $\eps>0$ and a matrix $\A\in\mathbb{R}^{n\times d}$ that arrives in a turnstile data stream, there exists a one-pass algorithm that outputs a set $\S$ of $\poly\left(k,s,\frac{1}{\eps}\right)$ rows, which includes a union $\T$ of $s$ $k$-dimensional subspaces such that
\[\PPr{\left(\sum_{i=1}^n d(\A_i,\T)^p\right)^{\frac{1}{p}}\le(1+\eps)\left(\sum_{i=1}^n d(\A_i,\H)^p\right)^{\frac{1}{p}}}\ge\frac{2}{3},\]
where $\H$ is the union of $s$ $k$-dimensional subspaces that is the optimal solution to the projective clustering problem. 
The algorithm uses $\poly\left(d,k,s,\frac{1}{\eps},\log n\right)$ bits of space. 
(See \thmref{thm:pc}.)
\end{theorem}

\paragraph{Volume maximization.}
The volume maximization problem takes as inputs a matrix $\A\in\mathbb{R}^{n\times d}$ and a parameter $k$ for the number of selected rows, and the goal is to output $k$ rows $\r_1,\ldots,\r_k$ of $\A$ that maximize the volume of the parallelepiped spanned by the rows. 
\cite{indyk2018composable, mahabadi2019composable} give core-set constructions for volume maximization that approximate the optimal solution within a factor of $\tO{k}^{k/2}$ and $\O{k}^{k}$ respectively, and can be implemented in the row-arrival model. 
Their algorithms are based on spectral spanners and local search based on directional heights and do not immediately extend to turnstile streams. 
Hence our adaptive sampling procedure gives the first turnstile algorithm for volume maximization that uses space sublinear in the input size. 
\begin{theorem}
Given a matrix $\A\in\mathbb{R}^{n\times d}$ that arrives in a turnstile data stream and an approximation factor $\alpha>1$, there exists a one-pass algorithm that outputs a set $\S$ of $k$ noisy rows of $\A$ such that
\[\PPr{\alpha^k(k!)\Vol(\S)\ge\Vol(\M)}\ge\frac{2}{3},\]
where $\Vol(\S)$ is the volume of the parallelepiped spanned by $\S$ and $\M$ is a set of $k$ rows that maximizes the volume. 
The algorithm uses $\tO{\frac{ndk^2}{\alpha^2}}$ bits of space. 
(See \thmref{thm:vm}.)
\end{theorem}
We complement \thmref{thm:vm} with a lower bound for the volume maximization problem on turnstile streams that is tight up to lower order terms. 
Additionally, we give a lower bound for volume maximization in the random order row-arrival model, which we will also show is tight up to lower order terms. 
Our lower bounds complement the thorough lower bounds for extent problems given by \cite{AgarwalS15}. 
\begin{theorem}
There exists a constant $C>1$ so that any one-pass streaming algorithm that outputs a $C^k$ approximation to the volume maximization problem with probability at least $\frac{63}{64}$ in the random order row-arrival model requires $\Omega(n)$ bits of space. 
Moreover for any integer $p>0$, any $p$-pass turnstile streaming algorithm that gives an $\alpha^k$ approximation to the volume maximization problem requires $\Omega\left(\frac{n}{kp\alpha^2}\right)$ bits of space. 
(See \corref{thm:turnstile:vm} and \corref{cor:ra:turnstile}.)
\end{theorem}
Finally, we give a corresponding upper bound for volume maximization in the row-arrival model competitive with our lower bound. 
\begin{theorem}
Let $1<C<(\log n)/k$. 
There exists a one-pass streaming algorithm in the row-arrival model that computes a subset $\S$ of size $k$ of points in $\mathbb{R}^d$ such that 
\[\PPr{\O{Ck}^{k/2}\Vol(\S)\ge\Vol(\M)}\ge\frac{2}{3},\]
where $\Vol(\S)$ is the volume of the parallelepiped spanned by $\S$ and $\M$ is a set of $k$ rows that maximizes the volume. 
The algorithm uses $\O{n^{\O{1/C}} d}$ bits of space.
(See \lemref{volmax-tradeoff}.)
\end{theorem}

\subsection{Techniques}
Our first observation is that in many applications, the role played by adaptive sampling is to sample rows of a matrix $\A\in\mathbb{R}^{n\times d}$ with probability proportional to either the distance or the squared distance from some subspace $\H_j$ that we have already chosen by step $j$ of the sampling procedure. 
Adaptive sampling then imbues some randomness into the sampling procedure, which would otherwise reduce to the greedy paradigm of iteratively choosing the row of $\A(\I-\H_j^\dagger\H_j)$ with the largest squared norm. 
That is, we care more about the rows of $\A(\I-\H_j^\dagger\H_j)$ than the rows of $\A$. 

Thus our first component towards our adaptive sampling algorithm is an $L_{p,2}$ sampler with $p\in\{1,2\}$, which takes turnstile updates to $\A$ and post-processing query access to a matrix $\P\in\mathbb{R}^{d\times d}$ and outputs index $i\in[n]$ with probability roughly, i.e., within $(1\pm\eps)$ factor of $\frac{\norm{\A_i\P}_2^p}{\norm{\A\P}^p_{p,2}}$. 
By setting $\P=\I-\H_j^\dagger\H_j$, the probability of sampling each row of $\A\P$ then approximately follows the adaptive sampling distribution. 

\paragraph{$L_{p,2}$ sampler.}
We first describe how to sample rows of $\A$ when $\P$ is the identity matrix, so that we output an index $i\in[n]$ with probability roughly $\frac{\norm{\A_i}_2^p}{\norm{\A}^p_{p,2}}$ with $p\in\{1,2\}$.  
Our scheme generalizes a line of work for $\ell_p$ sampling~\cite{MonemizadehW10,SohlerW11,AndoniKO11,JowhariST11,JayaramW18}, where the input is a vector $\f$ of $n$ coordinates that are updated through a turnstile stream and the goal is to sample an index $i\in[n]$ with probability roughly $\frac{|f_i|^p}{\norm{\f}_p^p}$, from coordinates of a vector input to rows of a matrix input. 
These prior $\ell_p$ sampling algorithms have essentially followed the same template of performing a linear transformation on $\f$ to obtain a new vector $\z$, using an instance of CountSketch on $\z$ to recover a vector $\y$, and then running a statistical test on $\y$. 
If the statistical test fails, then the algorithm aborts; otherwise the coordinate of $\y$ with the maximum magnitude is output. 
The algorithm is repeated a number of times to ensure a high probability of success. 
 
Generalizing the template of $\ell_p$ sampling, we observe that if some scaling factor $t_i\in[0,1]$ is chosen uniformly at random, then $\PPr{\frac{\norm{\A_i}_2}{t_i^{1/p}}\ge K^{1/p}\norm{\A}_{p,2}}=\frac{\norm{\A_i}_2^p}{K\norm{\A}^p_{p,2}}$, where $K$ is any parameter that we choose.  
Thus if $\B_i=\frac{1}{t_i^{1/p}}\A_i$ for $i\in[n]$ and we temporarily suppose that only the row $x$ satisfies $\norm{\B_x}_2\ge T$, where $T=K^{1/p}\norm{\A}_{p,2}$ is the threshold, for only the index $i=x$, then our task would reduce to identifying $\B_x$ in sublinear space. 
To this end, if $\B$ is the matrix whose rows are $\B_1,\ldots,\B_n$, then we hash the rows of $\B$ to a CountSketch data structure to recover the row with the largest norm, which must necessarily be $\B_x$ if the error in CountSketch is small enough. 

Namely, if the error in CountSketch data structure is too large, then our algorithm will erroneously identify some scaled row as exceeding the threshold $T$ when the scaled row does not, or vice versa. 
Hence our algorithm must first run a statistical test to determine whether the error in the CountSketch data structure caused by the randomness of the data structure is sufficiently small. 
If the CountSketch error is determined to be too large by the statistical test, the algorithm aborts; otherwise the algorithm outputs the row with the largest norm if it exceeds $T$. 

Now there can still be some error if multiple rows have norms close to or exceeding $T$, but it turns out that by choosing the appropriate parameters, the probability that there exists a row whose norm exceeds $T$ is $\Omega\left(\frac{1}{K}\right)$ and the probability that the statistical test fails or that multiple rows have norms close to or exceeding $T$ is $\O{\frac{\eps}{K}}$, which incurs a relative $(1\pm\eps)$ perturbation of the sampling probabilities. 
Thus a single instance of the sampler outputs an index from roughly the desired distribution with probability $\Omega\left(\frac{1}{K}\right)$ and with probability $1-\Omega(\frac{1}{K})$, it aborts and outputs nothing. 
Hence for $p\in\{1,2\}$, we obtain a constant probability of success using $\poly\left(\frac{1}{\eps},\log n\right)$ space by setting $K=\poly\left(\frac{1}{\eps},\log n\right)$, repeating with $\O{K}$ instances, and taking the output of the first successful instance. 

It remains to argue that CountSketch and norm estimation generalize to $L_{p,2}$ error for matrices, which we do through standard arguments in \secref{sec:l2:sampler}. 
In fact, the data structures maintained by the generalized matrix CountSketch and $L_{p,2}$ norm estimation procedures are linear combinations of the rows of $\A$, so we can right multiply the rows that are stored in the $L_{p,2}$ sampler by $\P$ to simulate sampling rows of $\A\P$. 
In other words, if we had a stream of updates to the matrix $\A\P$, the resulting data structure on the stream would be equivalent to maintaining the data structure on a stream of updates to the matrix $\A$, and then multiplying each row of the data structure by $\P$ in post-processing. 
Hence we can also sample rows of $\A\P$ with probabilities proportional to the residual $\norm{\A_i\P}_2^p$, which will be crucial for our adaptive sampler. 

\paragraph{Adaptive sampler.} 
Recall that adaptive sampling iteratively samples rows of $\A\in\mathbb{R}^{n\times d}$ with probability proportional to the $p\th$ power of their distances from the subspace spanned by the rows that have already been sampled in previous rounds, for $k$ rounds. 
Thus if $\H_j$ is the matrix formed by the rows sampled by step $j$, then we would like to sample $i\in[n]$ with probability $\norm{\A_i(\I-\H_j^\dagger\H_j)}_2^p$ with the largest squared norm. 
Given our $L_{p,2}$ sampler, a natural approach is to run $k$ instances of the sampler throughout the stream. 
Once the stream completes, we use the first instance to sample a row of $\A$, which forms $\H_2$ (recall that $\H_1=\emptyset$ is ``used'' to sample in the first iteration). 
Since our $L_{p,2}$ sampler supports post-processing multiplication by a matrix $\P$, we subsequently use the $j\th$ instance to $L_{p,2}$ sample a row of $\A(\I-\H_j^\dagger\H_j)$, which we then append to $\H_j$ to form $\H_{j+1}$. 
Repeating this $k$ times, we would like to argue this simulates $k$ steps of adaptive sampling. 

The first issue with this approach is that our $L_{p,2}$ cannot return the original rows of $\A$, but only some noisy perturbation of the sampled row. 
It is easy to see that returning the noisy rows of $\A$ is unavoidable for sublinear space by considering a stream whose final update to some random coordinate is arbitrarily large, while the previous updates were small. 
Then the row containing the coordinate of the final update should be sampled with large probability, but that row can only be completely recovered if all entries of the matrix are stored. 
Fortunately we show that if we sample the index $x$, then we output a row $\r=\A_x+\v$, where the noisy component $\v$ is a linear combination of rows of $\A$ that satisfies $\norm{\v}_2\le\eps\norm{\A_x}_2$. 
Thus, the norms of the sampled rows are somewhat preserved. 

On the other hand, sampling noisy rows of $\A$ rather than the original rows of $\A$ can drastically alter the subspace spanned by the matrix formed by the rows. 
This in turn can significantly alter the sampling probabilities in future rounds. 
Consider the following example, which is depicted in \figref{fig:bad}. 
Let $\A$ be a matrix that has $\begin{bmatrix} 0 & 1\end{bmatrix}$ for half of its rows and $\begin{bmatrix} M & 0\end{bmatrix}$ for some large $M>0$ for the other half of its rows. 
Then with large probability we should sample some row $\u=\begin{bmatrix} M & 0\end{bmatrix}$ in the first step. 
However, due to noise in the sampler, we will actually obtain some noisy row $\v=\begin{bmatrix} M' & m\end{bmatrix}$, where $M'\approx M$ and $m\neq 0$. 
In \figref{fig:bad}, we depict $\u$ with the blue vector and $\v$ with the red vector. 

Now in the second round, if we had sampled $\u$ in the first round, then the only possible output of the adaptive sampler is a row $\begin{bmatrix} 0 & 1\end{bmatrix}$, since all the rows $\begin{bmatrix} M & 0\end{bmatrix}$ are contained in the subspace spanned by $\u$. 
However, since we actually sampled $\v$ in the first round, then the distance from $\u$ to the subspace spanned by $\v$ is nonzero. 
Furthermore, since $M$ is large, then it actually seems likely that we might sample a row $\begin{bmatrix} M & 0\end{bmatrix}$ rather than $\begin{bmatrix} 0 & 1\end{bmatrix}$. 
Thus we might sample some row that we should have not sampled or worse, we might repeatedly sample the same row! 
Pictorially, the blue vector $\u$ in \figref{fig:bad} has no projection away from itself, but results in the rightmost green vector when projected to the red vector $\v$, and thus $\u$ might be sampled \emph{again} with high probability. 

Similarly, the noisy perturbations may cause us to completely avoid rows that we should have sampled with nonzero probability if we had access to the original rows. 
In fact, this example shows that we cannot guarantee that our adaptive sampler gives a multiplicative $(1+\eps)$-approximation to the true sampling probabilities of each row in any round. 

\begin{figure*}
\centering
\subfloat[Troublesome distorted sampling probabilities.\figlab{fig:bad}]{
\begin{tikzpicture}
\draw[->,thick,blue](0,0) -- (4,0);
\draw[->,thick](0,0) -- (0,1);
\draw[->,thick,green](4/17,1/17) -- (0,1);
\draw[->,thick,red](0,0) -- (4,1);
\draw[->,thick,green](64/17,16/17) -- (4,0);
\end{tikzpicture}
}
\qquad
\qquad
\subfloat[Actual distorted sampling probabilities.\figlab{fig:actual}]{
\begin{tikzpicture}
\draw[->,thick,blue](0,0) -- (4,0);
\draw[->,thick](0,0) -- (0,1);
\draw[->,thick,green](10/101,1/101) -- (0,1);
\draw[->,thick,red](0,0) -- (4,0.4);
\draw[->,thick,green] (400/101,40/101) -- (4,0);
\end{tikzpicture}
}
\caption{Distortion of sampling probabilities by $L_{p,2}$ sampler. 
Suppose we should have sampled the blue vector but instead we obtain the red vector from the noisy output. 
Then the sampling probabilities in the second round will be the norms of the green vectors, even though the probability of sampling the blue vector in the second round should \emph{actually be zero} if we had sampled the actual blue vector. 
Thus the sampling probabilities are distorted. 
In particular, we are worried that in \figref{fig:bad}, the blue vector might be sampled again in the second round because the projection to the red vector has large norm. 
We show \figref{fig:bad} is unlikely and the actual scenario is more like \figref{fig:actual}, where the sampling probabilities are only perturbed by a small additive amount. 
}
\figlab{fig:probs}
\end{figure*}
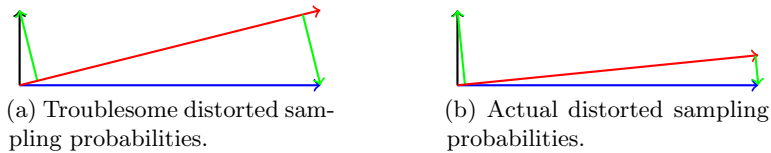

Our key observation is that the noisy row $\r$ output by our $L_{p,2}$ sampler not only has a noisy component $\v$ small in norm compared to $\A_x$, but also the component of $\v$ in the space orthogonal to $\A_x$ must be small. 
That is, $\r$ can also be written as $\r=\A_x+\w$, where $\norm{\w\Q}_2\le\eps\norm{\A_x}_2\frac{\norm{\A\Q}_F}{\norm{\A}_F}$ for \emph{any} projection matrix $\Q$. 
This tighter bound in any orthogonal direction allows us to bound in subsequent rounds the \emph{additive} error of sampling probabilities, which are based on the vector lengths in orthogonal directions. 
Thus, we show that with high probability, \figref{fig:bad} cannot happen and our actual situation is more like that in \figref{fig:actual}. 

Namely, if we write an orthonormal basis $U$ for the actual rows of $\A$ and an orthonormal basis $W$ for the noisy rows that we sample, we can show that the norm of a row projected onto $W$ has a small additive perturbation from when it is projected onto $U$. 
Thus we require the construction of orthonormal bases $U$ and $W$ from which we can easily extract the sampling probabilities of rows both with respect to the original rows and to the noisy rows. 
We achieve this by designing $U$ so that the first basis vectors of $U$ are precisely the true sampled rows of $\A$, followed by the noisy perturbations for each sample. 
We then argue that if we design $W$ so that the first basis vectors of $W$ are precisely the noisy rows that we sample, then the coefficients of each row represented in terms of $U$ and $W$ have only a small additive difference. 
By summing across all rows, then we can bound the total variation distance between sampling with the noisy rows of $\A$ and sampling with the actual rows of $\A$. 
That is, in each adaptive sampling round where we use a noisy row obtained from our $L_{p,2}$ sampler with error parameter $\eps>0$ rather than the actual row of $\A$, the total variation distance in the sampling distribution increases by an additive $\O{\eps}$. 
Now since we use the first sampled noisy row for $k-1$ additional rounds, the second sampled noisy row for $k-2$ additional rounds, and so forth, then by an inductive argument, the total variation distance between $k$ rounds of our algorithm and $k$ rounds of adaptive sampling is $\O{k^2\eps}$.

\paragraph{Applications.}
For many applications on turnstile streams, we show it suffices in each step to obtain a noisy row that is orthogonal to the previously selected rows, sampled with probability proportional to the $p\th$ power of the distance to the subspace spanned by those rows. 
Thus for $p=1$, our adaptive sampler allows us to perform residual based sampling in place of subspace embedding techniques used by previous work in various applications~\cite{KerberR15, SohlerW18, LevinSW18}. 
Additionally for $p=2$, our adaptive sampler allows us to simulate volume sampling, which has a wide range of applications~\cite{DeshpandeV06, DeshpandeRVW06, DeshpandeV07, CivrilM09}. 

For volume maximization on turnstile streams, we use a combination of our $L_{2,2}$ sampler and our generalized CountSketch data structure to simulate an approximation to the greedy algorithm of choosing the row with the largest residual at each step. 
If the largest residual found by CountSketch exceeds a certain threshold, we use that row; otherwise any row output by our adaptive sampler will be a good approximation to the row with the largest residual. 
Thus the volume of the parallelepiped spanned by these rows is a good approximation to the optimal solution. 

\paragraph{Volume Maximization.}
We provide lower bounds for volume maximization in turnstile streams and the row-arrival model through reductions from the Gap $\ell_\infty$ problem and the distributional set-disjointness problem, respectively. 
For both cases we show that embedding the same instance across multiple columns gives hardness of approximating within a factor with exponential dependency on $k$. 
For our algorithmic results in the row-arrival model, we first note that the composable core-set techniques of \cite{indyk2018composable} automatically gives a streaming algorithm for volume maximization. 
In fact, \cite{indyk2018composable} shows the stronger guarantee that any composable core-set for the directional height of a point set suffices to give a good approximation for volume maximization. 
Using this idea, we give a dimensionality reduction algorithm for volume maximization in the row-arrival model competitive with our lower bounds by embedding the input into a lower dimensional space. 

Recall that from Johnson-Lindenstrauss, right multiplication by a random matrix whose entries are drawn i.i.d from a Gaussian distribution suffices to preserve the directional heights of the points in an optimal set by a constant factor, say $2$, and thus the volume of the largest set of $k$ points is only distorted by a factor of $2^k$. 
We then prove that for every other subset of $k$ points, their volume does not increase by too much by showing that the eigenvalues of the matrix representation of the points are preserved by some factor $C$ with very high probability. 
Thus taking a union bound over all subsets of $k$ points, all volumes are preserved by a factor $C^k$ and we obtain dimensionality reduction of the problem by applying right multiplication of the random matrix to each of the input rows.

\paragraph{Paper Organization.}
We first handle $L_{2,2}$ sampling in \secref{sec:l2:sampler}. 
Using the $L_{2,2}$ sampler, we build an adaptive sampler in \secref{sec:noisy:adaptive} that samples rows with probability proportional to the squared distances of the subpsace spanned by previously selected rows. 
We show the applications of our adaptive sampler in \secref{sec:apps}, including projective clustering, subspace approximation, column subset selection, and volume maximization. 
We give lower bounds for volume maximization in \secref{sec:vm}, showing that our adaptive sampler gives nearly optimal algorithms in turnstile streams. 
Finally we give algorithms for volume maximization in the row-arrival model in \secref{sec:vm:ra}, competitive with the lower bounds in \secref{sec:vm}. 
For completeness, we detail $L_{1,2}$ sampling and adaptive sampling rows with probability proportional to the distances of the subspace spanned by previously selected rows in \appref{app:appendix}. 

\subsection{Preliminaries}
For any positive integer $n$, we use the notation $[n]$ to represent the set $\{1,\ldots,n\}$. 
We use the notation $x=(1\pm\eps)y$ to denote the containment $(1-\eps)y\le x\le(1+\eps)y$. 
We write $\poly(n)$ to denote some fixed constant degree polynomial in $n$ but we write $\frac{1}{\poly(n)}$ to denote some arbitrary degree polynomial in $n$. 
When an event has probability $1-\frac{1}{\poly(n)}$ of occurring, we say the event occurs with high probability. 
We use $\tO{\cdot}$ to omit lower order terms and similarly $\polylog(n)$ to omit terms that are polynomial in $\log n$.  

For our purposes, a turnstile stream will implicitly define a matrix $\A\in\mathbb{R}^{n\times d}$ through a sequence of $m$ updates. 
We use $\A_i$ to denote the $i\th$ row of $\A$ and $A_{i,j}$ to denote the $j\th$ entry of $\A_i$. 
The matrix $\A$ initially starts as the all zeros matrix. 
Each update in the stream has the form $(i_t,j_t,\Delta_t)$, where $t\in[m]$, $i_t\in[n]$, $j_t\in[d]$, and $\Delta_t\in\{-M,-M+1,\ldots,M-1,M\}$ for some large positive integer $M$. 
The update then induces the change $A_{i_t,j_t}\gets A_{i_t,j_t}+\Delta_t$ in $\A$.  
We assume throughout that $m,M=\poly(n)$ and $n\gg d$. 
We will typically only permit one pass over the stream, but for multiple passes the order of the updates remains the same in each pass. 

In the row-arrival model, the stream has length $n$ and the $i\th$ update in the stream is precisely row $\A_i$. 
Again we restrict each entry $A_{i,j}$ of $\A$ to be in the range $\{-M,-M+1,\ldots,M-1,M\}$ for some large positive integer $M=\poly(n)$. 
We assume that $\A$ can be adversarially chosen in the row-arrival model, but for the random order row-arrival model, once the entries of $\A$ are chosen, an arbitrary permutation of the rows of $\A$ is chosen uniformly at random, and the rows of that permutation constitute the stream. 
For the problems that we consider, the optimal solution is invariant to permutation of the rows of $\A$, so the random order does not impact the desired solution. 
Observe that algorithms for turnstile streams can be used in the row-arrival model, but not necessarily vice versa. 

We use $\I_k$ to denote the $k\times k$ identity matrix and we drop the subscript when the dimensions are clear. 
We use the notation $\A=\A_1\circ\A_2\circ\ldots\circ\A_n$ to denote that the matrix $\A$ is formed by the rows $\A_1,\ldots,\A_n$ and the notation $\A^\top$ to denote the transpose of $\A$. 
For a matrix $\M\in\mathbb{R}^{k\times d}$ with linearly independent rows, we use $\M^\dagger\in\mathbb{R}^{d\times k}$ to denote the Moore-Penrose pseudoinverse of $\M$, so that $\M^\dagger=\M^\top(\M\M^\top)^{-1}$ and $\M\M^{\dagger}=\I_k$. 

\begin{definition}[Vector/matrix norms]
For a vector $\v\in\mathbb{R}^n$, we have the Euclidean norm $\norm{\v}_2=\sqrt{\sum_{i=1}^n v_i^2}$ and more generally, $\norm{\v}_p=\left(\sum_{i=1}^n |v_i|^p\right)^{\frac{1}{p}}$. 
For a matrix $\A\in\mathbb{R}^{n\times d}$, we denote the Frobenius norm of $\A$ by $\norm{\A}_F=\sqrt{\sum_{i=1}^n\sum_{j=1}^d A_{i,j}^2}$. 
More generally, we write the $L_{p,q}$ norm of $\A$ by $\norm{\A}_{p,q}=\left(\sum_{i=1}^n\left(\sum_{j=1}^d |A_{i,j}|^q\right)^{\frac{p}{q}}\right)^{\frac{1}{p}}$, so that $\norm{\A}_F=\norm{\A}_{2,2}$. 
\end{definition}
For $\A\in\mathbb{R}^{n\times d}$, we use $\A_{tail(b)}$ to denote $\A$ with the $b$ rows of $\A$ with the largest Euclidean norm set to zeros. 
\begin{definition}[$L_{p,q}$ sampling]
Let $\A\in\mathbb{R}^{n\times d}$, $0\le\eps<1$, and $p,q>0$. 
An \emph{$L_{p,q}$ sampler} with $\eps$-relative error is an algorithm that outputs an index $i\in[n]$ such that for each $j\in[n]$,
\[\PPr{i=j}=\frac{\norm{\A_j}_q^p}{\norm{\A}^p_{p,q}}(1\pm\eps)+\O{n^{-c}},\]
for some arbitrarily large constant $c\ge 1$. 
In each case, the sampler is allowed to output fail with some probability $\delta$, in which case it must output $\bot$. 
When the underlying matrix is just a vector, i.e., $d=1$, we drop the $q$ term and call such an algorithm an $L_p$ sampler.  
\end{definition}

\begin{definition}[Adaptive sampling]
Let $\A\in\mathbb{R}^{n\times d}$ be a matrix from which we wish to sample and $\M\in\mathbb{R}^{m\times d}$ be a matrix corresponding to a specific subspace. 
For $p\in\{1,2\}$, an \emph{adaptive sampler} is an algorithm that outputs an index $i\in[n]$ and the corresponding row $\A_i$ such that for each $j\in[n]$,
\[\PPr{i=j}=\frac{\norm{\A_j\P}_2^p}{\norm{\A\P}^p_{p,2}},\]
where $\P=\I-\M^\dagger\M$. 
\end{definition}
In typical applications, we will wish to perform $k$ rounds of adaptive sampling with subspaces $\M_1,\M_2,\ldots,\M_k$, where $\M_1$ is the all zeros matrix, and each $\M_i$ will consist of the rows sampled from rounds $1$ to $i-1$. 
We will use the term adaptive sampling to refer to both a single round of sampling and multiple rounds of sampling interchangeably when the context is clear. 

Note that the adaptive sampling for input matrices $\A\in\mathbb{R}^{n\times d}$ and $\M\in\mathbb{R}^{m\times d}$ can be seen as $L_{p,2}$ sampling on an input matrix $\A\P$ with $\eps=0$ and $\P=\I-\M^\dagger\M$, but returning the row $\A_i$ instead of $\A_i\P$. 

\paragraph{AMS and CountSketch.}
We will refer to the classical AMS and CountSketch algorithms for intuition, but we require more generalized versions that we will present in \secref{sec:l2:sampler}. 
For the sake of completeness, the classical AMS algorithm~\cite{AlonMS99} can be formulated as taking a matrix $\A\in\mathbb{R}^{n\times d}$ through a turnstile stream and using $\O{\frac{1}{\eps^2}\log^2 n}$ bits of space to output a $(1+\eps)$-approximation to $\norm{\A}_F$ with high probability. 
For each entry $A_{i,j}$, the algorithm generates a random sign $h_{i,j}$ and maintains $S=\sum_{i=1}^n\sum_{j=1}^d h_{i,j}A_{i,j}$ throughout the stream.  
At the end of the stream, the algorithm uses $S^2$ as its estimator for $\norm{\A}_F^2$. 
By running $\O{\frac{1}{\eps^2}}$ instances of the estimator and taking the mean, the variance of the estimator decreases. 
By taking the median of $\O{\log n}$ means, the estimator succeeds with high probability.  
By taking $d=1$, the AMS algorithm can also be used to approximate $\norm{\v}_2$ for any vector $\v\in\mathbb{R}^n$ whose entries are updated in a turnstile stream. 

The classical CountSketch algorithm~\cite{CharikarCF04} can be used to find all entries $A_{i,j}$ of a matrix $\A\in\mathbb{R}^{n\times d}$ that arrives implicitly through a turnstile stream such that $|A_{i,j}|\ge\eps\norm{\A}_F$ for a constant input parameter $\eps>0$. 
The CountSketch data structure maintained by the algorithm is a $r\times b$ table $T$. 
For each row $k\in[r]$, a random sign $h_k(i,j)$ is generated for each $A_{i,j}$ and each entry $A_{i,j}$ is randomly hashed to a bucket $g_k(i,j)\in[b]$ in row $k$. 
Each bucket $\ell$ in row $k$ then maintains $T_{k,\ell}=\sum_{(i,j):g_k(i,j)=\ell} A_{i,j}h_k(i,j)$, which is a linear combination of the entries assigned to the bucket along with the random signs for the entries. 
Then for each row $k$, the estimator for $A_{i,j}$ is $T_{k,g_k(i,j)}h_k(i,j)$, which is the value in the bucket of row $k$ assigned to $A_{i,j}$, rescaled by the random sign. 
Finally, the estimator for $A_{i,j}$ by the CountSketch data structure is the median of the estimators of $A_{i,j}$ across all rows. 
It can be seen that $r=\O{\log n}$ and $b=\O{\frac{1}{\eps^2}}$ suffices to estimate each entry of $A_{i,j}$ within an additive $\frac{\eps}{4}\norm{\A}_F$ factor with high probability. 
Thus if $A_{i,j}\ge\eps\norm{\A}_F$, its estimated value will exceed $\frac{\eps}{2}\norm{\A}_F$ and will be output by CountSketch given an accurate estimation of $\norm{\A}_F$, such as by AMS. 

\vskip 0.2in
We require the following definition of total variation distance to bound the difference between two probability distributions, such as the ``ideal'' sampling distributions compared to the distributions provided by our algorithms. 
\begin{definition}[Total variation distance]
Let $\mu,\nu$ be two probability distributions on a finite domain $\Omega$. 
Then the total variation distance between $\mu$ and $\nu$ is defined as $\tvd(\mu,\nu)=\frac{1}{2}\sum_{x\in\Omega}|\mu(x)-\nu(x)|$. 
\end{definition}

\section{$L_{2,2}$ Sampler}
\seclab{sec:l2:sampler}
In this section, we first describe a turnstile streaming algorithm that takes a matrix $\A\in\mathbb{R}^{n\times d}$ that arrives as a data stream and post-processing query access to a matrix $\P\in\mathbb{R}^{d\times d}$, and outputs an index $i\in[n]$ of a row of $\A\P$ sampled with probability roughly $\frac{\norm{\A_i\P}_2^2}{\norm{\A\P}_F^2}$. 

\paragraph{High level idea.}
First suppose we only wanted to sample a row $i$ of $\A\in\mathbb{R}^{n\times d}$ with probability roughly $\frac{\norm{\A_i}_2^2}{\norm{\A}_F^2}$. 
By multiplying each row $\A_i$ with a random scaling factor $\frac{1}{\sqrt{t_i}}$, where $t_i\in[0,1]$ is chosen independently and uniformly at random, the probability that $\frac{1}{t_i}\norm{\A_i}_2^2\ge\norm{\A}_F^2$ is precisely the probability that $t_i\le\frac{\norm{\A_i}_2^2}{\norm{\A}_F^2}$, which is the desired probability of sampling row $i$. 

Now suppose only one row $\A_i$ satisfies $\frac{1}{t_i}\norm{\A_i}_2^2\ge\norm{\A}_F^2$, so that we would like to output $\A_i$. 
If we stored all rows of $\A$ as well as all scaling factors $t_j$, then we could identify and output this row, but the required space would be linear in the input size. 
Instead, we hash all scaled rows $\frac{1}{t_j}\norm{\A_j}$ to a number of buckets in a CountSketch data structure. 
Observe that if $\frac{1}{t_i}\norm{\A_i}_2^2\ge\norm{\A}_F^2$ for only one index $i$, then $\frac{1}{\sqrt{t_i}}\A_i$ must also be the scaled row with the largest norm. 
Moreover, it turns out that the mass of $\sum_{j=1}^n\frac{1}{t_j}\norm{\A_j}_2^2$ is dominated by a small number of rows.  
Hence with a sufficiently large number of buckets, the scaled row $i$ is the heavy hitter with the largest norm among all the heavy hitters of the scaled rows and so CountSketch will ideally identify the row $i$. 

This approach can fail due to two reasons. 
The first potential issue is if the accuracy of CountSketch does not suffice to identify the row $\A_i$ due to the noise from the tail of the mass of $\sum_{j=1}^n\frac{1}{t_j}\norm{\A_j}_2^2$. 
That is, if the noise of the tail due to the selection of the scaling factors $t_j$ prevents CountSketch from successfully identifying the heavy hitters, then this approach will fail. 
We can run a statistical test to identify when the noise is too large and preemptively abort accordingly. 
Moreover, if the CountSketch data structure maintains enough buckets, then the noise being sufficiently small happens with probability $\Omega(\eps)$, so we can run $\O{\frac{1}{\eps}}$ instances of the algorithm in parallel and take the first instance that does not abort. 

A separate issue is resolving the assumption that only one row $\A_i$ satisfies $\frac{1}{t_i}\norm{\A_i}_2^2\ge\norm{\A}_F^2$. 
As it turns out, many rows can exceed this threshold, but if we instead require $\frac{1}{t_i}\norm{\A_i}_2^2\ge\frac{1}{\eps}\norm{\A}_F^2$, then the probability that some row exceeds this threshold is $\Theta(\eps)$. 
The probability that multiple rows exceed this higher threshold is now $\O{\eps^2}$. 
Our algorithm outputs the row with the largest norm when some row exceeds the threshold, so in the case where multiple rows exceed the higher threshold we attribute the output to possible sampling probability perturbation. 
Hence the probability that multiple rows exceed the higher threshold only slightly perturbs the sampling probability of each row by a $(1\pm\eps)$ factor. 
Thus we can again repeat $\O{\frac{1}{\eps}}$ times until some row $\A_i$ satisfies $\frac{1}{t_i}\norm{\A_i}_2^2\ge\norm{\A}_F^2$. 
Similarly, if the error from CountSketch causes an inaccurate estimation of the row with the largest norm, then we might think the heaviest row does not exceed the threshold when it does in reality or vice versa. 
Fortunately, this only occurs when the row with the largest norm is very close to the threshold, which we again show only causes the sampling probability of each row to perturb by a $(1\pm\eps)$ factor. 

For technical reasons, we further increase the threshold and thus run a larger number of instances in parallel to avoid failure. 
We note that although CountSketch successfully identifies the row $i$, it can only output a noisy perturbation of $\A_i$. 
That is, it can only output some row $\r=\A_i+\v$, where the noisy component $\v$ satisfies $\norm{\v}_2\le\eps\norm{\A_i}_2$. 

Finally, we note that these procedures are all performed through linear sketches and that each bucket stores aggregate rows of the matrix $\A$. 
Thus if we had a stream of updates to the matrix $\A\P$, the resulting data structure would be equivalent to maintaining the data structure on a stream of updates to the matrix $\A$, and then multiplying each row of the data structure by $\P$ post-processing. 
Hence we can also sample rows of $\A\P$ with probabilities proportional to $\norm{\A_i\P}_2^2$.  

\subsection{Streaming Algorithms with Post-Processing}
We require generalizations of the celebrated AMS~\cite{AlonMS99} and CountSketch~\cite{CharikarCF04} algorithms to handle Frobenius norm estimation of $\A\P$ and to output the rows of $\A\P$ whose norm exceed a certain fraction of the total Frobenius norm, respectively. 
These generalizations are streaming algorithms that perform their desired function in low space even though query access to $\P$ is only provided after the stream ends. 

\begin{algorithm}[!htb]
\caption{Basic algorithm that estimates $\norm{\A\P}_F$, where $\P$ is a post-processing matrix}
\alglab{alg:ams:basic}
\begin{algorithmic}[1]
\Require{Matrix $\A\in\mathbb{R}^{n\times d}$, query access to matrix $\P\in\mathbb{R}^{d\times d}$ after the stream ends, constant parameter $\eps>0$.}
\Ensure{$(1+\eps)$-approximation of $\norm{\A\P}_F$.}
\State{Let $h_i\in\{-1,+1\}$ be $4$-wise independent for $i\in[n]$.}
\State{Let $\v\in\mathbb{R}^{1\times d}$ be a vector of zeros.}
\State{\textbf{Streaming Stage:}}
\For{each update $\Delta_t$ to entry $A_{i,j}$}
\State{Add $\Delta_t\cdot h_i$ to $v_j$.}
\EndFor
\State{\textbf{Processing $\P$ Stage:}}
\State{Output $\norm{\v\P}_2$.}
\end{algorithmic}
\end{algorithm}

We give in \algref{alg:ams:basic} the generalization of the AMS~\cite{AlonMS99} algorithm that estimates $\norm{\A\P}_F^2$, where $\A$ arrives in a stream and post-processing query access to $\P$ is given after the stream ends. 
Moreover, \algref{alg:ams:basic} is a linear sketch, so it can also be used to estimate $\norm{\A\P-\M}_F^2$ for a second arbitrary post-processing matrix $\M\in\mathbb{R}^{n\times d}$. 
\begin{lemma}
\lemlab{lem:ams}
Given a constant $\eps>0$, there exists a one-pass streaming algorithm \ams{} that takes updates to entries of a matrix $\A\in\mathbb{R}^{n\times d}$, as well as query access to post-processing matrices $\P\in\mathbb{R}^{d\times d}$ and $\M\in\mathbb{R}^{n\times d}$ that arrive after the stream, and outputs a quantity $\hat{F}$ such that $(1-\eps)\norm{\A\P-\M}_F\le\hat{F}\le(1+\eps)\norm{\A\P-\M}_F$. 
The algorithm uses $\O{\frac{d}{\eps^2}\log^2 n}$ bits of space and succeeds with high probability. 
\end{lemma}
\begin{proof}
Recall that the classic AMS estimator~\cite{AlonMS99} takes a vector $\f\in\mathbb{R}^{n\times 1}$ that arrives as a data stream and outputs an estimate $\hat{f}$ such that $(1-\eps)\norm{\f}_2\le\hat{f}\le(1+\eps)\norm{\f}_2$. 
The algorithm generates $4$-wise independent signs $s_i\in\{-1,+1\}$ for each $i\in[n]$ and maintains $\sum_{i=1}^n s_if_i$ in the stream. 
At the end of the stream, $\left(\sum_{i=1}^n s_if_i\right)^2$ is a good estimator for $\norm{\f}_2^2$. 
The algorithm can then be run $\O{\frac{1}{\eps^2}}$ times in parallel, taking the mean of the instances to decrease the variance and output a $(1+\eps)$-approximation for $\norm{\f}_2$. 
Taking the median of $\O{\log n}$ estimators further increases the probability of success to $1-\frac{1}{\poly(n)}$. 

For a matrix $\A\in\mathbb{R}^{n\times d}$ that arrives as a data stream, the Frobenius norm of $\A$ can then be approximated by running a classic AMS estimator for each of the $d$ columns of $\A$. 
That is, each row of $\A$ can be given a random sign and the algorithm stores the sum of the signed rows in the stream. 
The Frobenius norm estimator is then the two norm of the stored row. 
Since each of the classic AMS estimator for the $d$ columns of $\A$ succeeds with high probability, then by a union bound over the $d\ll n$ columns, the estimator is a $(1+\eps)$-approximation of the Frobenius norm of $\A$ with high probability. 
Moreover, since an entire row is stored, the algorithm requires an additional factor of $d$ space. 
This is less efficient than using a Frobenius norm estimator by hashing each entry of $\A$ to a separate sign and simply storing the sum of the scaled entries, but it is more flexible. 
In particular, we can apply linear transformations to the stored row to simulate right multiplication on $\A$. 

Let $\P\in\mathbb{R}^{d\times d}$ be a given post-processing matrix. 
To show that \algref{alg:ams:basic} provides a good approximation to $\norm{\A\P}_F$, it suffices to argue that running an AMS estimator for the $d$ columns of $\A$ and then multiplying by $\P$ afterwards is equivalent to running an AMS estimator for each of the $d$ columns of $\A\P$. 
Observe that running an AMS estimator for $\A\P$ simply requires multiplying each row of $\A\P$ by a random sign and adding the resulting signed rows. 
This is equivalent to multiplying each row of $\A$ by a random sign, adding the resulting signed rows of $\A$, and then multiplying by $\P$, which is exactly what \algref{alg:ams:basic} does. 
In other words, the AMS estimator is a linear transformation that maps from $\A$ to $\S\A$ for some sketching matrix $\S$, but seeing rows of $\A$ and then multiplying by $\P$ results in the same data structure as seeing the rows of $\A\P$, since $\S(\A\P)=(\S\A)\P$ by associativity. 
Finally, note that if we want to estimate $\norm{\A\P-\M}_F$ given a post-processing matrix $\M$, then we can compute $\S(\A\P-\M)=\S\A\P-\S\M$ for the AMS estimator, given $\S\A\P$ along with $\M$ and the sketching matrix $\S$. 

Each estimator stores a row with $d$ entries each using $\O{\log n}$ bits. 
The estimator is repeated $\O{\frac{1}{\eps^2}\log n}$ times to give a $(1+\eps)$-approximation and to obtain high probability guarantees. 
Thus, the algorithm requires $\O{\frac{d}{\eps^2}\log^2 n}$ bits of space in total. 
\end{proof}

We give in \algref{alg:countsketch:basic} the generalization of the CountSketch~\cite{CharikarCF04} algorithm that outputs all rows $i$ of $\A\P$ such that $\norm{\A_i\P}_2\ge\eps\norm{\A\P}_F$, where $\A$ arrives in a stream and post-processing query access to $\P$ is given after the stream ends. 
We call a row $i$ a \emph{heavy row} if $\norm{\A_i\P}_2\ge\eps\norm{\A\P}_F$. 
\begin{algorithm}[!htb]
\caption{Basic algorithm that outputs heavy rows of $\norm{\A\P}_F$, where $\P$ is a post-processing matrix}
\alglab{alg:countsketch:basic}
\begin{algorithmic}[1]
\Require{Matrix $\A\in\mathbb{R}^{n\times d}$, query access to matrix $\P\in\mathbb{R}^{d\times d}$ after the stream ends, constant parameter $\eps>0$.}
\Ensure{Slight perturbations of the rows $\A_i\P$ with $\norm{\A_i\P}_2\ge\eps\norm{\A\P}_F$.}
\State{$r\gets\Theta(\log n)$ with a sufficiently large constant.}
\State{$b\gets\Omega\left(\frac{1}{\eps^2}\right)$ with a sufficiently large constant.}
\State{Let $T$ be an $r\times b$ table of buckets, where each bucket stores an $\mathbb{R}^{1\times d}$ row, initialized to zeros.}
\State{Let $s_{i,j}\in\{-1,+1\}$ be $4$-wise independent for $i\in[n]$, $j\in[r]$.}
\State{Let $h_i:[n]\to[b]$ be $4$-wise independent for $i\in[r]$.}
\State{\textbf{Streaming Stage:}}
\For{each update $\Delta_t$ to entry $A_{i,j}$}
\For{each $k=1$ to $r$}
\State{Add $\Delta_t\cdot s_{i,k}$ to entry $j$ of the vector in bucket $h_k(i)$ of row $k$.}
\EndFor
\EndFor
\State{Let $\v_{k,\ell}$ be the vector in row $k$, bucket $\ell$ of $T$ for $k\in[r],\ell\in[b]$.}
\State{\textbf{Processing $\P$ Stage:}}
\For{$k\in[r],\ell\in[b]$}
\State{$\v_{k,\ell}\gets\v_{k,\ell}\P$}
\EndFor
\State{On query $i\in[n]$, report $\median_{k\in[r]}\norm{\v_{k,h_k(i)}}_2$.}
\end{algorithmic}
\end{algorithm}

We first show that if $\X=\A\P$ and the stream updates the entries of $\X$ rather than the entries of $\A$, then we can obtain a good approximation to the heavy rows. 
Equivalently, the statement reads that if $\P=\I$ is the identity matrix, then \algref{alg:countsketch:basic} finds the heavy rows of $\A\P$. 
We will ultimately show \algref{alg:countsketch:basic} finds the heavy rows of $\A\P$ for general $\P$ by using the same linear sketching argument as in the proof of \lemref{lem:ams}.

For a matrix $\X\in\mathbb{R}^{n\times d}$, recall that $\X_{tail(b)}$ denotes $\X$ with the $b$ rows of $\X$ with the largest norm set to zeros. 
The following lemma shows that the $\frac{2}{\eps^2}$ rows with the largest norm output by \algref{alg:countsketch:basic} forms a good estimate of $\X$, even with respect to the stronger Frobenius tail error.  
\begin{lemma}
\lemlab{lem:countsketch:tail}
For any matrix $\X\in\mathbb{R}^{n\times d}$, \algref{alg:countsketch:basic} outputs an estimate $\widehat{\X_i}$ for each row $\X_i$, which together form an estimate matrix $\widehat{\X}$. 
Then with high probability, for all $i\in[n]$, there exists a vector $\v_i$ such that $\widehat{\X_i}=\X_i+\v_i$ and $\norm{\v_i}_2\le\eps\norm{\X_{\taileps}}_F$. 
Consequently, $\left|\norm{\X_i}_2-\norm{\widehat{\X_i}}_2\right|\le\eps\norm{\X_{\taileps}}_F$ for all $i\in[n]$. 
Moreover, $\norm{\X_{\taileps}}_F\le\norm{\X-\tilde{\X}}_F\le 2\norm{\X_{\taileps}}_F$ with high probability, where $\tilde{\X}=\widehat{\X}-\widehat{\X}_{\taileps}$ denotes the top $\frac{2}{\eps^2}$ rows of $\widehat{\X}$ by norm. 
\end{lemma}
\begin{proof}
For a fixed $i\in[n]$ and row $k$ in the CountSketch table $T$, let $h_k(i)$ be the bucket to which $\X_i$ hashes.
Let $\mathcal{E}_1$ be the event that the $\frac{2}{\eps^2}$ rows with the highest norms excluding $\X_i$ are not hashed to $h_k(i)$. 
For $b=\Omega\left(\frac{1}{\eps^2}\right)$ with sufficiently large constant, $\mathcal{E}_1$ occurs with probability at least $\frac{9}{10}$. 
Let $\v_i$ be the sum of the vectors hashed to $h_k(i)$, excluding $\X_i$, so that the vector stored in bucket $h_k(i)$ is $\widehat{\X_i}=\X_i+\v_i$. 
Conditioned on $\mathcal{E}_1$, the expected squared norm of the noise in bucket $h_k(i)$ can be bounded by $\Ex{\norm{\v_i}_2^2}\leq\frac{\eps^2}{100}\norm{\X_{\taileps}}_F^2$ for sufficiently large $b$. 
Note that Jensen's inequality implies $\Ex{\norm{\v_i}_2}\le\frac{\eps}{10}\norm{\X_{\taileps}}_F$ and we also have $\Var(\norm{\v_i}_2)\le\frac{\eps^2}{100}\norm{\X_{\taileps}}_F^2$. 
Thus by Chebyshev's inequality, 
\[\PPr{\norm{\v_i}_2\ge\eps\norm{\X_{\taileps}}_F}\le\frac{1}{81},\]
conditioning on $\mathcal{E}_1$. 
Hence, 
\[\PPr{\norm{\v_i}_2\ge\eps\norm{\X_{\taileps}}_F}\le\frac{1}{81}+\frac{1}{10}.\]
By repeating for each of the $r=\Theta(\log n)$ rows and taking the median, we have from triangle inequality that $\left|\norm{\X_i}_2-\norm{\widehat{\X_i}}_2\right|\le\eps\norm{\X_{\taileps}}_F$ for all $i\in[n]$ with high probability, thus proving the first part of the claim.

For the second part of the claim, note that $\norm{\X_{\taileps}}_F\le\norm{\X-\tilde{\X}}_F$ trivially holds, since $\tilde{\X}$ is a matrix with at most $\frac{2}{\eps^2}$ nonzero rows, and $\X_{\taileps}$ has removed the $\frac{2}{\eps^2}$ rows of $\X$ with the largest mass from $\X$. 
Moreover, $\X-\tilde{\X}$ alters at most $\frac{2}{\eps^2}$ rows of $\X$, each by at most $\eps\norm{\X_{\taileps}}_F$. 
Thus,
\begin{align*}
\norm{\X-\tilde{\X}}_F\le\sqrt{\sum_{i=1}^{2/\eps^2}\left(\eps\norm{\X_{\taileps}}_F\right)^2}=\sqrt{2}\norm{\X_{\taileps}}_F.
\end{align*}
\end{proof}
Taking $b=\Theta\left(\frac{1}{\eps^2}\right)$ in \lemref{lem:countsketch:tail}, we have the following guarantees of $\countsketch$.
\begin{lemma}
\lemlab{lem:countsketch}
Given a constant $b>0$, there exists a one-pass streaming algorithm \countsketch{} that takes updates to entries of a matrix $\A\in\mathbb{R}^{n\times d}$, as well as query access to a post-processing matrix $\P\in\mathbb{R}^{d\times d}$ that arrives after the stream, and outputs all indices $i$ such that $\norm{\A_i\P}_2\ge\frac{1}{\sqrt{b}}\norm{\A\P}_F$. 
For each index $i$, $\countsketch$ also outputs a vector $\r$ such that $\r=\A_i\P+\v_i$ and $\norm{\v_i}_2\le\frac{1}{\sqrt{b}}\norm{(\A\P)_{tail(b)}}_F$. 
The algorithm uses $\O{db\log^2 n}$ bits of space and succeeds with high probability. 
\end{lemma}
\begin{proof}
Correctness follows from \lemref{lem:countsketch:tail} providing an accurate estimate of the norms of the heavy rows and \lemref{lem:ams} providing an accurate estimate of the Frobenius norm of $\A\P$. 
The space complexity results from using a table with $\Theta(\log n)$ rows and $b$ buckets in each row. 
Furthermore, each bucket consists of a vector of dimension $d$, whose entries are each represented using $\O{\log n}$ bits. 
Thus, the algorithm requires $\O{db\log^2 n}$ bits of space. 
\end{proof}

\subsection{$L_{2,2}$ Sampling Algorithm}
In this section, we give an algorithm for $L_{2,2}$ sampling that will ultimately be used to simulate adaptive sampling on turnstile streams.

\subsubsection{Algorithm Description}
Given subroutines that estimate $\norm{\A\P}_F$ and the heavy rows of $\A\P$, we implement our $L_{2,2}$ sampler in \algref{alg:l2:sampler}. 
Our algorithm first takes each row $\A_i$ of matrix $i$ and forms a row $\B_i=\frac{\A_i}{\sqrt{t_i}}$, where $t_i$ is a scaling factor drawn uniformly at random from $[0,1]$. 
Note that we have the following observation:
\begin{observation}
\obslab{obs:sample:prob}
For any value $\gamma>0$, $\PPr{\norm{\B_i\P}_2\ge\gamma\norm{\A\P}_F}=\frac{\norm{\A_i\P}^2_2}{\gamma^2\norm{\A\P}_F^2}$.
\end{observation}
\begin{proof}
Since $\B_i=\frac{\A_i}{\sqrt{t_i}}$ and $t_i$ is drawn uniformly at random from $[0,1]$, then we have
\begin{align*}
\PPr{\norm{\B_i\P}_2\ge\gamma\norm{\A\P}_F}&=\PPr{\frac{\norm{\A_i\P}^2_2}{t_i}\ge\gamma^2\norm{\A\P}_F^2}\\
&=\PPr{t_i\le\frac{\norm{\A_i\P}^2_2}{\gamma^2\norm{\A\P}_F^2}}=\frac{\norm{\A_i\P}^2_2}{\gamma^2\norm{\A\P}_F^2}.
\end{align*}
\end{proof}
Intuitively, \obsref{obs:sample:prob} claims that by setting $T\propto\norm{\A\P}_F$, we can identify a row of $\B\P$ whose norm exceeds $T$ to effectively $L_{2,2}$ sample a row of $\A\P$. 
For technical reasons, we set $T=\sqrt{\frac{C\log n}{\eps}}\norm{\A\P}_F$. 
Our algorithm then uses \countsketch{} to find heavy rows of $\B\P$ and \ams{} to give an estimate $\widehat{F}$ of $\norm{\A\P}_F$ to determine whether there exists a row of $\B\P$ whose norm exceeds $T$. 

Our algorithm also uses \countsketch{} and a separate instance of \ams{} to compute $\widehat{S}$, which estimates the error in the tail of $\B\P$ and also indicates how accurate \countsketch is. 
If $\widehat{S}$ is large, then our estimations for each row of $\B\P$ from \countsketch{} may be inaccurate, so our algorithm must abort. 
Otherwise, if $\widehat{S}$ is sufficiently small, then our estimations for each row of $\B\P$ is somewhat accurate. 
Thus if the row of $\B\P$ with the largest norm exceeds $\sqrt{\frac{C\log n}{\eps}}\widehat{F}$, which is our estimation for $T$, then we output that particular row rescaled by $\sqrt{t_i}$ to recover the (noisy) original row of $\A\P$. 
 
\begin{algorithm}[!htb]
\caption{Single $L_{2,2}$ Sampler}
\alglab{alg:l2:sampler}
\begin{algorithmic}[1]
\Require{Matrix $\A\in\mathbb{R}^{n\times d}$ that arrives as a stream, matrix $\P\in\mathbb{R}^{d\times d}$ that arrives after the stream, approximation parameter $\eps>0$. }
\Ensure{Noisy row $\r$ of $\A\P$ sampled roughly proportional to the squared row norms of $\A\P$.}
\State{\textbf{Pre-processing Stage:}}
\State{$b\gets\Omega\left(\frac{1}{\eps^2}\right)$, $r\gets\Theta(\log n)$ with sufficiently large constants}
\State{For $i\in[n]$, generate independent scaling factors $t_i\in[0,1]$ uniformly at random.}
\State{Let $\B$ be the matrix consisting of rows $\B_i=\frac{1}{\sqrt{t_i}}\A_i$.}
\State{Let $\ams_1$ and $\ams_2$ track the Frobenius norms of $\A\P$ and $\B\P$, respectively.}
\State{Let $\countsketch$ be an $r\times b$ table, where each entry is a vector in $\mathbb{R}^d$, to find the heavy hitters of $\B$.}
\State{\textbf{Streaming Stage:}}
\For{each row $\A_i$}
\Comment{Presented in row-arrival model but also works for turnstile streams}
\State{Update $\countsketch$ with $\B_i=\frac{1}{\sqrt{t_i}}\A_i$.}
\State{Update linear sketch $\ams_1$ with $\A_i$.}
\State{Update linear sketch $\ams_2$ with $\B_i=\frac{1}{\sqrt{t_i}}\A_i$.}
\EndFor
\State{\textbf{Processing $\P$ Stage:}}
\State{After the stream, obtain matrix $\P$.}
\State{Multiply each vector $\v$ in each entry of the $\countsketch$ table by $\P$: $\v\gets\v\P$.}
\State{Multiply each vector $\v$ in $\ams_1$ by $\P$: $\v\gets\v\P$.}
\State{Multiply each vector $\v$ in $\ams_2$ by $\P$: $\v\gets\v\P$.}
\State{\textbf{Extraction Stage:}}
\State{Use $\ams_1$ to compute $\widehat{F}$ with $\norm{\A\P}_F\le\widehat{F}\le 2\norm{\A\P}_F$.}
\State{Extract the $\frac{2}{\eps^2}$ (noisy) rows of $\B\P$ that are estimated by $\countsketch$ to have the largest norms.}
\State{Let $\M\in\mathbb{R}^{d\times d}$ be the matrix with $\frac{2}{\eps^2}$-nonzero rows consisting of these top (noisy) rows.}
\State{Use $\ams_2$ to compute $\widehat{S}$ with $\norm{\B\P-\M}_F\le\widehat{S}\le 2\norm{\B\P-\M}_F$.}
\State{Let $\r_i$ be the (noisy) row of $\A\P$ in $\countsketch$ with the largest norm.}
\State{Let $C>0$ be some large constant so that the probability of failure is $\O{\frac{1}{n^{C/2}}}$.}
\If{$\widehat{S}>\sqrt{\frac{C\log n}{\eps}}\widehat{F}$ or $\norm{\r_i}_2<\sqrt{\frac{C\log n}{\eps}}\widehat{F}$}
\State{\Return FAIL.}
\Else
\State{\Return $\r=\sqrt{t_i}\r_i$.}
\EndIf
\end{algorithmic}
\end{algorithm}

\subsubsection{Analysis}
Conditioning on only a single row $\B_i\P$ satisfying $\norm{\B_i\P}_2\ge T=\sqrt{\frac{C\log n}{\eps}}\norm{\A\P}_F$, we could immediately identify this row if we had access to all rows of $\B\P$, as well as $\norm{\A\P}_F$, but this requires too much space. 
Instead, we use \countsketch{} to find the heavy rows of $\B\P$ and compare their norms to an estimate of $T$. 
However, if the error caused by \countsketch is high due to the randomness of the data structure, then the estimations of the row norms may be inaccurate and so our algorithm should abort. 
Our algorithm uses an estimator $\widehat{S}$ to compute the tail of $\B\P$, which bounds the error caused by \countsketch. 
We first show that the event of our algorithm aborts because the tail estimator $\widehat{S}$ is too large has small probability and is independent of the index $i$ and the value of $t_i$. 
\begin{lemma}
\lemlab{lem:tail:failure}
For each $j\in[n]$ and value of $t_j$,
\[\PPr{\widehat{S}>\sqrt{\frac{C\log n}{\eps}}\widehat{F}\,\Big|\,t_j}=\O{\eps}+\frac{1}{\poly(n)}.\]
\end{lemma}
\begin{proof}
Let $\mathcal{E}_1$ be the event that:
\begin{enumerate}
\item
$\norm{\A\P}_F\le\widehat{F}\le 2\norm{\A\P}_F$
\item
$\norm{\B\P-\M}_F\le\widehat{S}\le 2\norm{\B\P-\M}_F$
\item
$\norm{(\B\P)_{\taileps}}_F\le\norm{\B\P-\M}_F\le 2\norm{(\B\P)_{\taileps}}_F$
\end{enumerate}
Fix an index $j\in[n]$ and let $t_j$ be any fixed value $t\in[0,1]$. 
Since the goal of the lemma is to show that the probability of the failure event is independent of our choice of $j$ and $t$, these variables will actually not appear in the remainder of the proof. 
 
Observe that $\mathcal{E}_1$ holds with high probability by \lemref{lem:countsketch} and \lemref{lem:ams}.
Conditioned on $\mathcal{E}_1$, it suffices to bound the probability that $4\norm{(\B\P)_{\taileps}}_F>\sqrt{\frac{C\log n}{\eps}}\norm{\A\P}_F$.

Let $U=\sqrt{\eps}\norm{\A\P}_F$ and define the indicator variable for whether row $\B_i\P$ is heavy. 
That is, we define $y_i=1$ if $\norm{\B_i\P}_2>U$ and $y_i=0$ otherwise. 
Define $z_i$ as the scaled indicator variable $z_i=\frac{1}{U^2}\norm{\B_i\P}_2^2(1-y_i)$ so that $z_i\in[0,1]$ represents a scaled contribution of the small rows. 
Define $Y=\sum_{i\neq j} y_i$ to be the total number of heavy rows, $Z=\sum_{i\neq j} z_i$ to be the total scaled contribution of the small rows, and $\W\in\mathbb{R}^{n\times d}$ to be the matrix of the heavy rows, i.e., $\W_i=\B_i\P$ if $y_i=1$ and $\W_i$ is the row of all zeros otherwise. 
Observe that $\W$ contains at most $Y+1$ nonzero rows and $U^2Z=\norm{\B\P-\W}_F^2$. 
Moreover, $\norm{(\B\P)_{\taileps}}_F\le U\sqrt{Z}$ unless $Y\ge\frac{2}{\eps^2}$. 
Hence if $\mathcal{E}_2$ denotes the event that $Y\ge\frac{2}{\eps^2}$ and $\mathcal{E}_3$ denotes the event that $Z\ge\frac{C\log n}{16U^2\eps}\norm{\A\P}_F^2$, then it suffices to bound the probability of the events $\mathcal{E}_2$ and $\mathcal{E}_3$ by $\O{\eps}$, since $\neg\mathcal{E}_2\wedge\neg\mathcal{E}_3$ implies $4\norm{(\B\P)_{\taileps}}_F\le\sqrt{\frac{C\log n}{\eps}}\norm{\A\P}_F$.
In other words, the probability of failure due to the tail estimator is small if the number of heavy rows is small ($\neg\mathcal{E}_2$) and the total contribution of the small rows is small ($\neg\mathcal{E}_3$).  

By \obsref{obs:sample:prob}, $\Ex{y_i}=\frac{\norm{\A_i\P}_2^2}{U^2}$ and so $\Ex{Y}\le\frac{1}{\eps}$ by linearity of expectation since $U=\sqrt{\eps}\norm{\A\P}_F$. 
Hence $\PPr{\mathcal{E}_2}=\O{\eps}$ by Markov's inequality for sufficiently small $\eps$.  

To bound $\PPr{\mathcal{E}_3}$, observe that $z_i>0$ only for $\norm{\B_i\P}_2\le U$ or equivalently, $t_i\ge\frac{\norm{\A_i\P}_2^2}{\eps\norm{\A\P}_F^2}$. 
For sufficiently small $\eps$, $\frac{\norm{\A_i\P}_2^2}{\eps\norm{\A\P}_F^2}\ge\frac{\norm{\A_i\P}_2^2}{\norm{\A\P}_F^2}$.
Thus,
\begin{align*}
\Ex{z_i}\le\int_{\norm{\A_i\P}_2^2/\norm{\A\P}_F^2}^1 z_i\,dt_i=\int_{\norm{\A_i\P}_2^2/\norm{\A\P}_F^2}^1\frac{1}{t_i}\frac{1}{U^2}\norm{\A_i\P}_2^2\,dt_i. 
\end{align*}
Let $\mathcal{E}_4$ be the event that $t_i\ge\frac{1}{n^{C/2}}$ for all $i\in[n]$, so that $\PPr{\mathcal{E}_4}\ge1-\frac{1}{n^{C/2-1}}$.
Conditioned on $\mathcal{E}_4$, we have
\begin{align*}
\Ex{z_i\,|\,\mathcal{E}_4}\le\frac{1}{1-\frac{1}{n^{C/2}}}\int_{\frac{1}{n^{C/2}}}^1\frac{1}{t_i}\frac{1}{U^2}\norm{\A_i\P}_2^2\,dt_i\le\frac{C\log n}{U^2}{\norm{\A_i\P}_2^2}.
\end{align*}
Hence, we have $\Ex{Z\,|\,\mathcal{E}_4}\le\frac{C\log n}{\eps}$ and so the probability that $Z>\frac{C\log n}{16U^2\eps}\norm{\A\P}_F^2=\frac{C\log n}{16\eps^2}$ for sufficiently small $\eps$ is bounded by $\O{\eps}$ by Markov's inequality. 
Since the events $\neg\mathcal{E}_1,\mathcal{E}_2,\mathcal{E}_3,\mathcal{E}_4$ each occur with probability at most $\O{\eps}+\frac{1}{\poly(n)}$, then the claim follows.
\end{proof}

\noindent
The probability of sampling each row $i\in[n]$ will still be slightly distorted due to the noise from \countsketch, since we do not have exact values for the norm of each row. 
Similarly, if multiple rows exceed the threshold, we will output the row with the largest norm, which also alters the sampling probability of each row. 
We now show that these events only slightly perturb the probability of sampling each index $i$ and moreover, the output row is a small noisy perturbation of the original row.  
\begin{lemma}
\lemlab{lem:sampling:prob}
Conditioned on a fixed value of $\widehat{F}$, the probability that \algref{alg:l2:sampler} outputs (noisy) row $i$ is $\left(1\pm\O{\eps}\right)\frac{\eps}{C\log n}\frac{\norm{\A_i\P}_2^2}{\widehat{F}^2}+\frac{1}{\poly(n)}$. 
\end{lemma}
\begin{proof}
We first define the following set of events.
\begin{itemize}
\item
Let $\mathcal{E}$ denote the event that $t_i<\frac{\eps}{C\log n}\frac{\norm{\A_i\P}_2^2}{\widehat{F}^2}$ so that the algorithm should ideally output index $i$ and observe that $\PPr{\mathcal{E}}=\frac{\eps}{C\log n}\frac{\norm{\A_i\P}_2^2}{\widehat{F}^2}$ since $t_i\in[0,1]$ is selected uniformly at random. 
\item
Let $\mathcal{E}_1$ denote the event that one of the data structures $\countsketch$, $\ams_1$, or $\ams_2$ fails. 
Note that $\mathcal{E}_1$ occurs with probability $\frac{1}{\poly(n)}$ by \lemref{lem:countsketch} and \lemref{lem:ams}. 
\item
Let $\mathcal{E}_2$ denote the event that $\widehat{S}>\sqrt{\frac{C\log n}{\eps}}\widehat{F}$. 
Conditioned on $\mathcal{E}$, the probability of $\mathcal{E}_2$ is $\O{\eps}$ by \lemref{lem:tail:failure}. 
\item
Let $\mathcal{E}_3$ denote the event that some other row $\B_j\P$ also either exceeds the threshold or is close enough to the threshold, thus possibly preventing $\B_i\P$ from being reported. 
Specifically, $\mathcal{E}_3$ can only occur if some other row $j$ satisfies $\norm{\B_j\P}_2\ge\sqrt{\frac{C\log n}{\eps}}\widehat{F}-\sqrt{C\eps\log n}\widehat{F}$. 
Since $\B_j\P=\frac{1}{\sqrt{t_j}}\A_j\P$ and $t_j$ is chosen uniformly at random from $[0,1]$, then row $j$ exceeds this threshold with probability at most $\O{\frac{\eps}{C\log n}\frac{\norm{\A_j\P}_2^2}{\widehat{F}^2}}$ by \obsref{obs:sample:prob}. 
Taking a union bound over all $n$ rows, the probability of $\mathcal{E}_3$ is $\O{\frac{\eps}{\log n}}$. 
\item
Let $\mathcal{E}_4$ denote the event that $\norm{\B_i\P}$ exceeds the threshold but is not reported due to noise in the CountSketch data structure, i.e., $\norm{\r_i}_2<\sqrt{\frac{C\log n}{\eps}}\widehat{F}$.  
We now analyze the probability of $\mathcal{E}_4$. 
Conditioning on $\neg\mathcal{E}_2$, we have $\widehat{S}\le\sqrt{\frac{C\log n}{\eps}}\widehat{F}$. 
Conditioning on $\neg\mathcal{E}_1$, then $\norm{\B\P-\M}_F\le\widehat{S}$. 
Thus by \lemref{lem:countsketch:tail} (or \lemref{lem:countsketch}), 
\[\left|\norm{\B_i\P}_2-\norm{\widehat{\B_i\P}}_2\right|\le\eps\norm{(\B\P)_{\taileps}}_F\le\eps\norm{\B\P-\M}_F\le\eps\widehat{S}\le\sqrt{C\eps\log n}\widehat{F}.\]
Hence, $\mathcal{E}_4$ can only occur for 
\[\sqrt{\frac{C\log n}{\eps}}\widehat{F}\le\norm{\B_i\P}_2\le\sqrt{\frac{C\log n}{\eps}}\widehat{F}+\sqrt{C\eps\log n}\widehat{F},\]
which occurs with probability at most $\O{\frac{\eps^2}{C\log n}\frac{\norm{\A_i\P}_2^2}{\widehat{F}^2}}$ over the choice of $t_i$. 
\end{itemize}

Conditioning on $\mathcal{E}$, the sampler should return $\A_i\P$ but may fail to do so because of any of the events $\mathcal{E}_1$, $\mathcal{E}_2$, $\mathcal{E}_3$, or $\mathcal{E}_4$.
Putting things together, $\mathcal{E}_1$ occurs with probability $\frac{1}{\poly(n)}$. 
Conditioning on $\mathcal{E}$, $\mathcal{E}_2$ and $\mathcal{E}_3$ each occur with probability $\O{\eps}$, so the probability of $\mathcal{E}$ and at least one of $\mathcal{E}_2$ or $\mathcal{E}_3$ occurring is $\O{\frac{\eps^2}{C\log n}\frac{\norm{\A_i\P}_2^2}{\widehat{F}^2}}$, which is also the probability of $\mathcal{E}_4$. 
Thus the sampling probability of each $\A_i\P$ follows. 

Finally, we emphasize that for the index $i$ selected, it holds by \lemref{lem:countsketch:tail} that
\[\left|\norm{\B_i\P}_2-\norm{\widehat{\B_i\P}}_2\right|\le\sqrt{C\eps\log n}\widehat{F}\]
and $\norm{\widehat{\B_i\P}}_2\ge\sqrt{\frac{C\log n}{\eps}}\widehat{F}$. 
Thus, $\norm{\widehat{\B_i\P}}_2$ is a $(1+\eps)$ approximation to $\norm{\B_i\P}_2$ and similarly, $\sqrt{t_i}\norm{\widehat{\B_i\P}}_2$ has norm $(1+\eps)$ within that of $\norm{\A_i\P}_2$. 
\end{proof}	

Since each row is sampled with roughly the desired probability, we now analyze the space complexity of the resulting $L_{2,2}$ sampler. 
\begin{theorem}
\thmlab{thm:l2:sampling}
Given $\eps>0$, there exists a one-pass streaming algorithm that takes rows of a matrix $\A\in\mathbb{R}^{n\times d}$ as a turnstile stream and a matrix $\P\in\mathbb{R}^{d\times d}$ after the stream, and outputs (noisy) row $i$ of $\A\P$ with probability $\left(1\pm\O{\eps}\right)\frac{\norm{\A_i\P}_2^2}{\norm{\A\P}_F^2}+\frac{1}{\poly(n)}$. 
The algorithm uses $\O{\frac{d}{\eps^3}\log^3 n\log\frac{1}{\delta}}$ space to succeed with probability $1-\delta$. 
\end{theorem}
\begin{proof}
By \lemref{lem:sampling:prob} and the fact that $\norm{\A\P}_F\le\widehat{F}\le 2\norm{\A\P}_F$ with high probability by \lemref{lem:ams}, each row $\A_i\P$ is output with probability $\left(1+\eps\right)\frac{\norm{\A_i\P}_2^2}{\norm{\A\P}_F^2}+\frac{1}{\poly(n)}$, conditioned on the sampler outputting \emph{some} index rather than aborting. 
Recall that the sampler outputs index if the tail estimator $\widehat{S}$ is small and the estimated norm of some row exceeds the threshold. 
\lemref{lem:tail:failure} shows that the tail estimator is small only with probability $\O{\eps}$ while a straightforward computation shows that the probability that the estimated norm of some row exceeding the threshold is $\Theta\left(\frac{\eps}{\log n}\right)$. 
Thus the sampler can be repeated $\O{\frac{1}{\eps}\log n\log\frac{1}{\delta}}$ times to obtain probability of success at least $1-\delta$. 
By \lemref{lem:ams}, each instance of \ams{} uses $\O{\frac{d}{\eps^2}{\log^2 n}}$ space. 
Moreover, each sampler uses a $\O{\frac{1}{\eps^2}}\times\O{\log n}$ table, and each entry in the table is a vector of $d$ integers, the total space complexity is $\O{\frac{d}{\eps^3}\log^3 n\log\frac{1}{\delta}}$. 
\end{proof}

\paragraph{Generation of Uniform Random Variables.} 
First observe that with high probability, each of the uniform random variables $t_i$ are least $\frac{1}{\poly(ndmM)}=\frac{1}{\poly(n)}$ precision, where $m=\poly(n)$ is the length of the stream and $M=\poly(n)$ is the largest change an update can induce in the matrix $\A\in\mathbb{R}^{n\times d}$. 
Then we claim it suffices to generate the uniform random variables $t_i$ up to $\O{\log n}$ bits of precision, for a sufficiently large constant. 
Indeed note that the truncation perturbs the values of the uniform random variables by an additive $\frac{1}{\poly(n)}$ value, which induces a $\O{\frac{1}{\poly(n)}}$ additive error for each CountSketch bucket. 
Therefore, we can incorporate the additive error induced by truncating the uniform random variables at $\O{\log n}$ bits of precision into the $\O{\frac{1}{\poly(n)}}$ additive error of the $L_{2,2}$ sampler. 

We say a family $\mathcal{H}=\{h:[n]\to[m]\}$ is \emph{$\delta$-approximate} if for all $r\le n$, possible inputs $i_1,\ldots,i_r\in[n]$ and possible outputs $o_1,\ldots,o_1\in[m]$,
\[\PPPr{h\in H}{h(i_1)=o_1\wedge\ldots\wedge h(i_r)=o_r}=\frac{1}{m^r}+\delta.\]
If $\delta=0$ for all $r\le k$, the function is \emph{$k$-wise independent}. 

We observe that $\O{1}$-wise independent random variables $t_i$ would suffice for justifying the low-probability failure events in \lemref{lem:tail:failure} through Chebyshev's inequality. 
Recall that $k$-wise independent random variables can be generated from a polynomial of degree $k$ over a field of size $\O{\poly(n)}$~\cite{WegmanC81}, which can be stored using $\O{k\log n}$ space~\cite{WegmanC81}.  
Thus for the purposes of our $L_{2,2}$ sampler, using $\O{1}$-wise independent random variables $t_i$ instead of fully independent random variables gives the exact guarantees as \thmref{thm:l2:sampling}.

However, \secref{sec:noisy:adaptive} requires Chernoff bounds to analyze the size of specific sets, which will not na\"{i}vely work with $\O{1}$-wise independent random variables. 
Instead, we can apply the limited independence Chernoff-Hoeffding bounds in~\cite{SchmidtSS95} using $\O{\log n}$-wise independent random variables, thus limiting the probability of the failure events by $\O{\frac{1}{\poly(n)}}$. 
Moreover, $\delta$-approximate $k$-wise hash functions can be generated using $\O{k+\log n+\log\frac{1}{\delta}}$ bits, e.g., by composing the generators of \cite{AlonBI86} and \cite{NaorN93}. 
Thus for $\delta=\frac{1}{\poly(n)}$, the error can again be absorbed into the $\O{\frac{1}{\poly(n)}}$ additive error of the $L_{2,2}$ sampler, while the family of hash functions requires $\O{\log n}$ bits to store. 
\section{Noisy Adaptive Squared Distance Sampling}
\seclab{sec:noisy:adaptive}
Given a matrix $\A\in\mathbb{R}^{n\times d}$ that arrives in a data stream, either turnstile or row-arrival, we want to simulate $k$ rounds of adaptive sampling. 
That is, in the first round we want to sample some row $\r_1$ of $\A$, such that each row $\A_i$ is selected with probability proportional to its squared row norm $\norm{\A_i}_2^2$. 
Once rows $\r_1,\ldots,\r_{j-1}$ are selected, then the $j\th$ round of adaptive sampling samples each row $\A_i$ with probability proportional to the squared row norm of the orthogonal component to $\R_{j-1}$, $\norm{\A_i(\I-\R_{j-1}^\dagger\R_{j-1})}_2^2$, where for each $j\leq k$, $\R_j=\r_1\circ\ldots\circ\r_j$. 

Observe that if only a single round of adaptive sampling were required, the problem would reduce to $L_{2,2}$ sampling, which we can perform in a stream through \algref{alg:l2:sampler}. 
In fact, the post-processing stage of \algref{alg:l2:sampler} would not be necessary since the post-processing matrix would be $\P=\I$, which is the identity matrix. 
Moreover, the sketch of $\A$ of \algref{alg:l2:sampler} is oblivious to the choice of the post-processing matrix $\P$, so we would like to repeat this $k$ times by creating $k$ separate instances of the $L_{2,2}$ sampler of \algref{alg:l2:sampler} and for the $j\th$ instance, multiply by the post-processing matrix $\P_j=\I-\R_{j-1}^\dagger\R_{j-1}$. 
Unfortunately, if $f(j)$ is the index of the row of $\A\P_j$ that is selected in the $j\th$ round, the row $\r_j$ that the $L_{2,2}$ sampler outputs is not $\A_{f(j)}\P_j$ but rather a noisy perturbation of it, which means in future rounds we are not sampling with respect to a subspace containing $\A_{f(j)}\P_j$ but rather a subspace containing $\r_j$. 
This is particularly a problem if $\r_j$ is parallel to another row $\A_i$ that is not contained in the subspace of $\A_{f(j)}\P_j$, then in future rounds the probability of sampling $\A_i$ is zero, when it should in fact be nonzero. 
Although the above example shows that the noisy perturbation does not preserve relative sampling probabilities for each row, we show that the perturbations give a good additive approximation to the sampling probabilities. 
That is, we bound the total variation distance between sampling with respect to the true rows of $\A$ and sampling with respect to the noisy rows of $\A$. 
We give our algorithm in full in \algref{alg:noisy:adaptive}. 

\begin{algorithm}[!htb]
\caption{Noisy Adaptive Sampler}
\alglab{alg:noisy:adaptive}
\begin{algorithmic}[1]
\Require{Matrix $\A\in\mathbb{R}^{n\times d}$ that arrives as a stream $\A_1,\ldots,\A_n\in\mathbb{R}^d$, parameter $k$ for number of rows to be sampled, constant parameter $\eps>0$.}
\Ensure{$k$ Noisy and projected rows of $\A$.}
\State{Create instances $\alg_1,\ldots,\alg_k$ of the $L_{2,2}$ sampler of \algref{alg:l2:sampler} where the number of buckets $b=\Theta\left(\frac{\log^2 n}{\eps^2}\right)$ is sufficiently large.}
\State{Let $\M$ be empty $0\times d$ matrix.}
\State{\textbf{Streaming Stage:}}
\For{each row $\A_i$}
\State{Update each sketch $\alg_1,\ldots,\alg_k$}
\EndFor
\State{\textbf{Post-processing Stage:}}
\For{$j=1$ to $j=k$}
\State{Post-processing matrix $\P\gets\I-\M^\dagger\M$.}
\State{Update $\alg_j$ with post-processing matrix $\P$.}
\State{Let $\r_j$ be the noisy row output by $\alg_j$.}
\State{Append $\r_j$ to $\M$: $\M\gets\M\circ\r_j$.}
\EndFor
\State{\Return $\M$.}
\end{algorithmic}
\end{algorithm}

For the purpose of the analysis, we first show that if the $L_{2,2}$ sampler outputs row $\r_1$ that is a noisy perturbation of row $\A_{f(1)}\P$, then not only can we bound $\norm{\A_{f(1)}\P-\r_1}_2$ as in \lemref{lem:countsketch:tail}, but also we can bound the norm of the component of $\r_1$ orthogonal to $\A_{f(1)}\P$. 
This is significant because future rounds of sampling will focus on the norms of the orthogonal components for the sampling probabilities. 
\begin{lemma}
\lemlab{lem:orthogonal:noise}
Given a matrix $\A\in\mathbb{R}^{n\times d}$ and a matrix $\P\in\mathbb{R}^{d\times d}$, as defined in Line 8 and round $i\le k$, of \algref{alg:noisy:adaptive}, suppose index $j\in[n]$ is sampled (in round $i$). 
Then with high probability, the sampled (noisy) row $\r_i$ satisfies $\r_i=\A_j\P+\v_e$ with 
\[\norm{\v_e\Q}_2\le\frac{\eps\sqrt{\eps}\norm{\A\P\Q}_F}{\sqrt{C\log n}\norm{\A\P}_F}\norm{\A_j\P}_2,\]
for any projection matrix $\Q\in\mathbb{R}^{d\times d}$. 
Hence, $\v_e$ is orthogonal to each noisy row $\r_y$, where $y\in[i-1]$. 
\end{lemma}
\begin{proof}
Let $\Q\in\mathbb{R}^{d\times d}$ be a projection matrix. 
For each $x\in[n]$, let $\B_x=\frac{\A_x}{\sqrt{t_x}}$ be the rescaled row of $\A_x$. 
Let $\E$ be the noise in the bucket corresponding to the selected row $j$, so that the output vector is $\A_j+\sqrt{t_j}\E$ and the noise is $\v_e=\sqrt{t_j}\E$. 
Note that $\E$ is a linear combination of rows of $\A\P$ and thus $\E$ is orthogonal to all previous noisy rows $\r_y$ with $y\in[i-1]$. 
Let $\B\in\mathbb{R}^{n\times d}$ be the rescaled matrix of $\A$ so that row $x$ of $\B$ is $\B_x$ for $x\in[n]$. 
Recall that $t_x\in[0,1]$ is selected uniformly at random for each $x\in[n]$, so that for each integer $c\ge0$, 
\[\PPr{\frac{\norm{\A_x\P\Q}_2^2}{t_x}\ge\frac{\norm{\A\P\Q}_F^2}{2^c}}\le\frac{2^c\norm{\A_x\P\Q}_2^2}{\norm{\A\P\Q}_F^2},\]
e.g., by \obsref{obs:sample:prob}.
Since $\B_x=\frac{\A_x}{\sqrt{t_x}}$, then by linearity of expectation over $x\in[n]$, we can bound the expected size of each of the disjoint level sets $S_c:=\left\{x\in[n]\,:\,\frac{\norm{\A\P\Q}_F^2}{2^{c-1}}>\norm{\B_x\P\Q}_2^2\ge\frac{\norm{\A\P\Q}_F^2}{2^c}\right\}$ by $\Ex{|S_c|}\le\min(2^c,n)$ for each $c$. 
From the independence of the scaling factors $t_x$, then the Chernoff bound implies that
\[\PPr{\left|S_c\right|\le\min(2^{c+\beta}\log n,n)}\ge1-\frac{1}{\poly(n)},\]
for sufficiently large constant $\beta$. 
Thus the Frobenius norm of $\B\P\Q$ can be roughly upper bounded by the Frobenius norm of $\A\P\Q$ by a union bound over level sets $S_c$ for $c\le\log n$ and upper bounding the norms of each of the rows in level sets $S_c$ with $c>\log n$ by $\frac{\norm{\A\P\Q}^2_F}{n}$ and thus the total mass of the level sets $S_c$ with $c>\log n$ by $\norm{\A\P\Q}^2_F$. 
That is,  
\[\PPr{\norm{\B\P\Q}_F^2\ge 2^{\beta}\log^2 n\norm{\A\P\Q}_F^2}\ge1-\frac{1}{\poly(n)}.\]
Hence the total mass $\norm{\B\P\Q}_F^2$ distributed across the CountSketch table is $\O{\log^2 n\norm{\A\P\Q}_F^2}$ with high probability.

By \lemref{lem:countsketch} and hashing rows of $\B\P\Q$ to a CountSketch table with $b=\Theta\left(\frac{\log^2 n}{\eps^2}\right)$ buckets with sufficiently large constant, the bucket corresponding to $\A_j$ has mass at most $\eps\norm{\A\P\Q}_F$ when projected onto $\Q$. 
That is, $\norm{\E\Q}_2\le\eps\norm{\A\P\Q}_F$. 
This can also be seen from the fact that \countsketch{} is a linear sketch and considering the error in a certain subspace $\Q$ is equivalent to right multiplication by $\Q$. 

Since row $j$ was selected, it must have been true that $\norm{\B_j\P}_2\ge\sqrt{\frac{C\log n}{\eps}}\norm{\A\P}_F$. 
Because $\B_j=\frac{\A_j}{\sqrt{t_j}}$, then $\sqrt{t_j}\le\frac{\sqrt{\eps}\norm{\A_j\P}_2}{\sqrt{C\log n}\norm{\A\P}_F}$.  
Therefore,
\[\norm{\sqrt{t_j}\E\Q}_2\le\frac{\eps\sqrt{\eps}\norm{\A\P\Q}_F}{\sqrt{C\log n}\norm{\A\P}_F}\norm{\A_j\P}_2.\]
In particular, since $\v_e\Q=\sqrt{t_j}\E\Q$, the above expression also bounds the Euclidean norm of $\v_e\Q$. 
Intuitively, not only is the overall noise of $\A\P$ well-partitioned among the buckets of CountSketch, but the noise in each direction $\A\P\Q$ must also be well-partitioned among the buckets of CountSketch with high probability. 
\end{proof}

Recall that our sampler only returns noisy rows $\r_i$, rather than the true rows of $\A\P_i$, where $\P_i$ is any post-processing matrix.  
This is problematic for multiple rounds of sampling, since $\r_i$ is then used to form the next post-processing matrix $\P_{i+1}$, rather than the true row. 
We next show that the total variation distance has not been drastically altered by sampling with respect to the noisy rows rather than the true rows. 

The main idea is that because the noise in each direction is proportional to the total mass in the subspace by \lemref{lem:orthogonal:noise}, we can bound the total perturbation in the squared norms of each row of $\A\P_i$. 
We first argue that if we obtain a noisy row in the first sampling iteration but then we obtain the true rows in the subsequent iterations, then the total variation distance between the resulting probability distribution of sampling each row is close to the ideal probability distribution of sampling each row if we had obtained the true rows over all iterations. 
It then follows from triangle inequality that the actual sampling distribution induced by obtaining noisy rows in each round is close to the ideal sampling distribution if we had obtained the true rows. 

To bound the perturbation in the sampling probability of each row, we require a change of basis matrix from a representation of vectors in terms of the true rows of $\A$ to a representation of vectors in terms of the noisy rows of $\A$. 
This change of basis matrix crucially must be close to the identity matrix, in order to preserve the perturbation in the squared norms. 

\begin{lemma}
\lemlab{lem:adaptive:tvd}
Let $f(1)$ be the index of a noisy row $\r_1$ sampled in the first iteration of \algref{alg:noisy:adaptive}. 
Let $\mathcal{P}_1$ be a process that projects away from $\A_{f(1)}$ and iteratively selects $k-1$ additional rows of $\A$ through adaptive sampling (with $p=2$).  
Let $\mathcal{P}_2$ be a process that projects away from $\r_1$ and iteratively selects $k-1$ additional rows of $\A$ through adaptive sampling (with $p=2$). 
Then for $\eps<\frac{1}{d}$, the total variation distance between the distributions of $k$ indices output by $\mathcal{P}_1$ and $\mathcal{P}_2$ is $\O{k\eps}$. 
\end{lemma}
\begin{proof}
Suppose $\mathcal{P}_1$ sequentially samples rows $\A_{f(1)},\ldots,\A_{f(k)}$. 
For each $t\in[k]$, let $\T_t=\A_{f(1)}\circ\ldots\circ\A_{f(t)}$ and $\Z_t=\I-\T_t^\dagger\T_t$ and $\R_t=\r_1\circ\A_{f(2)}\ldots\circ\A_{f(t)}$ and $\Y_t=\I-\R_t^\dagger\R_t$. 
We assume for the sake of presentation that $\A$ is a full-rank matrix, i.e. $\rank(\A)=d$ and prove the claim by induction. 

\paragraph{Base case.} 
For $t=2$, we first must show that the sampling distributions induced by $\r_1$ and $\A_{f(1)}$ are similar. 
Let $U=\{\u_1,\ldots,\u_d\}$ be the orthonormal basis for the row span of $\A$ so that $\u_1=\frac{\A_{f(1)}}{\norm{\A_{f(1)}}_2}$ points in the direction of $\A_{f(1)}$. 
Similarly, let $W=\{\w_1,\ldots,\w_d\}$ be an orthonormal basis for the row span of $\A$ obtained by applying the Gram-Schmidt process to the set $\{\w_1,\u_2,\ldots,\u_d\}$, where $\w_1=\frac{\r_1}{\norm{\r_1}_2}$. 
We argue that the change of basis matrix $\B$ from $U$ to $W$ must be close to the identity matrix. 

First, observe that from \lemref{lem:orthogonal:noise}, we have $\r_1=\norm{\A_{f(1)}}_2\left(\u_1+\sum_{i=1}^d(\pm\tau_i)\u_i\right)$, where $\tau_i\le\frac{\eps\sqrt{\eps}}{\sqrt{C\log n}}\frac{\norm{\A\P_i}_F}{\norm{\A}_F}$ with high probability and $\P_i=\u_i^\dagger\u_i$ is the projection matrix onto $\u_i$. 
Thus by setting $\tau^2=\sum_{i=1}^d\tau_i^2$, we can write the first row $\b_1$ of $\B$ so that the first entry is at least $1-\tau$ and entry $i$ is at most $\tau_i$ in magnitude. 
\begin{claim}
\claimlab{claim:cob}
The first entry in row $j>1$ is at most $2\tau_j$ in magnitude, entry $j$ in row $j$ is at least $1-3\tau_j$ in magnitude, and entry $i$ in row $j$ is at most $5\tau_i\tau_j$ in magnitude for $i<j$ with $i\neq 1$ and at most $2\tau_i\tau_j$ in magnitude for $i>j$. 
\end{claim}
\begin{proof}
We first consider the base case $j=2$ and determine $\b_2$ through the Gram-Schmidt process.  
For the elementary vector $\e_2\in\mathbb{R}^d$, note that $|\langle\e_2,\b_1\rangle|\norm{\b_1}_2\le\tau_2$. 
Thus we have that entry $i$ in $\b_2$ for $i\neq 2$ is at most $\frac{|\langle\e_i,\b_1\rangle|\tau_2}{1-\tau_2}$ in magnitude. 
Specifically for $i=1$, the first entry in $\b_2$ is bounded by $2\tau_2$ in magnitude, while for $i>2$, entry $i$ in $\b_2$ is bounded by $2\tau_2\tau_i$ in magnitude.  
It follows that the second entry in $\b_2$ is at least $1-2\tau_2$. 

We use similar reasoning to bound the entries in row $j$ of $\B$. 
Note that for sufficiently small $\eps<\frac{1}{d}$, we have
\[\sum_{i=1}^{j-1}|\langle\e_j,\b_i\rangle|\norm{\b_i}_2\le\tau_j+\sum_{i=2}^{j-1}5\tau_j\tau_i \le2\tau_j.\] 
Thus from the Gram-Schmidt process, entry $i$ in $\b_j$ for $i\neq j$ is at most $\frac{1}{1-2\tau_j}\left|\sum_{\ell=1}^{j-1}\langle\e_i,\b_{\ell}\rangle\cdot\langle\e_j,\b_{\ell}\rangle\right|$, which is at most 
\[\frac{1}{1-2\tau_j}\left(\tau_j+\sum_{\ell=2}^{j-1}4\tau_j\tau_{\ell}^2\right)\le2\tau_j\] 
in magnitude for $i=1$ for sufficiently small $\eps<\frac{1}{d}$ and at most 
\[\frac{1}{1-2\tau_j}\left(\tau_i\tau_j+2\tau_i\tau_j +\sum_{\ell=2,\ell\neq i}^{j-1}10\tau_j\tau_i\tau_{\ell}^2\right)\le5\tau_i\tau_j\] 
in magnitude for $i<j$ with $i\neq 1$ and at most
\[\frac{1}{1-2\tau_j}\left(\tau_i\tau_j+\sum_{\ell=2}^{j-1}4\tau_j\tau_i\tau_{\ell}^2\right)\le2\tau_i\tau_j\]
in magnitude for $i>j$. 
Thus it follows that entry $j$ in row $j$ of $\B$ is at least $1-3\tau_j$, which completes the induction.
\end{proof}
Therefore by \claimref{claim:cob}, we have
\begin{align}
\B=\begin{bmatrix}
1-\O{\tau} & \pm\O{\tau_2} & \pm\O{\tau_3} & \pm\O{\tau_4} & \ldots & \pm\O{\tau_d}\\
\pm\O{\tau_2} & 1-\O{\tau_2} & \pm\O{\tau_2\tau_3} & \pm\O{\tau_2\tau_4} & \ldots & \pm\O{\tau_2\tau_d}\\
\pm\O{\tau_3} & \pm\O{\tau_3\tau_2} & 1-\O{\tau_3} & \pm\O{\tau_3\tau_4} & \ldots & \pm\O{\tau_3\tau_d}\\
\pm\O{\tau_4} & \pm\O{\tau_4\tau_2} & \pm\O{\tau_4\tau_3} & 1-\O{\tau_4} & \ldots & \pm\O{\tau_4\tau_d}\\
\vdots & \vdots & \vdots & \vdots & \ddots & \vdots \\
\pm\O{\tau_d} & \pm\O{\tau_d\tau_2} & \pm\O{\tau_d\tau_3} & \pm\O{\tau_d\tau_4} & \ldots & 1-\O{\tau_d}
\end{bmatrix}.
\label{eqn:cob:matrix}\tag{$\star$}
\end{align}
for $\eps<\frac{1}{d}$, where the $\O{\cdot}$ notation in (\ref{eqn:cob:matrix}) hides a constant that is at most $5$. 

We can write each row $\A_s$ in terms of basis $U$ as $\A_s=\sum_{i=1}^d\lambda_{s,i}\u_i$ and in terms of basis $W$ as $\A_s=\sum_{i=1}^d\zeta_{s,i}\w_i$. 
Then since we project away from $\A_{f(1)}$, we should have sampled $\A_s$ with probability $\frac{\norm{\A_s\Z_{t-1}}_2^2}{\norm{\A\Z_{t-1}}_F^2}=\frac{\sum_{i=2}^d\lambda_{s,i}^2}{\sum_{j=1}^n\sum_{i=2}^d\lambda_{j,i}^2}$ in the second round but instead we sample it with probability $\frac{\norm{\A_s\Y_{t-1}}_2^2}{\norm{\A\Y_{t-1}}_F^2}=\frac{\sum_{i=2}^d\zeta_{s,i}^2}{\sum_{j=1}^n\sum_{i=2}^d\zeta_{j,i}^2}$. 
From the change of basis matrix $\B$, we can also write 
\[
\zeta_{s,i}=(1-\O{\tau_i})\lambda_{s,i}\pm\O{\tau_i}\lambda_{s,1}\pm\sum_{j\notin \{i,1\}}\O{\tau_i\tau_j}\lambda_{s,j},
\]
for $i\ge 2$. 
Therefore we can bound the difference
\begin{align*}
|\zeta_{s,i}^2-\lambda_{s,i}^2|&\le25\Big(\tau_i\lambda_{s,i}^2+\tau_i^2\lambda_{s,1}^2+\tau_i\lambda_{s,1}\lambda_{s,i}+\sum_{j,\ell\neq\{i,1\}}\tau_i^2\tau_j\tau_\ell\lambda_{s,j}\lambda_{s,\ell}\\
&+\sum_{j\neq\{i,1\}}\tau_i\tau_j\lambda_{s,i}\lambda_{s,j}+\sum_{j\neq\{i,1\}}\tau_i^2\tau_j\lambda_{s,1}\lambda_{s,j}\Big)\\
&\le 25\left(\tau_i\lambda_{s,i}^2+d\tau_i^2\lambda_{s,1}^2+\tau_i\lambda_{s,1}\lambda_{s,i}+4\sum_{j=2}^d d\tau_j^2\lambda_{s,j}^2\right),
\end{align*}
where the last inequality follows from AM-GM and that all values of $\tau_i,\tau_j\leq \eps^{3/2}$ and thus $\tau<1$ for $\eps<\frac{1}{d}$. 
We also have $\tau_i\lambda_{s,1}\lambda_{s,i}\le\eps\lambda_{s,i}^2+\frac{\tau_i^2}{\eps}\lambda_{s,1}^2$. 
Thus, 
\begin{align*}
\left|\sum_{i=2}^d\zeta_{s,i}^2-\sum_{i=2}^d\lambda_{s,i}^2\right|&\le25\sum_{i=2}^d\left[2\left(\eps+4d^2\tau_i^2\right)\lambda_{s,i}^2+\frac{2\tau_i^2}{\eps}\lambda_{s,1}^2\right]\\
&\le 25\sum_{i=2}^d\left(6\eps\lambda_{s,i}^2+\frac{2\tau_i^2}{\eps}\lambda_{s,1}^2\right),
\end{align*}
since $\tau_i^2<\eps^3$, and $\eps<\frac{1}{d}$. 
Moreover, $\tau_i^2\le\frac{\eps^3}{C\log n}\frac{\sum_{a=1}^n\lambda_{a,i}^2}{\sum_{a=1}^n\sum_{b=1}^d\lambda_{a,b}^2}$, thus we have $\sum_{s=1}^n\frac{\tau_i^2}{\eps}\lambda_{s,1}^2 \leq \frac{\eps^2}{C\log n}\sum_{s=1}^n \lambda_{s,i}^2$. 
Therefore, we have $\left|\sum_{j=1}^n\sum_{i=2}^d\zeta_{j,i}^2-\sum_{j=1}^n\sum_{i=2}^d\lambda_{j,i}^2\right|\le200\eps\sum_{j=1}^n\sum_{i=2}^d\lambda_{j,i}^2$.
In other words, $\norm{\A\Y_1}_F^2$ is within a $(1+200\eps)$ factor of $\norm{\A\Z_1}_F^2$. 
Moreover, $\frac{\norm{\A_s\Y_1}_2^2}{\norm{\A\Y_1}_F^2}$ and $\frac{\norm{\A_s\Z_1}_2^2}{\norm{\A\Z_1}_F^2}$ are probability distributions that each sum to $1$ across all $s$. 
Thus we have the distortion in the sampling probability of $\A_s$ is
\[\left|\frac{\sum_{i=2}^d\lambda_{s,i}^2}{\sum_{j=1}^n\sum_{i=2}^d\lambda_{j,i}^2}-\frac{\sum_{i=2}^d\zeta_{s,i}^2}{\sum_{j=1}^n\sum_{i=2}^d\zeta_{j,i}^2}\right|\le\frac{2(1+200\eps)25}{\sum_{j=1}^n\sum_{i=2}^d\lambda_{j,i}^2}\sum_{i=2}^d\left(6\eps\lambda_{s,i}^2+2\frac{\tau_i^2}{\eps}\lambda_{s,1}^2\right).\]
Taking the sum over all rows $\A_s$ and noting $\tau_i^2\le\frac{\eps^3}{C\log n}\frac{\sum_{a=1}^n\lambda_{a,i}^2}{\sum_{a=1}^n\sum_{b=1}^d\lambda_{a,b}^2}$, we have that 
\[\sum_{i=1}^n\left|\frac{\norm{\A_i\Y_1}_2^2}{\norm{\A\Y_1}_F^2}-\frac{\norm{\A_i\Z_1}_2^2}{\norm{\A\Z_1}_F^2}\right|\le799\eps,\]
for sufficiently small $\eps$. 
Thus including the $\frac{1}{\poly(n)}$ event of failure from \lemref{lem:orthogonal:noise}, the total variation distance is at most $800\eps$, which completes our base case.

\paragraph{Inductive step.} 
Suppose that the total variation distance between the distributions of the first $t-1$ indices sampled by $\mathcal{P}_1$ and $\mathcal{P}_2$ is at most $800(t-1)\eps$. 
We consider the difference in the probability distribution induced by linearly independent vectors $\A_{f(1)},\ldots,\A_{f(t-1)}$ and the probability distribution induced by linearly independent vectors $\r_1,\A_{f(2)},\ldots,\A_{f(t-1)}$. 
We can define $U=\{\u_1,\ldots,\u_d\}$ to be an orthonormal basis for the row span of $\A$ such that $\{\u_1,\ldots,\u_s\}$ is a basis for the row span of $\{\A_{f(1)},\ldots,\A_{f(s)}\}$ for each $2\le s\le t-1$. 
Similarly, let $W=\{\w_1,\ldots,\w_d\}$ be an orthonormal basis for the row span of $\A$ such that $\{\w_1,\ldots,\w_s\}$ is an orthonormal basis that extends the row span of $\{\r_1,\A_{f(2)},\ldots,\A_{f(s)}\}$ for each $2\le s\le t-1$. 
We again have from \lemref{lem:orthogonal:noise} that with high probability, $\r_1=\norm{\A_{f(1)}}_2\left(\u_1+\sum_{i=1}^d(\pm\O{\tau_i})\u_i\right)$ and $\tau_i=\frac{\eps\sqrt{\eps}}{\sqrt{C\log n}}\frac{\norm{\A\P_i}_F}{\norm{\A}_F}$ with constant at most $5$ hidden in the $\O{\cdot}$ notation and $\P_i=\u_i^\dagger\u_i$ is the projection matrix onto $\u_i$. 
We condition on this relationship between $\r_1$ and the basis $U$ and incorporate the $\frac{1}{\poly(n)}$ probability of failure into our variation distance at the end of the inductive step. 
Thus, we can apply the Gram-Schmidt process to obtain the change of basis matrix $\B$ from $U$ to $W$ whose entries are again bounded as in (\ref{eqn:cob:matrix}). 
We emphasize that the same bounds apply in the matrix $\B$ since we still receive a noisy row in the first iteration of the sampling procedure and we receive the true rows in the subsequent iterations, just as in the base case. 
Thus, \lemref{lem:orthogonal:noise} is only invoked in determining the bounds of the first row $\b_1$ and the subsequent bounds are determined using the Gram-Schmidt process, exactly as in \claimref{claim:cob}. 

We again write each row $\A_s$ in terms of basis $U$ as $\A_s=\sum_{i=1}^d\lambda_{s,i}\u_i$ and in terms of basis $W$ as $\A_s=\sum_{i=1}^d\zeta_{s,i}\w_i$. 
Then since we project away from $\A_{f(1)},\ldots,\A_{f(t-1)}$, we should have sampled $\A_s$ with probability $\frac{\norm{\A_s\Z_{t-1}}_2^2}{\norm{\A\Z_{t-1}}_F^2}=\frac{\sum_{i=t}^d\lambda_{s,i}^2}{\sum_{j=1}^n\sum_{i=t}^d\lambda_{j,i}^2}$ in round $t$ but instead we sample it with probability $\frac{\norm{\A_s\Y_{t-1}}_2^2}{\norm{\A\Y_{t-1}}_F^2}=\frac{\sum_{i=t}^d\zeta_{s,i}^2}{\sum_{j=1}^n\sum_{i=t}^d\zeta_{j,i}^2}$. 
From the change of basis matrix $\B$, we again have that for $i\ge 2$:
\begin{align*}
|\zeta_{s,i}^2-\lambda_{s,i}^2|&\le25\Big(\tau_i\lambda_{s,i}^2+\tau_i^2\lambda_{s,1}^2+\tau_i\lambda_{s,1}\lambda_{s,i}+\sum_{j,\ell\neq\{i,1\}}\tau_i^2\tau_j\tau_\ell\lambda_{s,j}\lambda_{s,\ell}\\
&+\sum_{j\neq\{i,1\}}\tau_i\tau_j\lambda_{s,i}\lambda_{s,j}+\sum_{j\neq\{i,1\}}\tau_i^2\tau_j\lambda_{s,1}\lambda_{s,j}\Big).
\label{eqn:coeff:diff}\tag{$\boxplus$}
\end{align*}
From AM-GM, we have:
\begin{align*}
\tau_i\lambda_{s,1}\lambda_{s,i}&\le\eps\lambda_{s,i}^2+\frac{\tau_i^2}{\eps}\lambda_{s,1}^2\\
\tau_i^2\tau_j\tau_\ell\lambda_{s,j}\lambda_{s,\ell}&\le\tau_i^2\tau_j^2\lambda_{s,j}^2+\tau_i^2\tau_{\ell}^2\lambda_{s,\ell}^2,\\
\tau_i\tau_j\lambda_{s,i}\lambda_{s,j}&\le\eps^2\lambda_{s,i}^2+\frac{\tau_i^2\tau_j^2}{\eps^2}\lambda_{s,j}^2,\\
\tau_i^2\tau_j\lambda_{s,1}\lambda_{s,j}&\le\tau_i^2\lambda_{s,1}^2+\tau_i^2\tau_j^2\lambda_{s,j}^2.
\end{align*}
Thus for $\eps<\frac{1}{d}$, $|\zeta_{s,i}^2-\lambda_{s,i}^2|\le25\left(2\eps\lambda_{s,i}^2+\frac{2\tau_i^2}{\eps}\lambda_{s,1}^2+4\sum_{j=2}^d\frac{\tau_i^2\tau_j^2}{\eps^2}\lambda_{s,j}^2\right)$. 
Then 
\begin{align*}
\left|\sum_{i=t}^d\zeta_{s,i}^2-\sum_{i=t}^d\lambda_{s,i}^2\right|\le25\sum_{i=t}^d\left(2\eps\lambda_{s,i}^2+\frac{2\tau_i^2}{\eps}\lambda_{s,1}^2+4\sum_{j=2}^d\frac{\tau_i^2\tau_j^2}{\eps^2}\lambda_{s,j}^2\right).
\end{align*}
Now recall that $\tau_i^2=\frac{\eps^3}{C\log n}\frac{\norm{\A\P_i}_F^2}{\norm{\A}_F^2}=\frac{\eps^3}{C\log n}\frac{\sum_{a=1}^n\lambda_{a,i}^2}{\sum_{a=1}^n\sum_{b=1}^d\lambda_{a,b}^2}$. 
Therefore we get that, 
\[
\sum_{s=1}^n\sum_{i=t}^d\sum_{j=2}^d \left(\frac{\tau_i^2}{\eps^2}\right) \lambda_{s,j}^2= \sum_{i=t}^d \left(\frac{\tau_i^2}{\eps^2}\right) \sum_{s=1}^n \sum_{j=2}^d \lambda_{s,j}^2
\le\sum_{i=t}^d \frac{\eps}{C\log n}\sum_{s=1}^n\lambda_{s,i}^2.
\] 
Similarly, we get that,
\[
\sum_{s=1}^n\sum_{i=t}^d \left(\frac{\tau_i^2}{\eps}\right) \lambda_{s,1}^2= \sum_{i=t}^d \left(\frac{\tau_i^2}{\eps}\right) \sum_{s=1}^n \lambda_{s,1}^2
\le\sum_{i=t}^d \frac{\eps^2}{C\log n}\sum_{s=1}^n\lambda_{s,i}^2.
\] 
Therefore, we can bound 
\[
\sum_{s=1}^n \left|\sum_{i=t}^d\zeta_{s,i}^2-\sum_{i=t}^d\lambda_{s,i}^2\right|\le200\eps\sum_{s=1}^n\sum_{i=t}^d\lambda_{s,i}^2
\]
so that $\norm{\A\Y_{t-1}}_F^2$ is once again within a $(1+200\eps)$ factor of $\norm{\A\Z_{t-1}}_F^2$. 
Moreover, 
\[\sum_{s=1}^n\left|\frac{\sum_{i=t}^d\zeta_{s,i}^2}{\sum_{j=1}^n\sum_{i=t}^d\lambda_{j,i}^2}-\frac{\sum_{i=t}^d\lambda_{s,i}^2}{\sum_{j=1}^n\sum_{i=t}^d\lambda_{j,i}^2}\right|\le200\eps.\]
Since we consider the total variation distance across the probability distribution, then the sampling probabilities each sum to $1$ and we have
\[\sum_{i=1}^n\left|\frac{\norm{\A_i\Y_{t-1}}_2^2}{\norm{\A\Y_{t-1}}_F^2}-\frac{\norm{\A_i\Z_{t-1}}_2^2}{\norm{\A\Z_{t-1}}_F^2}\right|=\sum_{s=1}^n\left|\frac{\sum_{i=t}^d\zeta_{s,i}^2}{\sum_{j=1}^n\sum_{i=t}^d\zeta_{j,i}^2}-\frac{\sum_{i=t}^d\lambda_{s,i}^2}{\sum_{j=1}^n\sum_{i=t}^d\lambda_{j,i}^2}\right|\le2(1+200\eps)200\eps\le799\eps,\]
for sufficiently small $\eps$. 
Thus including the $\frac{1}{\poly(n)}$ probability of failure from \lemref{lem:orthogonal:noise}, the total variation distance between the probability distributions of the index of the sample output by $\mathcal{P}_1$ and $\mathcal{P}_2$ in round $t$ is at most $800\eps$. 

From the inductive hypothesis, the total variation distance between the probability distributions of $t-1$ indices corresponding to samples output by $\mathcal{P}_1$ and $\mathcal{P}_2$ across $t-1$ rounds is at most $800(t-1)\eps$. 
Now for any sequence of rows $\mathcal{S}=\{\A_{f(2)},\ldots,\A_{f(t-1)}\}$, let $\mathcal{E}_{\mathcal{S},1}$ be the event that the rows of $\mathcal{S}$ are sequentially sampled, given that the first sampled row is $\A_{f(1)}$ and let $\mathcal{E}_{\mathcal{S},2}$ be the event that rows of $\mathcal{S}$ are sequentially sampled, given that the first sampled row is $\r_1$. 
Let $\mathsf{BAD}$ be the sets $\mathcal{S}$ such that at least one of $\A_{f(1)}\cup\mathcal{S}$ or $\r_1\cup\mathcal{S}$ is not linearly independent. 
Observe that $P_{\mathsf{BAD}}:=\sum_{\mathcal{S}\in\mathsf{BAD}}\left|\PPr{\mathcal{E}_{\mathcal{S},1}}-\PPr{\mathcal{E}_{\mathcal{S},2}}\right|\le 800(t-1)\eps$, since sampling a row $\A_{f(i)}$ that is linearly dependent with $\A_{f(1)},\ldots,\A_{f(i-1)}$ occurs with probability zero so then sampling a row $\A_{f(i)}$ that is linearly dependent with $\r_1,\ldots,\A_{f(i-1)}$ must be realized in the total variation distance in the first $t-1$ rounds. 

Otherwise, we have that the total variation distance between the probability distributions of the index corresponding to the sample output by $\mathcal{P}_1$ and $\mathcal{P}_2$ in round $t$ is at most $800\eps$. 
Let $p_j$ be the event that $\A_j$ is sampled in round $t$. 
Then we have that the probability that $\mathcal{S}\cup\A_j$ is sequentially sampled, given that the first sampled row is $\A_{f(1)}$, is $\PPr{\mathcal{E}_{\mathcal{S}\cup\A_j,1}}=\PPr{\mathcal{E}_{\mathcal{S},1}}\PPr{p_j|\mathcal{E}_{\mathcal{S},1}}$ and the probability that $\mathcal{S}\cup\A_j$ is sequentially sampled, given that the first sampled row is $\r_1$, is $\PPr{\mathcal{E}_{\mathcal{S}\cup\A_j,2}}=\PPr{\mathcal{E}_{\mathcal{S},2}}\PPr{p_j|\mathcal{E}_{\mathcal{S},2}}$. 
We have $\sum_{\mathcal{S}\notin\mathsf{BAD}}\left|\PPr{\mathcal{E}_{\mathcal{S},1}}-\PPr{\mathcal{E}_{\mathcal{S},2}}\right|\le800(t-1)\eps-P_{\mathsf{BAD}}$. 
Moreover for $\mathcal{S}\notin\mathsf{BAD}$, we have $\sum_{j\in[n]}\left|\PPr{p_j|\mathcal{E}_{\mathcal{S},1}}-\PPr{p_j|\mathcal{E}_{\mathcal{S},2}}\right|\le800\eps$. 
Thus we have 
\[\sum_{\mathcal{S}\notin\mathsf{BAD}}\sum_{j\in[n]}\left|\PPr{\mathcal{E}_{\mathcal{S}\cup\A_j,1}}-\PPr{\mathcal{E}_{\mathcal{S}\cup\A_j,2}}\right|\le800(t-1)\eps-P_{\mathsf{BAD}}+800\eps.\]
 
Since $\sum_{\mathcal{S}\in\mathsf{BAD}}\sum_{j\in[n]}\left|\PPr{\mathcal{E}_{\mathcal{S}\cup\A_j,1}}-\PPr{\mathcal{E}_{\mathcal{S}\cup\A_j,2}}\right|\le P_{\mathsf{BAD}}$, then we have that the total variation distance is at most $800t\eps$, which completes the induction. 
Thus the total variation distance between the probability distributions of $k$ indices output by $\mathcal{P}_1$ and $\mathcal{P}_2$ across $k$ rounds is at most $800k\eps$.
\end{proof}

Since the total variation distance induced by a single noisy row is small, we obtain that the total variation distance between offline adaptive sampling and our adaptive sampler is small by rescaling the error parameter. 
Thus we now provide the full guarantees for our adaptive sampler. 
\begin{theorem}
\thmlab{thm:adaptive:sampler}
Given a matrix $\A\in\mathbb{R}^{n\times d}$ that arrives in a turnstile stream, there exists a one-pass algorithm \adaptive{} that outputs a set of $k$ indices such that the probability distribution for each set of $k$ indices has total variation distance $\eps$ of the probability distribution induced by adaptive sampling with respect to squared distances to the selected subspace in each iteration. 
The algorithm uses $\O{\frac{d^3k^6}{\eps^3}\log^6 n}$ bits of space. 
\end{theorem}
\begin{proof}
Consider a set of $k+1$ processes $\mathcal{P}_1,\mathcal{P}_2,\ldots,\mathcal{P}_{k+1}$, where for each $i\in[k+1]$, $\mathcal{P}_i$ is a process that samples noisy rows from the $L_{2,2}$ sampler for the first $i-1$ rounds and actual rows from $\A$ beginning with round $i$, through adaptive sampling with $p=2$. 
Observe that $\mathcal{P}_1$ is the actual adaptive sampling process and $\mathcal{P}_{k+1}$ is the noisy process of \algref{alg:noisy:adaptive}. 
Then \lemref{lem:adaptive:tvd} argues that the total variation distance between the output distributions of the $k$ indices sampled by $\mathcal{P}_1$ and $\mathcal{P}_2$ is at most $\O{k\eps}$. 
In fact, the proof of \lemref{lem:adaptive:tvd} also shows that the total variation distance between the output distributions of the indices sampled by $\mathcal{P}_i$ and $\mathcal{P}_{i+1}$ is at most $\O{k\eps}$ for any $i\in[k]$. 
This is because the sampling distributions of $\mathcal{P}_i$ and $\mathcal{P}_{i+1}$ is identical in the first $i$ rounds, so we can use the same argument starting at round $i$ using the input matrix $\A\Q$ rather than $\A$, where $\Q$ is the projection matrix away from the noisy rows sampled in the first $i$ rounds. 
Let $\mu_i$ be the probability distribution of the $k$ indices output by $\mathcal{P}_i$. 
Thus from a triangle inequality argument, we have that 
\[\tvd(\mathcal{P}_1,\mathcal{P}_{k+1})\le\sum_{i=1}^k\tvd(\mathcal{P}_i,\mathcal{P}_{i+1})=\sum_{i=1}^k\O{k\eps}=\O{k^2\eps}.\]
In other words, the total variation distance between the probability distribution of the $k$ indices output by \algref{alg:noisy:adaptive} and the probability distribution of the $k$ indices output by adaptive sampling is at most $\O{k^2\eps}$. 
Then we obtain total variation distance $\eps$ by the appropriate rescaling factor. 

\lemref{lem:adaptive:tvd} requires the error parameter to be less than $\frac{1}{d}$. 
To analyze the space complexity, observe that with the error parameter $\O{\frac{\eps}{dk^2}}$, then \lemref{lem:adaptive:tvd} suggests that $\O{\frac{d^2k^4\log^2 n}{\eps^2}}$ buckets are necessary in each CountSketch structure in the $L_{2,2}$ sampler. 
Thus each CountSketch structure is a $\O{\frac{d^2k^4\log^2 n}{\eps^2}}\times\O{\log n}$ table. 
Each entry in the table is a vector of $d$ integers that use $\O{d\log n}$ bits of space for each vector, and the sampler can be repeated $\O{\frac{k}{\eps}\log^2 n}$ times to obtain probability of success at least $1-\frac{1}{\poly(n)}$. 
This forms one $L_{2,2}$ sampler, but we need $k$ iterations of the $L_{2,2}$ sampler to simulate $k$ rounds of adaptive sampling. 
Therefore, the total space complexity is $\O{\frac{d^3k^6}{\eps^3}\log^6 n}$. 
\end{proof}
Note that the proof of \lemref{lem:adaptive:tvd} also showed that $\norm{\A\Y_t}_F^2$ is within a $(1+\O{\eps})$ factor of $\norm{\A\Z_t}_F^2$. 
In \thmref{thm:adaptive:sampler}, we have now sampled $k$ noisy rows rather than a single noisy row followed by $k-1$ true rows, but we also rescale the error parameter down to $\O{\frac{\eps}{dk^2}}$. 
\begin{corollary}
\corlab{cor:sampler:distortion}
Suppose \algref{alg:noisy:adaptive} samples noisy rows $\r_1,\ldots,\r_k$ rather than the actual rows $\A_{f(1)},\ldots,\A_{f(k)}$. 
Let $\T_k=\A_{f(1)}\circ\ldots\circ\A_{f(k)}$, $\Z_k=\I-\T_k^\dagger\T_k$, $\R_k=\r_1\circ\ldots\circ\r_k$ and $\Y_k=\I-\R_k^\dagger\R_k$. 
Then $(1-\eps)\norm{\A\Y_k}_F^2\le\norm{\A\Z_k}_F^2\le(1+\eps)\norm{\A\Y_k}_F^2$ with probability at least $1-\eps$. 
\end{corollary}
At first glance, it might seem strange that \corref{cor:sampler:distortion} obtains increased accuracy with higher probability, but recall that \algref{alg:noisy:adaptive} has a space dependency on $\poly\left(\frac{1}{\eps}\right)$. 
\section{Applications}
\seclab{sec:apps}
In this section, we give a number of data summarization applications for our adaptive sampler. 
In each application, the goal is to find a set $S$ of $k$ rows of an underlying matrix $\A\in\mathbb{R}^{n\times d}$ defined on a turnstile stream that optimizes a given predetermined function, which quantifies how well $S$ represents $\A$. 
In particular among other things, we must show that for the purposes of each application, (1) it suffices to return a noisy perturbation of the orthogonal component at each sampling iteration, rather than the original row of the underlying matrix and (2) the algorithm still succeeds with an additive perturbation to sampling probabilities, rather than an ideal $(1\pm\eps)$ multiplicative perturbation. 

\subsection{Column/Row Subset Selection}
\seclab{sec:rss}
We first show that our adaptive sampling procedure can be used to give turnstile streaming algorithms for column/row subset selection.
Recall that in the row (respectively column) subset selection problem, the inputs are an approximation parameter $\eps>0$, a parameter $k$ for the number of selected rows or columns, and a matrix $\A\in\mathbb{R}^{n\times d}$ that arrives as a data stream and the goal is to output a set $\M$ of $k$ rows (respectively columns) of $\A$ such that $\norm{\A-\A\M^\dagger\M}_F^2\le(1+\eps)\norm{\A-\A^*_k}_F^2$ (respectively $\norm{\A-\M\M^\dagger\A}_F^2\le(1+\eps)\norm{\A-\A^*_k}_F^2$), where $\A^*_k$ is the best rank $k$ approximation to $\A$. 
For the remainder of the section, we focus on row subset selection with the assumption that $n\gg d$ as continuation of the adaptive sampling scheme in previous sections, but we note that our results extend naturally to column subset selection. 

Recall that volume sampling induces a probability distribution on subsets of rows rather than individual rows of $\A$. 
For a subset $\T$ of $k$ rows of $\A$, let $\Delta(\T)$ be the simplex defined by these rows and the origin and $\Vol(\T)$ be the volume of $\Delta(\T)$. 
Then the volume sampling probability distribution samples each subset $\T$ of $k$ rows of $\A$ with probability
\[p_{\T}=\frac{\Vol(\T)^2}{\sum_{\S:|\S|=k}\Vol(\S)^2},\]
where $\S$ is taken across all subsets of $k$ rows of $\A$. 

\cite{DeshpandeRVW06} gives the following relationship between volume sampling and row subset selection.
\begin{theorem}\cite{DeshpandeRVW06}
\thmlab{thm:vs:rss}
Given a matrix $\A\in\mathbb{R}^{n\times d}$, let $\T$ be a subset of $k$ rows of $\A$ generated from the volume sampling probability distribution. 
Then
\[\mathbb{E}_{\T}\left[\norm{\A-\A\T^\dagger\T}_F^2\right]\le(k+1)\norm{\A-\A^*_k}_F^2,\]
where $\A^*_k$ is the best rank $k$ approximation to $\A$. 
\end{theorem}

For the remainder of \secref{sec:rss}, we consider the adaptive sampling scheme in \algref{alg:as} to obtain a subset of $k$ rows of $\A$, proportional to the squared row norms of the orthogonal projection at each step. 
\begin{algorithm}[!htb]
\caption{Offline Adaptive Sampling by Squared Row Norms of Orthogonal Projection}
\alglab{alg:as}
\begin{algorithmic}[1]
\Require{Matrix $\A\in\mathbb{R}^{n\times d}$, integer $k>0$}
\Ensure{Subset of $k$ rows of $\A$}
\State{$\M\gets\emptyset$}
\For{$i=1$ to $i=k$}
\State{Choose $\r$ to be $\A_j$, with probability $\frac{\norm{\A_j(\I-\M^\dagger\M)}_2^2}{\norm{\A(\I-\M^\dagger\M)}_F^2}$, for $j\in[n]$.}
\State{$\M\gets\M\circ\r$}
\EndFor
\State{\Return $\M$}
\end{algorithmic}
\end{algorithm}

\cite{DeshpandeV06} shows that the adaptive sampling probabilities can be bounded by a multiple of the volume sampling probabilities.
\begin{lemma}\cite{DeshpandeV06}
\lemlab{lem:as:vs}
Given a matrix $\A\in\mathbb{R}^{n\times d}$, let $p_{\T}$ be the probability of sampling a set $\T$ of $k$ rows from the volume sampling probability distribution and let $q_{\T}$ be the probability of sampling $\T$ from the adaptive sampling probability distribution, as in \algref{alg:as}.  
Then $q_{\T}\le k!p_{\T}$.
\end{lemma}

Thus our adaptive sampling procedure immediately gives a one-pass turnstile streaming algorithm for column/row subset selection.
\begin{theorem}
\thmlab{thm:rss}
Given a matrix $\A\in\mathbb{R}^{n\times d}$ that arrives in a turnstile data stream, there exists a one-pass algorithm that outputs a set $\R$ of $k$ (noisy) rows of $\A$ such that 
\[\PPr{\norm{\A-\A\R^\dagger\R}_F^2\le 16(k+1)!\norm{\A-\A^*_k}_F^2}\ge\frac{2}{3}.\] 
The algorithm uses $\poly\left(d,k,\log n\right)$ bits of space. 
\end{theorem}
\begin{proof}
\thmref{thm:adaptive:sampler} states that the indices of the rows sampled by \adaptive{} have probability distribution roughly equivalent to the probability distribution of the indices of the rows sampled by adaptive sampling, as in \algref{alg:as}. 
We would thus like to apply \lemref{lem:as:vs} and \thmref{thm:vs:rss}, but we need to avoid specific failure events in our analysis. 
First, it may be the case that \adaptive{} samples a set $\R$ so that $\norm{\A-\A\R^\dagger\R}_F^2$ is very large, but the corresponding set $\T$ has very little probability of being sampled by the adaptive sampling probability distribution. 
We absorb this failure event into the total variation distance. 
Second, we must analyze the difference between $\norm{\A-\A\R^\dagger\R}_F^2$ for the set $\R$ noisy rows compared to $\norm{\A-\A\T^\dagger\T}_F^2$ for the actual rows. 
This is handled by \corref{cor:sampler:distortion}, which bounds the difference. 
We formalize these notions below. 

Let $\mathcal{E}$ be the event that algorithm \adaptive{} of \thmref{thm:adaptive:sampler} with error parameter $\eps=\frac{1}{12}$ sequentially samples an ordered set $\R$ of $k$ noisy rows corresponding to an arbitrary ordered subset $\T$ of $k$ rows of $\A$ and $\norm{\A-\A\R^\dagger\R}_F^2\le\left(1+\frac{1}{12}\right)\norm{\A-\A\T^\dagger\T}_F^2$ and note that $\PPr{\mathcal{E}}\ge\frac{11}{12}$ by \corref{cor:sampler:distortion}. 
Let $\mathcal{E}_{\T}$ be the event that algorithm \adaptive{} of \thmref{thm:adaptive:sampler} with error parameter $\eps=\frac{1}{12}$ samples a set $\R$ of $k$ noisy rows corresponding to the \emph{specific} subset $\T$ of $k$ rows of $\A$ and $\norm{\A-\A\R^\dagger\R}_F^2\le\left(1+\frac{1}{12}\right)\norm{\A-\A\T^\dagger\T}_F^2$. 

Let $\widehat{q_{\T}}$ be the probability that the indices corresponding to $\T$ are sampled by \adaptive{} of \thmref{thm:adaptive:sampler}. 
Let $q_{\T}$ be the probability of sampling $\T$ from the adaptive sampling probability distribution, as in \algref{alg:as}.  
Let $\mathcal{S}$ be the collection of $k$-sets of indices of rows $\T$ such that $\widehat{q_{\T}}>2q_{\T}$ and note that $\sum_{\T\in\mathcal{S}}q_{\T}\le\frac{1}{12}$ since $\{q_{\T}\}_{\T}$ and $\{\widehat{q_{\T}}\}_{\T}$ have total variation distance at most $\frac{1}{12}$ using \adaptive{} with error parameter $\eps=\frac{1}{12}$ by \thmref{thm:adaptive:sampler}. 
When the event $\mathcal{E}$ occurs, we abuse notation by saying $\R\notin\mathcal{S}$ if the indices corresponding to the set of sampled rows does not belong in $\mathcal{S}$. 
Observe that algorithmically, we do not know the indices corresponding to the set $\R$ of $k$ noisy rows, but analytically each row of $\R$ must correspond to a certain row of $\T$, based on the scaling of the uniform random variables at each round of \adaptive.
Then 
\begin{align*}
\underset{\R\notin\mathcal{S}}{\mathbb{E}}\left[\norm{\A-\A\R^\dagger\R}_F^2\,\Big|\,\mathcal{E}\right]&=\sum_{\T\notin\mathcal{S}}\widehat{q_{\T}}\cdot\mathbb{E}\left[\norm{\A-\A\R^\dagger\R}_F^2\,\Big|\,\mathcal{E}_{\T}\right]\\
&\le\frac{1}{\PPr{\mathcal{E}_{\T}}}\sum_{\T\notin\mathcal{S}}\widehat{q_{\T}}\left(1+\frac{1}{12}\right)\norm{\A-\A\T^\dagger\T}_F^2\\
&\le\frac{12}{11}\sum_{\T\notin\mathcal{S}}2q_{\T}\left(1+\frac{1}{12}\right)\norm{\A-\A\T^\dagger\T}_F^2\\
&\le\frac{13}{11}\sum_{\T\notin\mathcal{S}}2q_{\T}\norm{\A-\A\T^\dagger\T}_F^2,
\end{align*}
where the penultimate inequality follows from $\T\notin\mathcal{S}$ implying that $\widehat{q_{\T}}\le 2q_{\T}$. 
By \lemref{lem:as:vs} and \thmref{thm:vs:rss}, it follows that
\begin{align*}
\underset{\R\notin\mathcal{S}}{\mathbb{E}}\left[\norm{\A-\A\R^\dagger\R}_F^2\,\Big|\,\mathcal{E}\right]&\le\frac{26}{11}\sum_{\T:|\T|=k}p_{\T}k!\norm{\A-\A\T^\dagger\T}_F^2\\
&\le\frac{26}{11}(k+1)!\norm{\A-\A^*_k}_F^2,
\end{align*}
where $p_{\T}$ is the probability of sampling a set $\T$ of $k$ rows from the volume sampling probability distribution. 
Therefore by Markov's inequality, the correctness of the claim follows by additionally taking a union bound over $\PPr{\neg\mathcal{E}}\le\frac{1}{12}$ and $\sum_{\T\in\mathcal{S}}q_{\T}\le\frac{1}{12}$. 
The space of the algorithm follows from taking $\eps=\O{1}$ in \thmref{thm:adaptive:sampler}. 
\end{proof}

\subsection{Subspace Approximation}
\seclab{sec:apps:sa}
We next show that our adaptive sampling procedure can be used to give turnstile streaming algorithms for the subspace approximation. 
Recall that in the subspace approximation problem, the inputs are a parameter $p\ge 1$, a matrix $\A\in\mathbb{R}^{n\times d}$ that arrives as a data stream, and a parameter $k$ for the target dimension of the subspace, and the goal is to output a $k$-dimensional linear subspace $\H$ that minimizes $\left(\sum_{i=1}^n d(\A_i,\H)^p\right)^{\frac{1}{p}}$, where $d(\A_i,\H)=\norm{\A_i(\I-\H^\dagger\H)}_2$ is the distance from $\A_i$ to the subspace $\H$.

We consider a generalized version of the adaptive sampling scheme that appears in \algref{alg:as} to obtain a subset of $k$ rows of $\A$, where rows are sampled with probabilities proportional to $p\th$ power of the distance to the subspace formed by the span of the sampled rows. 
The generalized version, which appears in \algref{alg:as:generalized}, corresponds to \algref{alg:as} when $p=2$.  
\begin{algorithm}[!htb]
\caption{Offline Adaptive Sampling by $p\th$ Power of Distance to Subspace}
\alglab{alg:as:generalized}
\begin{algorithmic}[1]
\Require{Matrix $\A\in\mathbb{R}^{n\times d}$, integers $k>0$, $p\ge 1$}
\Ensure{Subset of $k$ rows of $\A$}
\State{$\M\gets\emptyset$}
\For{$i=1$ to $i=k$}
\State{Choose $\r$ to be $\A_j$, with probability $\frac{d(\A_j,\M)^p}{\sum_{\ell=1}^n d(\A_{\ell},\M)^p}$, for $j\in[n]$.}
\State{$\M\gets\M\circ\r$}
\EndFor
\State{\Return $\M$}
\end{algorithmic}
\end{algorithm}

\cite{DeshpandeV07} shows that adaptive sampling based on the $p\th$ powers of the subspace distances can be used to give a good approximation to the subspace approximation problem.
\begin{theorem}
\cite{DeshpandeV07}
\thmlab{thm:as:sa}
Given a matrix $\A\in\mathbb{R}^{n\times d}$, let $\T$ be a subset of $k$ rows of $\A$ generated from the adaptive sampling probability distribution, as in \algref{alg:as:generalized}. 
Then
\[\mathbb{E}_{\T}\left[\sum_{i=1}^n d(\A_i,\T)^p\right]\le((k+1)!)^p\sum_{i=1}^n d(\A_i,\A^*_k)^p,\]
where $\A^*_k$ is the best rank $k$ solution to the subspace approximation problem. 
\end{theorem}

Thus our adaptive sampling procedure immediately gives a one-pass turnstile streaming algorithm for the subspace approximation problem with $p=2$, by a similar argument to \thmref{thm:rss}.
\begin{theorem}
\thmlab{thm:sap}
Given a matrix $\A\in\mathbb{R}^{n\times d}$ that arrives in a turnstile data stream, there exists a one-pass algorithm that outputs a set $\R$ of $k$ (noisy) rows of $\A$ such that 
\[\PPr{\left(\sum_{i=1}^n d(\A_i,\R)^2\right)^{\frac{1}{2}}\le 4(k+1)!\left(\sum_{i=1}^n d(\A_i,\A^*_k)^2\right)^{\frac{1}{2}}}\ge\frac{2}{3},\]
where $\A^*_k$ is the best rank $k$ solution to the subspace approximation problem.  
The algorithm uses $\poly\left(d,k,\log n\right)$ bits of space. 
\end{theorem}
\begin{proof}
As in the proof of \thmref{thm:rss}, let $\mathcal{E}$ be the event that algorithm \adaptive{} of \thmref{thm:adaptive:sampler} with error parameter $\eps=\frac{1}{12}$ samples a set $\R$ of $k$ noisy rows corresponding to an arbitrary subset $\T$ of rows of $\A$ and $\sum_{i=1}^n d(\A_i,\R)^2\le\left(1+\frac{1}{12}\right)\sum_{i=1}^n d(\A_i,\T)^2$, so that $\PPr{\mathcal{E}}\ge\frac{11}{12}$ by \corref{cor:sampler:distortion}. 
For a set $\T$ of rows of $\A$, let $\mathcal{E}_{\T}$ be the event that algorithm \adaptive{} of \thmref{thm:adaptive:sampler} with error parameter $\eps=\frac{1}{12}$ samples a set $\R$ of $k$ noisy rows corresponding to $\T$ and $\sum_{i=1}^n d(\A_i,\R)^2\le\left(1+\frac{1}{12}\right)\sum_{i=1}^n d(\A_i,\T)^2$. 

Let $\widehat{q_{\T}}$ be the probability that the indices corresponding to $\T$ are sampled by \adaptive{} of \thmref{thm:adaptive:sampler} and $q_{\T}$ be the probability of sampling $\T$ from the adaptive sampling probability distribution, as in \algref{alg:as:generalized} with $p=2$.  
Let $\mathcal{S}$ be the set of rows $\T$ such that $\widehat{q_{\T}}>2q_{\T}$ and note that $\sum_{\T\in\mathcal{S}}q_{\T}\le\frac{1}{12}$ since $\{q_{\T}\}_{\T}$ and $\{\widehat{q_{\T}}\}_{\T}$ have total variation distance at most $\frac{1}{12}$ using \adaptive{} with error parameter $\eps=\frac{1}{12}$ by \thmref{thm:adaptive:sampler}. 
We again abuse notation by saying $\R\notin\mathcal{S}$ when the event $\mathcal{E}$ occurs if the indices corresponding to the set of sampled rows $\R$ does not belong in $\mathcal{S}$. 
Then
\begin{align*}
\underset{\R\notin\mathcal{S}}{\mathbb{E}}\left[\sum_{i=1}^n d(\A_i,\R)^2\,\Big|\,\mathcal{E}\right]&=\sum_{\T\notin\mathcal{S}}\widehat{q_{\T}}\cdot\mathbb{E}\left[\sum_{i=1}^n d(\A_i,\R)^2\,\Big|\,\mathcal{E}_{\T}\right]\\
&\le\frac{1}{\PPr{\mathcal{E}_{\T}}}\sum_{\T\notin\mathcal{S}}\widehat{q_{\T}}\cdot\left(1+\frac{1}{12}\right)\sum_{i=1}^n d(\A_i,\T)^2\\
&\le\frac{12}{11}\sum_{\T\notin\mathcal{S}}2q_{\T}\left(1+\frac{1}{12}\right)\sum_{i=1}^n d(\A_i,\T)^2\\
&\le\frac{13}{11}\sum_{\T\notin\mathcal{S}}2q_{\T}\sum_{i=1}^n d(\A_i,\T)^2,
\end{align*}
where the penultimate inequality follows from $\T\notin\mathcal{S}$ implying that $\widehat{q_{\T}}\le 2q_{\T}$. 
By \thmref{thm:as:sa}, 
\begin{align*}
\underset{\R\notin\mathcal{S}}{\mathbb{E}}\left[\sum_{i=1}^n d(\A_i,\R)^2\,\Big|\,\mathcal{E}\right]\le\frac{26}{11}((k+1)!)^2\sum_{i=1}^n d(\A_i,\A^*_k)^2.
\end{align*}
Therefore by Markov's inequality and taking a union bound over $\PPr{\neg\mathcal{E}}\le\frac{1}{12}$ and $\sum_{\T\in\mathcal{S}}q_{\T}\le\frac{1}{12}$, we have that
\[\PPr{\sum_{i=1}^n d(\A_i,\R)^2\le 16((k+1)!)^2\sum_{i=1}^n d(\A_i,\A^*_k)^2}\ge\frac{2}{3},\]
as desired. 
The space of the algorithm follows from taking $\eps=\O{1}$ in \thmref{thm:adaptive:sampler}.
\end{proof}

To simulate \algref{alg:as:generalized} for $p=1$, we need to sample a row $\A_i\P$ with probability proportional to $\norm{\A_i\P}_2$ rather than $\norm{\A_i\P}_2^2$. 
We show how to do this in \algref{alg:l12:noisy:adaptive} in \secref{sec:l12:adaptive}. 
By applying \thmref{thm:l12:adaptive:sampler} and using the same argument as \thmref{thm:sap}, we also obtain a one-pass turnstile streaming algorithm for the subspace approximation problem for $p=1$. 
\begin{theorem}
\thmlab{thm:l12:sap}
Given a matrix $\A\in\mathbb{R}^{n\times d}$ that arrives in a turnstile data stream, there exists a one-pass algorithm that outputs a set $\T$ of $k$ (noisy) rows of $\A$ such that 
\[\PPr{\sum_{i=1}^n d(\A_i,\T)\le 4(k+1)!\left(\sum_{i=1}^n d(\A_i,\A^*_k)\right)}\ge\frac{2}{3},\]
where $\A^*_k$ is the best rank $k$ solution to the subspace approximation problem.  
The algorithm uses $\poly\left(d,k,\log n\right)$ bits of space. 
\end{theorem}
\begin{proof}
Let $\mathcal{E}$ be the event that \thmref{thm:l12:adaptive:sampler} with $\eps=\frac{1}{12}$ samples a set $\R$ of $k$ noisy rows corresponding to an arbitrary subset $\T$ of rows of $\A$ and $\sum_{i=1}^n d(\A_i,\R)\le\left(1+\frac{1}{12}\right)\sum_{i=1}^n d(\A_i,\T)$, so that $\PPr{\mathcal{E}}\ge\frac{11}{12}$ by \corref{cor:l12:sampler:distortion}. 
For a set $\T$ of rows of $\A$, let $\mathcal{E}_{\T}$ be the event that \thmref{thm:l12:adaptive:sampler} with $\eps=\frac{1}{12}$ samples a set $\R$ of $k$ noisy rows corresponding to $\T$ and $\sum_{i=1}^n d(\A_i,\R)\le\left(1+\frac{1}{12}\right)\sum_{i=1}^n d(\A_i,\T)$. 

Let $\widehat{q_{\T}}$ be the probability that the indices corresponding to $\T$ are sampled by \thmref{thm:l12:adaptive:sampler} and $q_{\T}$ be the probability of sampling $\T$ from the adaptive sampling probability distribution, as in \algref{alg:as:generalized} with $p=1$.  
Let $\mathcal{S}$ be the set of rows $\T$ such that $\widehat{q_{\T}}>2q_{\T}$ and note that $\sum_{\T\in\mathcal{S}}q_{\T}\le\frac{1}{12}$ since $\{q_{\T}\}_{\T}$ and $\{\widehat{q_{\T}}\}_{\T}$ have total variation distance at most $\frac{1}{12}$ by \thmref{thm:l12:adaptive:sampler}. 
Then
\begin{align*}
\underset{\R\notin\mathcal{S}}{\mathbb{E}}\left[\sum_{i=1}^n d(\A_i,\R)\,\Big|\,\mathcal{E}\right]&=\sum_{\T\notin\mathcal{S}}\widehat{q_{\T}}\cdot\mathbb{E}\left[\sum_{i=1}^n d(\A_i,\R)\,\Big|\,\mathcal{E}_{\T}\right]\\
&\le\frac{1}{\PPr{\mathcal{E}_{\T}}}\sum_{\T\notin\mathcal{S}}\widehat{q_{\T}}\left(1+\frac{1}{12}\right)\sum_{i=1}^n d(\A_i,\T)\\
&\le\frac{13}{11}\sum_{\T\notin\mathcal{S}}2q_{\T}\sum_{i=1}^n d(\A_i,\T),
\end{align*}
where the last inequality follows from the fact that $\PPr{\mathcal{E}}\ge\frac{11}{12}$ and from $\T\notin\mathcal{S}$ implying that $\widehat{q_{\T}}\le 2q_{\T}$. 
By \thmref{thm:as:sa}, 
\begin{align*}
\underset{\R\notin\mathcal{S}}{\mathbb{E}}\left[\sum_{i=1}^n d(\A_i,\R)^2\,\Big|\,\mathcal{E}\right]\le\frac{26}{11}(k+1)!\sum_{i=1}^n d(\A_i,\A^*_k).
\end{align*}
By Markov's inequality and a union bound over $\PPr{\neg\mathcal{E}}\le\frac{1}{12}$ and $\sum_{\T\in\mathcal{S}}q_{\T}\le\frac{1}{12}$, we have that
\[\PPr{\sum_{i=1}^n d(\A_i,\R)\le 16(k+1)!\sum_{i=1}^n d(\A_i,\A^*_k)}\ge\frac{2}{3},\]
as desired. 
The space of the algorithm follows from taking $\eps=\O{1}$ in \thmref{thm:l12:adaptive:sampler}.
\end{proof}

\cite{DeshpandeV07} also shows that adaptive sampling can be used to give a bicriteria approximation to the subspace approximation problem. 
We use the notation $\tO{\cdot}$ in the remainder of \secref{sec:apps:sa} to omit $\polylog\left(k,\frac{1}{\eps}\right)$ factors, with degrees depending on $p$.

\begin{algorithm}[!htb]
\caption{Repeated Offline Adaptive Oversampling by $p\th$ Power of Distance to Subspace}
\alglab{alg:as:repeated}
\begin{algorithmic}[1]
\Require{Matrix $\A\in\mathbb{R}^{n\times d}$, integers $k>0$, $p\ge 1$}
\Ensure{Subset of $\tO{k^2\cdot \left(\frac{k}{\eps}\right)^{p+1}}$ rows of $\A$}
\State{Let $\S_0$ be a set of $k$ rows obtained through \algref{alg:as}.}
\State{$\delta\gets\frac{\eps}{\log k}$, $\S\gets\S_0$, $t\gets\O{k\log k}$}
\For{$j=1$ to $j=t$}
\For{$i=1$ to $i=k$}
\State{Sample a subset of rows $\S_i$ to be $\O{\left(\frac{2k}{\delta}\right)^p\frac{k}{\delta}\log\frac{k}{\delta}}$ rows of $\A$, where each $\A_{\ell}$ is selected with probability $\frac{d(\A_{\ell},\S_{i-1})^p}{\sum_{\ell=1}^n d(\A_{\ell},\S_{i-1})^p}$, for $\ell\in[n]$.}
\State{$\S\gets\S\circ\S_i$}
\EndFor
\State{$\S_1=\ldots=\S_k=\emptyset$}
\EndFor
\State{\Return $\S$}
\end{algorithmic}
\end{algorithm}

\begin{theorem}
\cite{DeshpandeV07}
\thmlab{thm:as:sa:bicriteria}
Given a matrix $\A\in\mathbb{R}^{n\times d}$ and a parameter $\eps>0$, let $m=\O{k}$ and let $\T_1,\ldots,\T_m$ each be a subset of $\tO{k^2\cdot\left(\frac{k}{\eps}\right)^{p+1}}$ rows generated from the adaptive sampling probability distribution with respect to repeated oversampling, as in \algref{alg:as:repeated}. 
Then for $\T=\T_1\cup\ldots\cup\T_m$,
\[\PPr{\left(\sum_{i=1}^n d(\A_i,\T)^p\right)^{\frac{1}{p}}\le(1+\eps)\left(\sum_{i=1}^n d(\A_i,\A^*_k)^p\right)^{\frac{1}{p}}}\ge\frac{3}{4},\]
where $\A^*_k$ is the best rank $k$ solution to the subspace approximation problem. 
\end{theorem}
By a similar argument to \thmref{thm:rss} and \thmref{thm:sap}, our adaptive sampling procedure gives a one-pass turnstile streaming algorithm that produces a bicriteria approximation to the subspace approximation problem with $p=2$.
\begin{theorem}
\thmlab{thm:sap:bicriteria}
Given a matrix $\A\in\mathbb{R}^{n\times d}$ that arrives in a turnstile data stream, there exists a one-pass algorithm that outputs a set $\T$ of $\tO{k^3\cdot\left(\frac{k}{\eps}\right)^{p+1}}$ (noisy) rows of $\A$ such that 
\[\PPr{\left(\sum_{i=1}^n d(\A_i,\T)^p\right)^{\frac{1}{p}}\le(1+\eps)\left(\sum_{i=1}^n d(\A_i,\A^*_k)^p\right)^{\frac{1}{p}}}\ge\frac{2}{3}\]
for $p\in\{1,2\}$, where $\A^*_k$ is the best rank $k$ solution to the subspace approximation problem. 
The algorithm uses $\poly(d,k,\frac{1}{\eps},\log n,\log\frac{k}{\eps})$ bits of space. 
\end{theorem}
\begin{proof}
The proof follows the same template as \thmref{thm:rss} that analyzes both the total variation distance between the ideal distribution and the actual distribution, as well as the quality of the approximation by the noisy rows to the actual rows. 
We first consider the case $p=2$ and observe that \thmref{thm:as:sa:bicriteria} requires \algref{alg:as:repeated} to sample $r:=\tO{k^3\cdot\left(\frac{k}{\eps}\right)^{p+1}}$ rows in total. 
We cannot quite run \algref{alg:noisy:adaptive} as stated, since it uses $k$ instances of the $L_{2,2}$ sampler in \algref{alg:l2:sampler} to sequentially sample $k$ rows. 

On the other hand, by creating $\poly(r,\log n)$ instances of the $L_{2,2}$ sampler \algref{alg:l2:sampler} with sufficiently small error parameter, we can simulate each round of \algref{alg:as:repeated} used to generate $\T_1,\ldots,\T_m$ in the statement of \thmref{thm:as:sa:bicriteria}. 
We require $\poly(r,\log n)$ instances of the $L_{2,2}$ sampler to produce $r$ samples, since recall that each sampler has some probability of outputting FAIL. 
To sample a set $\R_j$ corresponding to $\T_j$, we first use \algref{alg:noisy:adaptive} to sample a set $\Y_0$ of $k$ (noisy) rows of $\A$ with total variation distance $\O{\frac{\eps}{r}}$ from the distribution of $\S_0$ in \algref{alg:as:repeated}. 
Now for $\delta=\frac{\eps}{\log k}$, each time we should have sampled a set $\S_i$ of $\O{\left(\frac{2k}{\delta}\right)^p\frac{k}{\delta}\log\frac{k}{\delta}}$ rows of $\A$ after projecting away from $\S_{i-1}$, we instead use instances of the $L_{2,2}$ sampler \algref{alg:l2:sampler} to sample a set $\Y_i$ of $\O{\left(\frac{2k}{\delta}\right)^p\frac{k}{\delta}\log\frac{k}{\delta}}$ noisy rows of $\A$ by projecting away from $\Y_{i-1}$. 

Since we perform $r$ rounds of sampling in total, then using the same argument as \thmref{thm:adaptive:sampler}, we can bound the total variation distance between our output distribution and the distribution of \thmref{thm:as:sa:bicriteria} by $\O{\eps}$. 
Let $S$ be the set of $r$ indices corresponding to sets $\R$ of rows sampled by repeated iterations of \algref{alg:as:repeated} so that $\left(\sum_{i=1}^n d(\A_i,\R)^p\right)^{\frac{1}{p}}\le\left(1+\eps\right)\left(\sum_{i=1}^n d(\A_i,\A^*_k)^p\right)^{\frac{1}{p}}$. 
Let $\mathcal{E}$ be the event that the set $\R$ of rows sampled by repeated iterations of \algref{alg:as:repeated} correspond to a set of $r$ indices from $S$. 
Thus the probability that the indices of the $r$ samples $\T$ produced by the $L_{2,2}$ samplers correspond to indices of $S$ is at least $\PPr{\mathcal{E}}-\eps$. 

Moreover for the noisy rows $\T$ that we sample, $\left(\sum_{i=1}^n d(\A_i,\T)^p\right)^{\frac{1}{p}}\le(1+\eps)\left(\sum_{i=1}^n d(\A_i,\R)^p\right)^{\frac{1}{p}}$ with probability at least $1-\eps$ by \corref{cor:sampler:distortion}. 
By \thmref{thm:as:sa:bicriteria}, $\PPr{\mathcal{E}}\ge\frac{3}{4}$. 
Thus for sufficiently small $\eps$ and by a rescaling argument, the probability that $\left(\sum_{i=1}^n d(\A_i,\T)^p\right)^{\frac{1}{p}}\le(1+\eps)\left(\sum_{i=1}^n d(\A_i,\A^*_k)^p\right)^{\frac{1}{p}}$ is at least $\frac{2}{3}$. 
Then the correctness of the claim holds and the total space required is $\poly(d,k,\frac{1}{\eps},\log n,\log\frac{k}{\eps})$ by argument in \thmref{thm:adaptive:sampler}. 
For $p=1$, we use the $L_{1,2}$ sampler \algref{alg:l12:sampler} in place of the $L_{2,2}$ sampler \algref{alg:l2:sampler} to sample rows $\Y_i$ for $i>0$ with probability proportional to their distances from the current subspace at each iteration, rather than the squared distances. 
Correctness then follows from the same argument, using \thmref{thm:l12:adaptive:sampler} and \corref{cor:l12:sampler:distortion} for the $L_{1,2}$ samplers. 
\end{proof}

\subsection{Projective Clustering}
We now show that our adaptive sampling procedure can also be used to give turnstile streaming algorithms for projective clustering, where the inputs are a parameter $p\ge 1$, a matrix $\A\in\mathbb{R}^{n\times d}$ that arrives as a data stream and parameters $k$ for the target dimension of each subspace and $s$ for the number of subspaces, and the goal is to output $s$ $k$-dimensional linear subspaces $\H_1,\ldots,\H_s$ that minimizes: 
\[\left(\sum_{i=1}^n d(\A_i,\H)^p\right)^{\frac{1}{p}},\]
where $\H=\H_1\cup\ldots\cup\H_s$ and $d(\A_i,\H)$ is the distance from $\A_i$ to union $\H$ of $s$ subspaces $\H_1,\ldots,\H_s$. Again we use $\tO{\cdot}$ to omit $\polylog\left(k,s,\frac{1}{\eps}\right)$ factors, with degrees depending on $p$. 
 
\begin{algorithm}[!htb]
\caption{Dimensionality Reduction for Projective Clustering}
\alglab{alg:pc}
\begin{algorithmic}[1]
\Require{Matrix $\A\in\mathbb{R}^{n\times d}$, integers $k>0$, $s>0$, $p\ge 1$, and a subspace $\V$ of dimension at least $k$.}
\Ensure{Subset of $\tO{\left(\frac{k^2}{\eps}\right)^p\frac{k^4 s}{\eps^2}}$ rows of $\A$}
\State{$\S\gets\emptyset$}
\For{$t=1$ to $t=\tO{\left(\frac{k^2}{\eps}\right)^p\frac{k^4 s}{\eps^2}}$}
\State{Sample a row $\r$ of $\A$, where each row $\A_i$ is selected with probability $\frac{d(\A_i,\S\cup\V)}{\sum_{j=1}^n d(\A_j,\S\cup\V)}$.}
\State{$\S\gets\S\circ\r$}
\EndFor
\State{\Return $\S$}
\end{algorithmic}
\end{algorithm}
\cite{DeshpandeV07} also shows that adaptive sampling can be used to perform dimensionality reduction for projective clustering. 
\begin{theorem}
\cite{DeshpandeV07}
\thmlab{thm:as:pc}
Let $\V$ be a subspace of dimension at least $k$ such that
\[\left(\sum_{i=1}^n d(\A_i,\V)^p\right)^{\frac{1}{p}}\le2\left(\sum_{i=1}^n d(\A_i,\H)^p\right)^{\frac{1}{p}},\]
where $\H$ is the union of $s$ $k$-dimensional subspaces that is the optimal solution to the projective clustering problem. 
Then with probability at least $\frac{3}{4}$, \algref{alg:pc} outputs a set $\S$ such that $\V\cup\S$ contains a union $\T$ of $s$ $k$-dimensional subspaces such that
\[\left(\sum_{i=1}^n d(\A_i,\T)^p\right)^{\frac{1}{p}}\le(1+\eps)\left(\sum_{i=1}^n d(\A_i,\H)^p\right)^{\frac{1}{p}}.\]
\end{theorem}
Since the optimal solution to the projective clustering problem is certainly no better than the optimal $ks$-dimensional subspace, we use the bicriteria subspace approximation algorithm of \thmref{thm:sap:bicriteria} with input dimension $ks$.  
Thus, we obtain a one-pass turnstile streaming algorithm for the projective clustering problem.
\begin{theorem}
\thmlab{thm:pc}
Given a matrix $\A\in\mathbb{R}^{n\times d}$ that arrives in a turnstile data stream, there exists a one-pass algorithm that outputs a set $\S$ of $\tO{(ks)^{p+4}+\frac{k^4s}{\eps^2}\left(\frac{k^2}{\eps}\right)^p}$ (noisy) rows of $\A$, which includes a union $\T$ of $s$ $k$-dimensional subspaces such that
\[\PPr{\left(\sum_{i=1}^n d(\A_i,\T)^p\right)^{\frac{1}{p}}\le(1+\eps)\left(\sum_{i=1}^n d(\A_i,\H)^p\right)^{\frac{1}{p}}}\ge\frac{2}{3}\]
for $p=\{1,2\}$, where $\H$ is the union of $s$ $k$-dimensional subspaces that is the optimal solution to the projective clustering problem. 
The algorithm uses $\poly\left(d,k,s,\frac{1}{\eps},\log n\right)$ bits of space. 
\end{theorem}
\begin{proof}
Note that the optimal solution $\H$ of $s$ $k$-dimensional subspaces is no better than the solution $\A^*_{ks}$ of the $ks$ subspace approximation problem: 
\[\left(\sum_{i=1}^n d(\A_i,\A^*_{ks})^p\right)^{\frac{1}{p}}\le\left(\sum_{i=1}^n d(\A_i,\H)^p\right)^{\frac{1}{p}},\]
By setting $\eps=\O{1}$ in \thmref{thm:sap:bicriteria}, we can first obtain a set $\V$ of $\tO{(ks)^3\cdot\left(ks\right)^{p+1}}$ rows of $\A$ such that
\[\left(\sum_{i=1}^n d(\A_i,\V)^p\right)^{\frac{1}{p}}\le2\left(\sum_{i=1}^n d(\A_i,\A^*_{ks})^p\right)^{\frac{1}{p}}\le2\left(\sum_{i=1}^n d(\A_i,\H)^p\right)^{\frac{1}{p}},\]
which satisfies the conditions of \thmref{thm:as:pc}. 
This can be done with arbitrarily high constant probability by taking $\V$ to be a number of independent instances of the algorithm in \thmref{thm:sap:bicriteria}. 

We now simulate \algref{alg:pc} by iteratively sampling $r:=\tO{\left(\frac{k^2}{\eps}\right)^p\frac{k^4 s}{\eps^2}}$ rows projected away from $\V$. 
The rest of the proof uses the same template as \thmref{thm:sap:bicriteria}. 
We first describe the case where $p=2$. 
By using $\poly(r,\log n)$ independent copies of \algref{alg:l2:sampler} with sufficiently small error parameter, we can iteratively sample rows with probability proportional to their squared distance away from the span of the previous rows and $\V$. 
By \thmref{thm:adaptive:sampler}, it follows that the total variation distance of the sampled rows $\S$ using the $L_{2,2}$ samplers is some small constant $\O{\eps}$ from the output distribution of \thmref{thm:as:pc}. 
By \corref{cor:sampler:distortion}, the objective on the sampled output is within $(1+\O{\eps})$ of the offline adaptive sampler. 
Hence, the existence of a union $\T$ of $s$ $k$-dimensional subspaces that is a good approximation of the optimal solution follows from \thmref{thm:as:pc}. 
Since the probability of failure of \thmref{thm:as:pc} is at most $\frac{1}{4}$, then for sufficiently small $\eps$, it holds that $\left(\sum_{i=1}^n d(\A_i,\T)^p\right)^{\frac{1}{p}}\le(1+\eps)\left(\sum_{i=1}^n d(\A_i,\H)^p\right)^{\frac{1}{p}}$ with probability at least $\frac{2}{3}$. 
By \thmref{thm:sap:bicriteria}, we can obtain $\V$ using $\poly\left(d,k,s,\frac{1}{\eps},\log n\right)$ bits of space.  
By \thmref{thm:adaptive:sampler}, we need $\poly(d,k,\frac{1}{\eps},\log n,\log\frac{k}{\eps})$ bits of space to sample the $r$ rows of $\A$, so the space complexity follows. 

For $p=1$, we instead use $\poly(r,\log n)$ independent copies of the $L_{1,2}$ sampler \algref{alg:l12:sampler}. 
The same argument then follows using \thmref{thm:l12:adaptive:sampler} to bound the total variation distance from the output distribution of \thmref{thm:as:pc} by $\O{\eps}$ and \corref{cor:l12:sampler:distortion} to bound the objective on the sampled output compared to that of the offline adaptive sampler. 
\end{proof}

\subsection{Volume Maximization}
\seclab{sec:vm:ts}
We now show that our sampling procedure can also be used to give turnstile streaming algorithms for volume maximization, where the inputs are a matrix $\A\in\mathbb{R}^{n\times d}$ that arrives as a data stream and a parameter $k$ for the number of selected rows, and the goal is to output $k$ rows $\r_1,\ldots,\r_k$ of $\A$ that maximize $\Vol(\R)$, where $\R=\r_1\circ\ldots\circ\r_k$. 
A possible approach to the volume maximization problem in an offline model is the greedy algorithm, which repeatedly chooses the row with the largest distance from the subspace spanned by the rows that have already been selected, for $k$ steps. 
\cite{CivrilM09} shows that this offline greedy algorithm gives a $k!$-approximation to the volume maximization problem. 
In fact, their analysis also implies that an offline \emph{approximate} greedy algorithm gives a good approximation to the volume maximization problem. 
\begin{theorem}
\cite{CivrilM09}
\lemlab{lem:as:vm}
Given a matrix $\A\in\mathbb{R}^{n\times d}$ and an integer $k>0$, let $\G$ be a set of $k$ rows chosen by the approximate greedy algorithm that repeatedly chooses a (noisy) row whose distance from the subspace spanned by the rows that have already been chosen is within a multiplicative $\alpha$ factor of the largest distance of a row to the subspace. 
Let $\V$ be a set of $k$ rows of $\A$ with the maximum volume. 
Then $\Vol(\V)\le (\alpha^k)k!\cdot\Vol(\G)$.  
\end{theorem}

\begin{algorithm}[!htb]
\caption{Volume Maximization}
\alglab{alg:as:vol:max}
\begin{algorithmic}[1]
\Require{Matrix $\A\in\mathbb{R}^{n\times d}$ that arrives as a stream $\A_1,\ldots,\A_n\in\mathbb{R}^d$, parameter $k$ for number of rows, approximation factor $\alpha>1$.}
\Ensure{$k$ Noisy and projected rows of $\A$.}
\State{Create instances $\alg_1,\ldots,\alg_k$, where each $\alg_j$ for $j\in[k]$ is $\O{\log n\log k}$ copies of the $L_{2,2}$ sampler of \algref{alg:l2:sampler} with error parameter $\eps=\frac{1}{4}$.}
\State{Create instances $\countsketch_1,\ldots,\countsketch_k$ with $\Theta\left(\frac{nk}{\alpha^2}\right)$ buckets.}
\State{Create instances $\ams_1,\ldots,\ams_k$ with error parameter $\O{1}$.}
\State{Let $\M$ be empty $0\times d$ matrix.}
\State{\textbf{Streaming Stage:}}
\For{each row $\A_i$}
\State{Update each sketch $\alg_1,\ldots,\alg_k$.}
\State{Update each sketch $\ams_1,\ldots,\ams_k$.}
\State{Update each sketch $\countsketch_1,\ldots,\countsketch_k$.}
\EndFor
\State{\textbf{Post-processing Stage:}}
\For{$j=1$ to $j=k$}
\State{Post-processing matrix $\P\gets\I-\M^\dagger\M$.}
\State{Update $\alg_j,\ams_j,\countsketch_j$ with post-processing matrix $\P$.}
\If{$\ams_j$ and $\countsketch_j$ find a (noisy) row $\r$ of $\A\P$ with $\norm{\r}^2_2\ge\frac{\alpha^2}{4nk}\norm{\A\P}_F^2$}
\State{Append the (noisy) row $\r_j$ with the largest norm to $\M$: $\M\gets\M\circ\r_j$.}
\Else
\State{Let $\r_j$ be the (noisy) row of $\A\P$ sampled by $\alg_j$.}
\State{Append $\r_j$ to $\M$: $\M\gets\M\circ\r_j$.}
\EndIf
\EndFor
\State{\Return $\M$.}
\end{algorithmic}
\end{algorithm}

\begin{lemma}
\lemlab{lem:as:greedy}
Let $\A\in\mathbb{R}^{n\times d}$ and $\alpha>1$ be an approximation factor. 
Let $\r_j$ be the row selected by \algref{alg:as:vol:max} in round $j$ and let $\R_j=\r_1\circ\ldots\circ\r_j$ for each $j\in[k]$ and $\R_0$ be the all zeros matrix. 
For each $j\in[k]$, let $\m_j$ be the row the maximizes $\norm{\A_i(\I-\R_{j-1}^\dagger\R_{j-1})}_2$. 
Then $\norm{\r_j}_2\ge\frac{1}{2\alpha}\norm{\m_j}_2$ with probability at least $1-\frac{1}{4k}-\frac{1}{\poly(n)}$. 
\end{lemma}
\begin{proof}
Let $\Y_j=\I-\R_j^{\dagger}\R_j$ for each $j\in[k]$. 
Suppose $\norm{\m_j}_2^2\ge\frac{\alpha^2}{4nk}\norm{\A\Y_{j-1}}_F^2$ so that $\m_j$ is $\frac{\alpha}{2\sqrt{nk}}$-heavy with respect to $\A\Y_{j-1}$. 
Each $\countsketch$ data structure in \algref{alg:as:vol:max} maintains $\Theta\left(\frac{nk}{\alpha^2}\right)$ buckets and $\m_j$ is a heavy row for $\A\Y_{j-1}$, so \algref{alg:as:vol:max} will output some noisy row $\r_j$ with $\norm{\r_j}_2^2\ge\frac{1}{2}\norm{\m_j}_2^2$ for sufficiently small $\O{1}$ error parameter by each $\ams$ data structure.  

On the other hand, suppose $\norm{\m_j}_2^2<\frac{\alpha^2}{4nk}\norm{\A\Y_{j-1}}_F^2$. 
Since $\A\Y_{j-1}$ contains $n$ rows and $\eps=\frac{1}{4}$, the $L_{2,2}$ sampler will select an index $s\in[n]$ such that $\norm{\A_s\Y_{j-1}}_2^2\ge\frac{1}{8nk}\norm{\A\Y_{j-1}}_F^2$ and output a row $\r_j$ such that $\norm{\r_j}_2\ge\frac{1}{\sqrt{2}}\norm{\A_s\Y_{j-1}}_2$, with probability at least $1-\frac{1}{4k}-\frac{1}{\poly(n)}$ 
Therefore, 
\[\norm{\r_j}_2^2\ge\frac{1}{16nk}\norm{\A\Y_{j-1}}_F^2>\frac{1}{4\alpha^2}\norm{\m_j}_2^2,\]
which suffices to imply $\norm{\r_j}_2\ge\frac{1}{2\alpha}\norm{\m_j}_2$. 
\end{proof}

\begin{theorem}
\thmlab{thm:vm}
Given a matrix $\A\in\mathbb{R}^{n\times d}$ that arrives in a turnstile data stream and an approximation factor $\alpha>0$, there exists a one-pass streaming algorithm that outputs a set of $k$ noisy rows of $\A$ that is an $\alpha^k(k!)$-approximation to volume maximization with probability at least $\frac{2}{3}$, using $\tO{\frac{ndk^2}{\alpha^2}}$ bits of space.  
\end{theorem}
\begin{proof}
Recall that a single instance of the $L_{2,2}$ sampler of \algref{alg:l2:sampler} succeeds with probability $\frac{1}{\log n}$ for error parameter $\eps=\O{1}$. 
Thus by using $\O{\log n\log k}$ copies of the $L_{2,2}$ sampler of \algref{alg:l2:sampler}, the probability that \algref{alg:as:vol:max} successfully acquires a sample in each round is at least $1-\O{\frac{1}{k}}$. 
Thus by \lemref{lem:as:greedy} and a union bound, \algref{alg:as:vol:max} repeatedly chooses rows $\r_1,\ldots,\r_k$ whose distance from the subspace spanned by the previously chosen rows is within a multiplicative factor of $2\alpha$ of the largest distance, with probability at least $\frac{2}{3}$. 
Let $\R$ be the parallelepiped spanned by $\r_1,\ldots,\r_k$. 
Thus by \lemref{lem:as:vm}, $(2\alpha)^k(k!)\Vol(\R)\ge\Vol(\V)$, where $\V$ is a set of $k$ rows of $\A$ with the maximum volume. 
The result then follows from rescaling $\alpha$. 

Each \countsketch{} data structure maintains $\Theta\left(\frac{nk}{\alpha^2}\right)$ buckets of vectors with $d$ entries, each with $\O{\log n}$ bits. 
Each \ams{} data structure uses $\O{d\log^2 n}$ bits of space. 
Each $L_{2,2}$ sampler uses $\O{\log n}$ buckets of vectors with $d$ entries, each with $\O{\log n}$ bits. 
Since \algref{alg:as:vol:max} requires $k$ instances of each data structure, then the total space complexity follows. 
\end{proof}
\section{Volume Maximization in the Row-Arrival Model}
\seclab{sec:vm:ra}
In this section, we consider the volume maximization problem on row-arrival streams. 
As before, we are given the rows of the matrix $\A\in\mathbb{R}^{n\times d}$ and a parameter $k$, and the goal is to output $k$ rows of the matrix whose volume is maximized. 
Throughout the section, we will use the equivalent view that the stream consists of $n$ points from $\mathbb{R}^d$. 

\subsection{Volume Maximization via Composable Core-sets.}
We first observe that we can get an approximation algorithm for volume maximization in the row-arrival model by using algorithms of \cite{indyk2018composable, mahabadi2019composable} for composable core-sets for volume maximization. 
We recall the definition of composable core-sets, such as given in~\cite{IndykMMM14}.
\begin{definition}[$\alpha$-composable core-set]
Let $V\subset\mathbb{R}^d$ be an input set. 
Then a function $c:V\to W$, where $W\subset V$, is an $\alpha$-composable core-set for an optimization problem with a maximization objective with respect to a function $f: 2^{\mathbb{R}^d}\to\mathbb{R}$ if for any collection of set $V_1,\ldots,V_m\subset\mathbb{R}^d$, 
\[f(c(V_1)\cup\ldots\cup c(V_m))\ge\frac{1}{\alpha}f(V_1\cup\ldots\cup V_m).\]
\end{definition}
\cite{indyk2018composable} gives composable core-sets for volume maximization.
\begin{theorem}
\thmlab{thm:vm:coreset}
\cite{indyk2018composable}
There exists a polynomial time algorithm for computing an $\tO{k}^{k/2}$-composable core-set of size $\tO{k}$ for the volume maximization problem. 
\end{theorem}
We can partition the stream into consecutive blocks and apply a core-set for each block to get a streaming algorithm for volume maximization. 
\begin{corollary}\corlab{core-set-streaming}
There exists a one pass streaming algorithm in the row-arrival model that computes a $\tO{k}^{k/(2\eps)}$-approximation to the volume maximization problem, using $\tO{\frac{1}{\eps}n^{\eps}dk}$ space.
\end{corollary}
\begin{proof}
Consider a $b$-ary tree over the stream with $n$ leaves that correspond to the elements of the stream in the order they arrive. 
Let $b=n^{\eps}$ so that the height of the tree is $\frac{\log n}{\log b} = \frac{1}{\eps}$. 
For each node in the tree, as soon as all the elements corresponding to its subtree arrive in the stream, we build a core-set of size $\tO{k}$ for the points using the algorithm of \cite{indyk2018composable} in \thmref{thm:vm:coreset}, and pass the core-set to the parent node. 
More precisely, when the node receives a composable core-set from each of its $b$ children, it computes a composable core-set over the union of the core-sets of its children and passes the new core-set on to its parent. 

Then the $k$ points reported by the root gives an $\left(\tO{k}^{k/2}\right)^{1/\eps} = \tO{k}^{k/(2\eps)}$ approximation to the volume maximization problem. 
Moreover, at each time step during the stream arrival, there is only one path of active nodes (nodes whose corresponding leaf nodes have arrived but not finished) in the tree from the root to the leaves. 
Each of the nodes on this active path might need to store a composable core-set of size $\tO{k}$ for each of its $b$ children. 
Since each point has dimension $d$, then the total memory usage of the algorithm is thus at most $\left(\frac{1}{\eps}\right)\cdot b \cdot \tO{dk} = \tO{\frac{1}{\eps}n^{\eps}dk}$.
\end{proof}

%%%%%%%%%%%%%%%%%%%%%%%%%%%%%%%
\subsection{Exponential Dependence on $d$}
In this section, we give a streaming algorithm whose space complexity depends exponentially on the dimension $d$. Our main tool is the $\eps$-kernels of \cite{agarwal2005geometric} improved by \cite{chan2006faster} for directional width of a point set. 
We first define the concept of the directional width.

\begin{definition}[Directional width \cite{agarwal2005geometric}]
Given a point set $P\subset \mathbb{R}^d$ and a unit direction vector $\x\in \mathbb{R^d}$, the directional width of $P$ with respect to $\x$ is defined to be $\omega(\x,P) = \max_{\p\in P} \langle \x,\p\rangle - \min_{\p\in P} \langle\x,\p\rangle$.
\end{definition}

The following lemma shows the existence of core-sets with size exponential in the directional width of a point set but independent of the number of points.

\begin{lemma}
\cite{chan2006faster}
\lemlab{agarwal-coreset}
For any $0<\eps<1$, there exists a one pass streaming algorithm that computes an $\eps$-core-set $Q\subseteq P$, such that for any unit direction vector $\x\in \mathbb{R}^d$, we have $\omega(\x,Q)\geq (1-\eps) \omega(\x,P)$. 
Moreover the algorithm uses space $\O{\left(\frac{1}{\eps}\log\frac{1}{\eps}\right)^{d-1}}$ and the core-set has size $\abs{Q}\leq \O{\frac{1}{\eps^{d-1}}}$.
\end{lemma}

Now let us define the the directional height (which was implicitly defined in \cite{mahabadi2019composable}), and its relation to directional width.

\begin{definition}[Directional height]
Given a point set $P\subset \mathbb{R}^d$ and a unit direction vector $\x\in \mathbb{R^d}$, the directional height of $P$ with respect to $x$ is defined to be $h(\x,P) = \max_{\p\in P} \abs{\langle \x,\p\rangle}$.
\end{definition}

\begin{lemma}\lemlab{width-to-height-coreset}
For a point set $P\subset \mathbb{R}^d$, an $\eps$-core-set $Q$ for directional width is a $(2\eps)$-core-set for directional height.
\end{lemma}
\begin{proof}
Let $\x$ be a unit direction vector in $\mathbb{R}^d$.
Let $\p_1=\argmax_{\p\in P} \langle \x,\p\rangle$ and $\p_2=\argmin_{\p\in P} \langle \x,\p\rangle$, and let $\q_1=\argmax_{\p\in Q} \langle \x,\p\rangle$ and $\q_2=\argmin_{\p\in Q} \langle \x,\p\rangle$. 
Now consider two cases. 
\begin{itemize}
\item 
First suppose that $\langle\x,\p_2\rangle\geq 0$. 
In this case, we have that $h(\x,P) = \langle \x,\p_1\rangle$. 
Since $Q$ is an $\eps$-core-set for directional width, then $\langle\x,\p_1\rangle-\langle\x,\p_2\rangle\le(1+\eps)\left(\langle\x,\q_1\rangle-\langle\x,\q_2\rangle\right)\le(1+\eps)\left(\langle\x,\q_1\rangle-\langle\x,\p_2\rangle\right)$. 
Therefore, $\langle \x,\q_1\rangle \geq (1-\eps) \langle \x, \p_1\rangle$ so that $h(\x,Q)\geq (1-\eps)h(\x,P)$.

The case of $\langle\x,\p_1\rangle\leq 0$ can be handled similarly.

\item
On the other hand, suppose $\langle\x,\p_1\rangle > 0$ and $\langle \x,\p_2\rangle < 0$. 
Assume without loss of generality that $\abs{\langle \x,\p_1\rangle}\geq \abs{\langle \x,\p_2\rangle}$. 
Then $h(\x,P) = \langle \x,\p_1\rangle$. 
As $Q$ is an $\eps$-core-set for the width, then we also have $\langle \x,\q_1\rangle \geq (1-2\eps)\langle \x,\p_1 \rangle= (1-2\eps)h(\x,P)$. 
Therefore, $h(\x,Q)\geq (1-2\eps)h(\x,P)$.
\end{itemize}
\end{proof}

Finally we define $k$-directional height and observe that a core-set for directional height leads to a core-set for $k$-directional height.

\begin{definition}[$k$-directional height \cite{mahabadi2019composable}]
Given a point set $P\subset \mathbb{R}^d$ and a $(k-1)$-dimensional subspace $\mathcal{H}$, $\x\in \mathbb{R}^d$, the $k$-directional height of $P$ with respect to $\mathcal{H}$ is defined to be $h_k(\mathcal{H},P) = \max_{\x\in\mathcal{H}^{\perp}} h(\x,P)$, where $\mathcal{H}^\perp$ is the orthogonal complement of $\mathcal{H}$.
\end{definition}

\begin{observation}\obslab{height-to-dirheight-coreset}
For a point set $P\subset \mathbb{R}^d$, an $\eps$-core-set $Q$ for directional height is an $\eps$-core-set for $k$-directional height.
\end{observation}
In fact, a core-set for directional height is stronger than a core-set for $k$-directional height since it preserves the height in all directions $x\in \mathcal{H}^{\perp}$; thus their maximum is preserved too.
\begin{lemma}[\cite{mahabadi2019composable}]\lemlab{dirheight-to-volmax-coreset}
For a point set $P\subset \mathbb{R}^d$, let $Q$ be its $\eps$-core-set for $k$-directional height. 
Then the solution of the $k$-volume maximization on $Q$ is within a factor of $1/(1-\eps)^k$ of the solution of $k$-volume maximization over $P$.
\end{lemma}

\begin{lemma}\lemlab{volmax-exp-d}
There exists a one pass streaming algorithm that outputs a $2^{k}$-approximation to volume maximization, using $\O{8^d}$ space.
\end{lemma}
\begin{proof}
We use the streaming algorithm of \lemref{agarwal-coreset} on the the set of rows of $\A$ with $\eps=\frac{1}{4}$, which gives a $\frac{1}{4}$-core-set $Q\in \mathbb{R}^{4^d\times d}$ for the directional width using space $\O{8^d}$. 
Using \lemref{width-to-height-coreset} and \obsref{height-to-dirheight-coreset}, this will be a $\frac{1}{2}$-core-set for the $k$-directional height of the rows of $\A$. 
Finally, using \lemref{dirheight-to-volmax-coreset}, the optimal solution of $Q$ approximates the maximum volume over rows of $\A$ within a factor of $2^k$.
\end{proof}

%%%%%%%%%%%%%%%%%%%%%%%%%%%%%%%

\subsection{Dimensionality Reduction}
In this section, we show how to reduce the dimension of each point to $d=\O{k}$. 
Using the result of the previous section, this will give a trade-off algorithm, improving over \corref{core-set-streaming} in terms of the dependence on the parameter $\frac{1}{\eps}$. 
We prove the following lemma.

\begin{lemma}
\lemlab{volmax-tradeoff}
Let $C$ be a trade-off parameter such that $1<C<(\log n)/k$. 
There exists a randomized streaming algorithm that uses $\O{n^{\O{1/C}} d}$ space to computes a subset of size $d$ whose volume maximization solution is a $\O{Ck}^{k/2}$ approximation to the optimal solution. 
\end{lemma}
Note that this result improves the algorithm of \corref{core-set-streaming} for $\frac{k}{\log n} < \eps$: setting $C=1/\eps$, this provides an algorithm with memory usage of $\O{n^{\eps}}$, with approximation factor of $\O{k/\eps}^{k/2}$, improving the dependence of the approximation factor on $\frac{1}{\eps}$ by an exponential factor.

\medskip
We now continue with the proof of \lemref{volmax-tradeoff}. 
Consider a random matrix $\G\in \mathbb{R}^{d\times r}$, for $r=\frac{\log n}{C}$, where each of its entries is an independent and identically distributed (i.i.d.) random variable drawn from the Gaussian distribution $\mathcal{N}(0,1/r)$. 
Consider the matrix $\A\G$ and observe that its rows exist in an $r$ dimensional space. 
Therefore, we can use the streaming algorithm of \lemref{volmax-exp-d} to find a subset of $k$ rows of $\A\G$ that serves as a good estimator for the maximum volume.  
This approach requires $\O{2^{3r}} = {\O{1}}^{\O{(\log n)/C}} = n^{\O{1/C}}$ memory space.

\begin{lemma}
\lemlab{lem:vm:dr}
Let $\G\in \mathbb{R}^{d\times r}$, for $r=\Omega\left(\frac{\log n}{C}\right)$, have each of its entries is drawn i.i.d from the Gaussian distribution $\mathcal{N}(0,1/r)$. 
With high probability, the maximum volume of the optimal $k$-subset of the rows of $\A\G$ is within $2^k$ of the maximum volume of the optimal $k$-subset of the rows of $\A$.
\end{lemma}
\begin{proof}
Let $P=\{\p_1,\cdots,\p_k\}$ be the subset of $k$ points among the rows of $\A$ that maximizes the volume. 
Moreover, let $R=\{\r_2,\cdots,\r_k\}$ where $r_i$ is the projection of $\p_i$ onto the subspace spanned by the points in $\{\p_1,\cdots,\p_{i-1}\}$ for each $i\in[k]$. 
Using the Johnson-Lindenstrauss Lemma with $\eps=\frac{1}{2}$ on the set of $2k$ points $P\cup R$, the lengths of each row of $\A$ is only distorted by a factor of at most two compared to the length of the corresponding row in $\A\G$ as long as $r=\Omega(\log k)$, which is always the case for $C < (\log n)/k$. 
Hence, the maximum volume $k$-subset of rows of $\A$ does not decrease by more than a factor of $2^k$ with high probability.
\end{proof}

We now show that for every other subset $S$ of $k$ points from the rows of $\A$, their volume does not increase by much with very high probability, so that we can union bound over all such subsets. 
The following lemma may seem counterintuitive at first, since the parameter $C$ appears in the approximation factor but not the probability. 
However, recall that the algorithm pays for the parameter $C$ in the space of the algorithm. 

\begin{lemma}
Let $S$ be a subset of size $k$ from the rows of $\A$. Then after applying $\G$, its volume does not increase by more than a factor of $(\sqrt{2Ck}+2)^{k} = \O{Ck}^{k/2}$ with probability at least $1-n^{-k}$.
\end{lemma}
\begin{proof}
Let $\R$ be the $k\times d$ submatrix corresponding to the rows of $\A$ that are in $S$. 
The volume of $\R$ after the embedding is equivalent to $\sqrt{\det(\R\G\G^\top\R^\top)}$. 
Now consider the singular value decomposition of $\R = \U\Sigma \V^\top$ where $\U$ and $\Sigma$ are $k\times k$ (as otherwise the original volume would have been $0$), and $\V$ is $d\times k$. 
Then we can rewrite this volume as $\sqrt{\det(\U\Sigma \V^\top \G \G^\top \V\Sigma \U^\top)} = \sqrt{\det(\Sigma \H \H^\top \Sigma)}$, where $\H = \V^\top\G$ is a $k\times r$ matrix with entries drawn i.i.d. from the Gaussian distribution $\mathcal{N}(0,1/r)$, due to the rotational invariance of $\G$. 
Then $\sqrt{\det(\Sigma \H\H^\top \Sigma)}$ is just the product of the singular values of $\H^\top \Sigma$. 

\begin{claim}
\cite{StackExchange}
The $i\th$ singular value of $\H^\top\Sigma$ is at most $\norm{\H}_2\Sigma_{i,i}$.
\end{claim}
\begin{proof}
We include the proof for completeness. 
Observe that for any vector $\x\in \mathbb{R}^k$, 
\[\frac{\langle \Sigma\H\H^\top\Sigma \x, \x\rangle}{\langle \x,\x\rangle} = \frac{\x^\top\Sigma\H\H^\top\Sigma \x}{\langle \x,\x\rangle} = \frac{\norm{\H^\top\Sigma \x}^2}{\langle \x,\x\rangle} \leq \frac{\norm{\H^\top}_2^2 \x^\top\Sigma^2 \x}{\langle \x,\x\rangle} = \norm{\H^\top}_2^2\frac{\langle\Sigma \x, \x\rangle}{\langle \x,\x\rangle}.\]
Note that the $i\th$ singular value of $\H^\top\Sigma$ is the square root of the $i\th$ eigenvalue of $\Sigma \H \H^\top \Sigma$. 
Thus the min-max theorem for the characterization of the eigenvalues of $\Sigma \H \H^\top \Sigma$ shows that for any $i$, the $i$th singular value of $\H^\top\Sigma$ can be bounded by $\norm{\H}_2$ times the $i\th$ singular value of $\Sigma$, which is $\Sigma_{i,i}$.
\end{proof}
Thus, it follows that the volume of $\R\G$ is at most $\norm{\H}_2^k \Vol(\R)$, where $\Vol(\R) = \prod_{i=1}^k \Sigma_{i,i}$ is the original volume of $\R$. 
In other words, right multiplication by $\G$ increases the volume of $\R$ by at most $\norm{\H}_2^k$ in the embedding. 
We now bound $\norm{\H}_2^k$ using the following.
\begin{lemma}[Corollary 35 of \cite{vershynin2010introduction}]
\lemlab{lem:gauss:matrix:spectral}
Let $\H$ be a matrix of size $k\times r$ whose entries are independent standard normal random variables. 
Then for every $t\geq 0$, it follows that $\sigma_{\max}\leq \sqrt{k}+\sqrt{r}+t$ with probability at least $1-2\exp(-t^2/2)$, where $\sigma_{\max}$ is the largest singular value of $\H$.
\end{lemma}

By \lemref{lem:gauss:matrix:spectral} and the fact that entries of $\H$ have variance $\frac{1}{r}$, we have that $\norm{\H}_2 \leq 1 + \sqrt{k/r} + t/\sqrt{r} \leq 2 + t/\sqrt{r}$ with probability at least $1-2e^{-t^2/2}$. 
Equivalently, $\norm{\H}_2 \leq 2 + s$ with probability at least $1-2^{-s^2 r/2}$. 
Setting $s = \sqrt{2Ck}$ and using $r = (\log n)/C$, we have that the volume of $\R$ increases by at most a $(\sqrt{2Ck}+2)^k$ factor after the embedding, with probability at least $1-2^{-2Ck (\log n)/(2C)} = 1-n^{-k}$.
\end{proof}

Thus we can union bound over all subsets of size $k$ of the $n$ rows of $\A$, to argue that with high probability, none of them will have a volume increase by more than a $(\sqrt{4Ck} + 2)^k$ factor. 
This completes the proof of \lemref{volmax-tradeoff}.
\section{Volume Maximization Lower Bounds}
\seclab{sec:vm}
In this section, we complement our adaptive sampling based volume maximization algorithms, i.e., \thmref{thm:vm}, with lower bounds on turnstile streams that are tight up to lower order terms. 
Our lower bounds hold even for multiple passes through the turnstile stream. 
Additionally, we give a lower bound for volume maximization in the random order row-arrival model that is competitive with the algorithms in \secref{sec:vm:ra}. 

\subsection{Turnstile Streams}
We first consider lower bounds for turnstile streams. 
In the Gap $\ell_\infty$ problem, Alice and Bob are given vectors $\x$ and $\y$ respectively with $x,y\in[0,m]^n$ for some $m>0$ and the promise that either $|x_i-y_i|\le 1$ for all $i\in[n]$ or there exists some $i\in[n]$ such that $|x_i-y_i|=m$. 
The goal is for Alice and Bob to perform some communication protocol to decide whether there exists an index $i\in[n]$ such that $|x_i-y_i|=m$, possibly over multiple rounds of communication. 
To succeed with probability $\frac{8}{9}$, Alice and Bob must use at least $\Omega\left(\frac{n}{m^2}\right)$ bits of total communication, even if they can communicate over multiple rounds. 
\begin{theorem}
\thmlab{thm:cc:gap:linfty}
\cite{Bar-YossefJKS04}
Any protocol that solves the Gap $\ell_\infty$ problem with probability at least $\frac{8}{9}$ requires $\Omega\left(\frac{n}{m^2}\right)$ total bits of communication.
\end{theorem}
We first reduce an instance of the Gap $\ell_\infty$ problem to giving an $\alpha$-approximation to the volume maximization problem when $k=d=1$. 
\begin{theorem}
Any $p$-pass turnstile streaming algorithm that gives an $\alpha$-approximation to the volume maximization problem requires $\Omega\left(\frac{n}{p\alpha^2}\right)$ bits of space. 
\end{theorem}
\begin{proof}
Let $\alg$ be a turnstile streaming algorithm that provides an $\alpha$-approximation to the volume maximization problem. 
Let $d=1$ and $k=1$ in the volume maximization problem and $\M\in\mathbb{R}^{n\times 1}$ be the underlying matrix. 
Given an instance of the Gap $\ell_\infty$ problem with $m=\alpha+1$, suppose Alice has vector $\x$ and Bob has vector $\y$. 
Alice creates a stream with the coordinates of $\x$ so that $\M=\x$ at the end of the stream. 
Alice then passes the state of the algorithm to Bob, who updates the stream with the coordinates of $\y$ so that $\M=\x-\y$ at the end of the stream. 
Note that if $|x_i-y_i|\le 1$ for all $i\in[n]$, then the maximum possible volume is $1$ (in this case the volume is just the norm as $k=1$), whereas if $|x_i-y_i|=m=\alpha+1$ for some $i\in[n]$, then the maximum volume is equal to $\alpha+1$. 
Since $\alg$ is an $\alpha$-approximation, $\alg$ can differentiate between these two cases and solve the Gap $\ell_\infty$ problem. 
Thus by \thmref{thm:cc:gap:linfty}, $\alg$ uses $\Omega\left(\frac{n}{\alpha^2}\right)$ bits of space over the $p$ passes and hence at least $\Omega\left(\frac{n}{p\alpha^2}\right)$ bits of space.
\end{proof}

\begin{corollary}
\corlab{thm:turnstile:vm}
Any $p$-pass turnstile streaming algorithm that gives an $\alpha^k$-approximation to the volume maximization problem requires $\Omega\left(\frac{n}{kp\alpha^2}\right)$ bits of space. 
\end{corollary}
\begin{proof}
We generalize the above construction to the case of $k=d>1$ for any value of $k>1$. 
Consider the same instance $\M\in \mathbb{R}^{n/k \times 1}$ as the above lemma but with Gap $\ell_\infty$ problem of size $\frac{n}{k}$ instead of $n$. 
Now we construct a new instance $\M'\in\mathbb{R}^{n\times k}$ as follows. 
For each row $i\in [n/k]$ and $j\in [k]$, let $\M'_{(i-1)k+j , j} = \M_{i,1}$ and let all the other entries be equal to $0$.
In words, Alice and Bob embed the problem $k$ times across the $d$ columns for $k=d$. 
Thus in one case the maximum possible volume is $1$, while in the other case the maximum volume is equal to $(\alpha+1)^k$. 
\end{proof}
\subsection{Row-Arrival Model}
We now present streaming lower bounds for the row-arrival model.  
We consider a version of the distributional set-disjointness communication problem $\DISJ_{n,d}$ in which Alice is given the set of vectors $U=\{u_1,\ldots,u_n\}$ and Bob is given the set of vectors $V=\{v_1,\ldots,v_n\}$. 
With probability $\frac{1}{4}$, $U$ and $V$ are chosen uniformly at random among all instances with the following properties:
\begin{itemize}
\item Any vector in $U\cup V$ is in $\{0,1\}^d$ and moreover its weight is exactly $\frac{d}{2}$,
\item $U\cap V$ is non-empty.
\end{itemize}
This forms the NO case. 
Otherwise with probability $\frac{3}{4}$, $U$ and $V$ are chosen uniformly at random among all instances with the following properties:
\begin{itemize}
\item Any vector in $U\cup V$ is in $\{0,1\}^d$ and moreover its weight is exactly $\frac{d}{2}$,
\item $U\cap V = \emptyset$.
\end{itemize}
This forms the YES case. 
The goal is for Alice and Bob to perform some communication protocol to decide whether the instance is a YES or a NO instance, i.e., whether $U\cap V=\emptyset$, possibly over multiple rounds of communication. 

The following result, originally due to Razborov~\cite{Razborov92} and generalized by others~\cite{KalyanasundaramS92,WoodruffZ12}, lower bounds the communication complexity of any randomized protocol that solves $\DISJ_{n,d}$ with probability at least $\frac{7}{8}$, even given multiple rounds of communication. 
\begin{theorem}
\thmlab{thm:cc:set-disjoint}
\cite{Razborov92, KalyanasundaramS92, WoodruffZ12}
Any protocol for $\DISJ_{n,d}$ that fails with probability at most $\frac{1}{6}$ requires $\Omega(n)$ bits of total communication.
\end{theorem}
We first reduce an instance of the distributional set-disjointness problem to giving a $C^k$ approximation to the volume maximization problem in the row-arrival model when the order of the stream can be adversarial. 
\begin{theorem}
\thmlab{thm:cc:multi:set-disjoint}
For constant $p$ and $C=\frac{16}{15}$, any $p$-pass streaming algorithm that outputs a $C^k$ approximation to the $(2k)$-volume maximization problem in the row-arrival model with probability at least $\frac{8}{9}$ requires $\Omega(n)$ bits of space. 
\end{theorem}
\begin{proof}
Suppose Alice and Bob have an instance of $\DISJ_{n/2k,d/k}$, so that Alice has a set $U=\{u_1,\ldots,u_{n/2k}\}$ of vectors of $\{0,1\}^{d/k}\subset\mathbb{R}^{d/k}$ and Bob has a set $V=\{v_1,\ldots,v_{n/2k}\}$ of vectors of $\{0,1\}^{d/k}\subset\mathbb{R}^{d/k}$. 
Alice creates the matrix $\A\in\mathbb{R}^{\frac{n}{2k}\times\frac{d}{k}}$ by setting row $r\in\left[\frac{n}{2k}\right]$ of $\A$ to be precisely $u_r$. 
Alice then creates a $\frac{n}{2}\times d$ block diagonal matrix $\M_A$ to be the direct sum $\A\oplus\A\oplus\ldots\oplus\A$, where there are $k$ terms in the direct sum.  

For each $r\in\left[\frac{n}{2k}\right]$, Bob takes vector $v_r$ and creates a new vector $w_r$ by setting $w_r$ to be the complement of $v_r$, so that $w_r$ is the unique binary vector with weight $\frac{k}{2}$, but $\langle w_r,v_r\rangle=0$.
Bob then creates the matrix $\B\in\mathbb{R}^{\frac{n}{2k}\times\frac{d}{k}}$ by setting row $r\in\left[\frac{n}{2k}\right]$ of $\B$ to be precisely $w_r$. 
Bob also creates a $\frac{n}{2}\times d$ block diagonal matrix $\M_B$ to be the direct sum $\B\oplus\B\oplus\ldots\oplus\B$, where there are $k$ terms in the sum.  
Finally, define $\M\in\mathbb{R}^{n\times d}$ to be the matrix $\M_A$ stacked on top of $\M_B$ so that
\[
\M_A=\begin{bmatrix}
\A & \bzero & \ldots & \bzero\\
\bzero & \A & \ldots & \bzero\\
\vdots & \vdots & \ddots & \vdots\\
\bzero & \bzero & \ldots & \A
\end{bmatrix},\qquad
\M_B=\begin{bmatrix}
\B & \bzero & \ldots & \bzero\\
\bzero & \B & \ldots & \bzero\\
\vdots & \vdots & \ddots & \vdots\\
\bzero & \bzero & \ldots & \B
\end{bmatrix},\qquad
\M=\begin{bmatrix}
\M_A\\
\M_B
\end{bmatrix}
.
\]
Let $C=\frac{16}{15}$ and suppose there exists a $p$-pass streaming algorithm $\alg$ that computes a $C^k$-approximation to the $2k$-volume maximization problem with probability at least $1-\frac{1}{9}$ while using $o\left(\frac{n}{p}\right)$ space. 
Let $\mathcal{E}_1$ denote the event that $\alg$ correctly computes a $C^k$ approximation to the $2k$-volume maximization problem, which by the statement of the theorem holds with probability $8/9$. 
We claim that if $\mathcal{E}_1$ occurs, then Alice and Bob can use $\alg$ to construct a $p$ round communication protocol that solves $\DISJ_{n/2k,d/k}$ with high probability using $o(n)$ total communication, which contradicts the $\Omega\left(n\right)$ communication complexity of solving $\DISJ_{n/2k,d/k}$ for constant $k$.

Alice can create a stream $S$ by inserting the rows of $\M_A$ into the stream, since Alice has knowledge of the rows $\{u_r\}_{r\in[n/2k]}$. 
Alice can run $\alg$ on this stream $S$ and then pass the state of the algorithm to Bob, who appends the rows of $\M_B$ onto the stream $S$ and runs $\alg$ on this portion of the stream, starting with the state passed from Alice. 
Bob then passes the state of the algorithm back to Alice, completing both a single communication round as well as a single pass of $\alg$ through $S$. 
Alice and Bob can repeatedly pass the state of the algorithm between each other, to emulate passes over the stream. 
Thus after $p$ rounds of communication, $\alg$ will have completed $p$ passes over $S$ and output an approximation $\hat{\D}$ to the $2k$-volume maximization problem. 

We first claim that in a NO instance of $\DISJ_{n/2k,d/k}$, then with high probability, the optimal solution $\D$ to the $2k$-volume maximization problem contains two orthogonal rows $a$ and $b$ whose nonzero entries are between columns $\frac{(i-1)d}{k}+1$ and $\frac{id}{k}$ for each $i\in[k]$. 
In a NO instance, when Alice embeds vector $v$ into a row $a$, Bob embeds the complement of $v$ into a row $b$. 
Hence, there are $2k$ orthogonal vectors in the NO case, so the volume of the parallelpiped spanned by $2k$ vectors is maximized with the choice of the $2k$ orthogonal vectors, in which case the determinant is $\left(\frac{d}{2k}\right)^{2k}$. 

To analyze the YES instance, we first let $\mathcal{E}_2$ denote the event that there exist two orthogonal rows $a$ and $b$ in $\M$ whose nonzero entries are between columns $\frac{(i-1)d}{k}+1$ and $\frac{id}{k}$ for some $i\in[k]$ and either both $a\in\M_A$ and $b\in\M_B$ or $a\in\M_B$ and $b\in\M_A$. 
In other words, rows $a$ and $b$ were inserted by different people. 
For a fixed $i\in[k]$, the probability that rows $a$ and $b$ were both inserted by Alice is the probability that two vectors among $\frac{n}{2k}$ vectors of $\mathbb{R}^{\frac{d}{k}}$ with weight $\frac{d}{2k}$ are orthogonal. 
By symmetric reasoning with Bob inserting both vectors and removing the instances where Alice and Bob have ``random'' orthogonal vectors, we note that
\begin{align*}
\PPr{\neg\mathcal{E}_2}&\ge 1-2\binom{n/2k}{2}\frac{1}{\binom{d/k}{d/2k}-n}\ge 1-\frac{n^2}{4k^2}\frac{1}{2^{d/2k}-n}.
\end{align*}
Note that since $\M_A$ and $\M_B$ are each direct sums of $k$ instances of $\A$ and $\B$, we do not need to take a union bound over all indices $i\in[k]$, though even such a union bound would still cause $\mathcal{E}_2$ to hold with low probability. 

Moreover, by symmetry the volume of the spanning parallelpiped is maximized when the $2k$ rows are $k$ direct sums of two rows\footnote{Even in the NO case, the maximizer of the determinant is the direct sum of $k$ terms with two rows.}. 
Let $\mathcal{E}_3$ be the event all the rows $u$ and $v$ among the rows in both $\A$ and $\B$ intersect by more than $\frac{d}{8k}$ coordinates. 
We claim that $\mathcal{E}_3$ holds with high probability in a YES instance. 
If that were true, then the maximum volume in the YES case is less than 
\[\left(\left(\frac{d}{2k}\right)^2-\left(\frac{d}{8k}\right)^2\right)^k=\left(\frac{15d^2}{64k^2}\right)^k,\]
which would show a separation between the YES and NO instances, since the volume in the NO case is $\left(\frac{d^2}{4k^2}\right)^k$. 
Thus for $C=\frac{16}{15}$ and conditioning on $\mathcal{E}_1$, $\neg\mathcal{E}_2$ and $\mathcal{E}_3$, Alice and Bob can use any $C^k$ approximation algorithm to the volume maximization problem differentiate between a YES instance and a NO instance of $\DISJ_{n/2k,d/k}$. 

It remains to prove the claim that $\mathcal{E}_3$ holds with high probability in a YES instance. 
\begin{claim}
In a YES instance, all the rows given to Alice and Bob intersect by more than $\frac{d}{8k}$ coordinates with probability at least $1-\frac{n^2}{k^2}\sqrt{\frac{\gamma^2 d}{k}}\frac{8^{d/3k}}{9^{3d/8k}}$, for some fixed constant $\gamma$. 
That is,
\[\PPr{\mathcal{E}_3}\ge 1-\frac{n^2}{k^2}\sqrt{\frac{\gamma^2 d}{k}}\frac{8^{d/3k}}{9^{3d/8k}}.\]
\end{claim}
\begin{proof}
For a fixed pair of vectors $a$ and $b$, the probability that $a$ and $b$ intersect in at most $\frac{d}{8k}$ coordinates without Alice and Bob having ``random'' orthogonal vectors is at most
\begin{equation}
\eqnlab{eqn:vectors:distance:prob}
\frac{d}{8k}\frac{\binom{d/2k}{d/8k}\binom{d/2k}{3d/8k}}{\binom{d/k}{d/2k}-n}>\frac{d}{16k}\frac{\binom{d/2k}{d/8k}\binom{d/2k}{3d/8k}}{\binom{d/k}{d/2k}},
\end{equation}
for sufficiently large $d/k$. 
By Stirling's approximation, there exists a fixed constant $\gamma$ such that \eqnref{eqn:vectors:distance:prob} is at most
\[\frac{d}{k}\sqrt{\frac{\gamma k}{d}}\frac{(d/2k)^{d/2k}(d/2k)^{d/2k}(d/2k)^{d/2k}(d/2k)^{d/2k}}{(d/k)^{d/k}(d/8k)^{d/8k}(3d/8k)^{3d/8k}(3d/8k)^{3d/8k}(d/8k)^{d/8k}}= \sqrt{\frac{\gamma^2 d}{k}}\frac{8^{d/3k}}{9^{3d/8k}}.\]
Taking a union bound over at most $\frac{n^2}{k^2}$ pairs of vectors, the probability that there exist two rows that intersect by at most $\frac{d}{8k}$ coordinates is at most $\frac{n^2}{k^2}\sqrt{\frac{\gamma^2 d}{k}}\frac{8^{d/3k}}{9^{3d/8k}}$.
\end{proof}
For $n>d$ and $d=\Theta(k\log\gamma n)$ with a sufficiently large constant, then $\PPr{\mathcal{E}_1}\ge\frac{8}{9}$, $\PPr{\neg\mathcal{E}_2}\ge 1-\frac{1}{\poly(n)}$, and $\PPr{\mathcal{E}_3}\ge 1-\frac{1}{\poly(n)}$. 
Thus Alice and Bob can use $\alg$ to decide $\DISJ_{n/2k,d/k}$ with probability at least
\[\PPr{\mathcal{E}_1\cap\neg\mathcal{E}_2\cap\mathcal{E}_3}>1-\frac{1}{8}.\]
Thus if $\alg$ uses $o\left(\frac{n}{p}\right)$ space per pass over $p$ passes, then the total communication between Alice and Bob is $o(n)$, which contradicts \thmref{thm:cc:set-disjoint}. 
It follows that any $C^k$ approximation algorithm to the volume maximization problem that succeeds with probability at least $1-\frac{1}{9}$ requires $\Omega(n)$ space for constant $k$. 
\end{proof}
Recall that for problems that are invariant to the permutation of the rows of the input matrix $\A$, once the entries of $\A$ are chosen, an arbitrary permutation of the rows of $\A$ is chosen uniformly at random, and the rows of that permutation constitute the stream in the random order row-arrival model. 
\begin{corollary}
\corlab{cor:ra:turnstile}
For $C=\frac{16}{15}$, any one-pass streaming algorithm that outputs a $C^k$ approximation to the $2k$-volume maximization problem in the random order row-arrival model with probability at least $\frac{63}{64}$ requires $\Omega(n)$ bits of space. 
\end{corollary}
\begin{proof}
First, observe that Alice and Bob construct matrix $\M_A$ and $\M_B$ from rows drawn uniformly at random. 
Thus in the NO case, the distribution of the matrix $\M_A$ and $\M_B$ follows the same distribution as when the rows of $\M$ arrive uniformly at random. 
In the YES case in the above model, the two orthogonal vectors must be in separate halves of the matrix $\M$. 
Namely, one vector is in $\M_A$ and one vector is in $\M_B$. 
In the random order model, the two orthogonal vectors are in separate halves with probability at least $\frac{1}{2}$. 
Since the YES case occurs with probability $\frac{1}{4}$, the total variation distance between the distribution of the rows in the random order model and the above distribution is $\frac{1}{8}$. 
Hence for a $\frac{7}{8}$ fraction of the inputs, Alice and Bob has the same distribution as that of \thmref{thm:cc:multi:set-disjoint}. 

In that case, any one-pass streaming algorithm $\alg$ that outputs a $C^k$ approximation to the $2k$-volume maximization problem with probability at least $\frac{63}{64}$ can decide between a YES instance and a NO instance of $\DISJ_{n/2k,d/k}$ for sufficiently large $n$ and $d$ with probability at least $\frac{31}{32}$ by the same argument as \thmref{thm:cc:multi:set-disjoint}. 
Hence, the total probability of failure of the protocol is at most $\frac{1}{8}+\frac{1}{32}\le\frac{1}{6}$ and so by \thmref{thm:cc:set-disjoint}, $\alg$ requires $\Omega(n)$ space.
\end{proof}

\section*{Acknowledgements}
D. Woodruff acknowledges support in part from the National Science Foundation under Grant No. CCF-1815840.

\def\shortbib{0}
\bibliographystyle{alpha}
\bibliography{references}

\appendix
\section{Noisy Distance Sampling}
\applab{app:appendix}
\subsection{$L_{1,2}$ Sampler}
Recall that for a matrix $\A\in\mathbb{R}^{n\times d}$, we define the $L_{p,q}$ norm of $\A$ by
\[\norm{\A}_{p,q}=\left(\sum_{i=1}^n\left(\sum_{j=1}^d |A_{i,j}|^q\right)^{\frac{p}{q}}\right)^{\frac{1}{p}}.\]
In this section, we describe an algorithm for sampling rows of a matrix $\A\P$ with probability proportional to $\norm{\A\P}_2$, which we call $L_{1,2}$ sampling.  
By comparison, in \secref{sec:l2:sampler} we sampled rows of $\A\P$ with probability proportional to $\norm{\A\P}_2^2$, which can be seen as $L_{2,2}$ sampling. 

Before describing our general $L_{1,2}$ sampler, we need a subroutine similar to \ams{} for estimating $\norm{\A\P}_{1,2}$, when the data stream updates entries of $\A$ and query access to $\P$ is only given in post-processing. 
We first describe a turnstile streaming algorithm of~\cite{AndoniBIW09} that can be used to compute a constant factor approximation to $\norm{\A}_{1,2}$ and then we show that it can be modified to approximate $\norm{\A\P}_{1,2}$ due its nature of being a linear sketch. 
For each $j$, define the level sets $S_j$ by $\left\{i\in[n]\,:\,\frac{\norm{\A}_{1,2}}{2^{j+1}}<\norm{\A_i}_2\le\frac{\norm{\A}_{1,2}}{2^{j}}\right\}$. 
The algorithm of~\cite{AndoniBIW09} approximates the number of rows in each level set $S_j$ by first implicitly subsampling rows at different rates. 
The rows that are sampled at each rate then form a \emph{level} and the rows in a particular level are then aggregated across a number of buckets. 
The norms of the aggregates across each bucket are then computed and by rescaling the number of aggregates that are in each level set, we obtain an accurate estimate of the sizes of the level sets. 
The sizes of the level sets are then used to output a good approximation to $\norm{\A}_{1,2}$. 

Crucially, the aggregates of the rows in the algorithm of~\cite{AndoniBIW09} is a linear combination of the rows. 
Hence by taking the aggregates and multiplying by $\P$ after the stream ends, we obtain aggregates of the rows of $\A\P$, which can then be used to estimate the sizes of the level sets of $\norm{\A\P}_{1,2}$. 
The algorithm of~\cite{AndoniBIW09} uses $d\,\polylog(n)$ space by storing aggregates of entire rows for each bucket across multiple levels. 
Thus, we have the following:
\begin{lemma}\cite{AndoniBIW09}
\lemlab{lem:estimator}
There exist a fixed constant $\xi>1$ and a one-pass turnstile streaming algorithm \estimator{} that takes updates to entries of a matrix $\A\in\mathbb{R}^{n\times d}$, as well as query access to post-processing matrices $\P\in\mathbb{R}^{d\times d}$ and $\M\in\mathbb{R}^{n\times d}$ that arrive after the stream, and outputs a quantity $\hat{F}$ such that $\norm{\A\P-\M}_{1,2}\le\hat{F}\le\xi\norm{\A\P-\M}_{1,2}$. 
The algorithm uses $d\,\polylog(n)$ bits of space and succeeds with high probability. 
\end{lemma}
Using the $L_{1,2}$ estimator, we can develop a $L_{1,2}$ sampler similar to our $\ell_2$ sampler. 
\begin{algorithm}[!htb]
\caption{Single $L_{1,2}$ Sampler}
\alglab{alg:l12:sampler}
\begin{algorithmic}[1]
\Require{Matrix $\A\in\mathbb{R}^{n\times d}$ that arrives as a stream $\A_1,\ldots,\A_n\in\mathbb{R}^d$, matrix $\P\in\mathbb{R}^{d\times d}$ that arrives after the stream, constant parameter $\eps>0$.}
\Ensure{Noisy row $\r$ of $\A\P$ sampled roughly proportional to the row norms of $\A\P$.}
\State{\textbf{Pre-processing Stage:}}
\State{$b\gets\Omega\left(\frac{1}{\eps^2}\right)$, $r\gets\Theta(\log n)$ with sufficiently large constants}
\State{For $i\in[n]$, generate independent scaling factors $t_i\in[0,1]$ uniformly at random.}
\State{Let $\B$ be the matrix consisting of rows $\B_i=\frac{1}{t_i}\A_i$.}
\State{Let $\estimator$ and $\ams$ track the $L_{1,2}$ norm of $\A\P$ and Frobenius norm of $\B\P$, respectively.}
\State{Let $\countsketch$ be an $r\times b$ table, where each entry is a vector $\mathbb{R}^d$.}
\State{\textbf{Streaming Stage:}}
\For{each row $\A_i$}
\State{Update $\countsketch$ with $\B_i=\frac{1}{t_i}\A_i$.}
\State{Update linear sketch $\estimator$ with $\A_i$.}
\State{Update linear sketch $\ams$ with $\B_i=\frac{1}{t_i}\A_i$.}
\EndFor
\State{\textbf{Processing $\P$ Stage:}}
\State{After the stream, obtain matrix $\P$.}
\State{Multiply each vector $\v$ in each entry of the $\countsketch$ table by $\P$: $\v\gets\v\P$.}
\State{Multiply each vector $\v$ in $\ams$ by $\P$: $\v\gets\v\P$.}
\State{Multiply each vector $\v$ in $\estimator$ by $\P$: $\v\gets\v\P$.}
\State{\textbf{Extraction Stage:}}
\State{Use $\estimator$ to compute $\widehat{F}$ with $\norm{\A\P}_{1,2}\le\widehat{F}\le\xi\norm{\A\P}_{1,2}$.}
\Comment{\lemref{lem:estimator}}
\State{Extract the $\frac{2}{\eps^2}$ (noisy) rows of $\B\P$ with the largest estimated norms by $\countsketch$.}
\State{Let $\M$ be the $\frac{2}{\eps^2}$-sparse matrix consisting of these top (noisy) rows.}
\State{Use $\ams$ to compute $\widehat{S}$ with $\norm{\B\P-\M}_F\le\widehat{S}\le 2\norm{\B\P-\M}_F$.}
\State{Let $\r_i$ be the (noisy) row in $\countsketch$ with the largest norm.}
\State{Let $C>0$ be some large constant so that the probability of failure is $\O{\frac{1}{n^{C/2}}}$.}
\If{$\widehat{S}>\frac{C\log n}{\eps}\widehat{F}$ or $\norm{\r_i}_2<\frac{C\log n}{\eps^2}\widehat{F}$}
\State{\Return FAIL.}
\Else
\State{\Return $\r=t_i\r_i$.}
\EndIf
\end{algorithmic}
\end{algorithm}

We first show the probability that the tail is too large, i.e., $\widehat{S}>\frac{C\log n}{\eps}\widehat{F}$, is independent of the index $i$ and the value of $t_i$. 
The proof is almost verbatim to \lemref{lem:tail:failure}, but the thresholds now depend on $\norm{\A\P}_{1,2}$ rather than $\norm{\A\P}_F$.  
\begin{lemma}
\lemlab{lem:l12:tail:failure}
For each $j\in[n]$ and value of $t_j$,
\[\PPr{\widehat{S}>\frac{C\log n}{\eps}\widehat{F}}=\O{\eps}+\frac{1}{\poly(n)}.\]
\end{lemma}
\begin{proof}
Let $\xi$ be defined as in \lemref{lem:estimator}. 
We first define the event $\mathcal{E}_1$ as when the following three inequalities hold:
\begin{enumerate}
\item
$\norm{\A\P}_{1,2}\le\widehat{F}\le\xi\norm{\A\P}_{1,2}$
\item
$\norm{\B\P-\M}_F\le\widehat{S}\le 2\norm{\B\P-\M}_F$
\item
$\norm{(\B\P)_{\taileps}}_F\le\norm{\B\P-\M}_F\le 2\norm{(\B\P)_{\taileps}}_F$
\end{enumerate}
Let $j\in[n]$ be a fixed index and $t_j=t$ be a fixed uniform random scaling variable. 
$\mathcal{E}_1$ holds with high probability by \lemref{lem:countsketch} and \lemref{lem:estimator}. 
We bound the probability that $2\xi\norm{(\B\P)_{\taileps}}_F>\frac{C\log n}{\eps}\norm{\A\P}_{1,2}$, which must hold if $\widehat{S}>\frac{C\log n}{\eps}\widehat{F}$. 

Let $U=\norm{\A\P}_{1,2}$ and for each $i\in[n]$, define the indicator variable $y_i=1$ if $\norm{\B_i\P}_2>U$ and $y_i=0$ otherwise. 
For each $i\in[n]$, define the scaled indicator variable $z_i=\frac{1}{U}\norm{\B_i\P}_2(1-y_i)$. 
Observe that $z_i\in[0,1]$ represents a scaled contribution of the rows that are not heavy. 
Let $Y=\sum_{i\neq j} y_i$ and $Z=\sum_{i\neq j} z_i$. 
Define the matrix $\W\in\mathbb{R}^{n\times d}$ so that for each $i\in[n]$, its row $i$ satisfies $\W_i=\B_i\P$ if $y_i=1$ and otherwise if $y_i=0$, then $\W_i$ is the row of all zeros. 
Hence, $\W$ has at most $Y+1$ nonzero rows and $UZ=\norm{\B\P-\W}_{1,2}$. 

If there are not too many heavy rows in $\W$, i.e., $Y<\frac{2}{\eps^2}$, then $\norm{(\B\P)_{\taileps}}_F\le UZ$ since the rows of $\W$ that are all zeros contain the tail of $\B\P$ and the Frobenius norm is at most the $L_{1,2}$ norm. 
Let $\mathcal{E}_2$ denote the event that $Y\ge\frac{2}{\eps^2}$ and $\mathcal{E}_3$ denote the event that $Z\ge\frac{C\log n}{2\xi U\eps}\norm{\A\P}_{1,2}=\frac{C\log n}{2\xi\eps}$. 
If we bound the probability of the events $\mathcal{E}_2$ and $\mathcal{E}_3$ by $\O{\eps}$, then $2\xi\norm{(\B\P)_{\taileps}}_F\le\frac{C\log n}{\eps}\norm{\A\P}_{1,2}$ with probability at least $1-\O{\eps}$, conditioned on $\mathcal{E}_1$. 

To bound $\mathcal{E}_2$, observe that $\Ex{y_i}=\frac{\norm{\A_i\P}_2}{U}$ so that $\Ex{Y}\le 1$ since $U=\norm{\A\P}_{1,2}$. 
Thus for sufficiently large $n$, $\PPr{\mathcal{E}_2}=\O{\eps}$, by Markov's inequality. 

To bound $\PPr{\mathcal{E}_3}$, observe that $z_i=\frac{1}{U}\norm{\B_i\P}_2(1-y_i)$ implies $z_i>0$ only if $y_i=0$, i.e., $\norm{\B_i\P}_2\le U$. 
Since $\B_i\P=\frac{\A_i\P}{t_i}$, then $z_i>0$ only for $t_i\ge\frac{\norm{\A_i\P}_2}{\norm{\A\P}_{1,2}}$. 
Therefore, 
\begin{align*}
\Ex{z_i}\le\int_{\norm{\A_i\P}_2/\norm{\A\P}_{1,2}}^1 z_i\,dt_i=\int_{\norm{\A_i\P}_2/\norm{\A\P}_{1,2}}^1\frac{1}{t_i}\frac{1}{U}\norm{\A_i\P}_2\,dt_i\le C\log n\frac{\norm{\A_i\P}_2}{\norm{\A\P}_{1,2}},
\end{align*}
conditioned on $t_i\ge\frac{1}{\poly(n)}$. 
Hence $\Ex{Z}\le C\log n$ so $\PPr{\mathcal{E}_3}=\PPr{Z>\frac{C\log n}{2\xi\eps}}=\O{\eps}$ by Markov's inequality. 
Therefore, the failure events $\neg\mathcal{E}_1\vee\mathcal{E}_2\vee\mathcal{E}_3$ occur with probability $\O{\eps}+\frac{1}{\poly(n)}$, and the claim follows.
\end{proof}

\begin{lemma}
\lemlab{lem:l12:sampling:prob}
Conditioned on a fixed value of $\widehat{F}$, the probability that \algref{alg:l12:sampler} outputs (noisy) row $i$ is $\left(1\pm\O{\eps}\right)\frac{\norm{\A_i\P}_2}{\widehat{F}}+\frac{1}{\poly(n)}$. 
\end{lemma}
\begin{proof}
We first define $\mathcal{E}$ to be the event that $t_i<\frac{\eps^2\norm{\A_i\P}_2}{(C\log n)\widehat{F}}$. 
Note that $\PPr{\mathcal{E}}=\frac{\eps^2\norm{\A_i\P}_2}{(C\log n)\widehat{F}}$. 
Next, we define $\mathcal{E}_1$ to be the event that $\countsketch$, $\ams$, or $\estimator$ fails and note that $\PPr{\mathcal{E}_1}=\frac{1}{\poly(n)}$ by \lemref{lem:countsketch}, \lemref{lem:ams}, and \lemref{lem:estimator}. 
We then define $\mathcal{E}_2$ to be the event that $\widehat{S}>\frac{C\log n}{\eps}\widehat{F}$ and note that $\PPr{\mathcal{E}_2}=\O{\eps}$ by \lemref{lem:l12:tail:failure}. 
Finally, we let $\mathcal{E}_3$ be the event that the CountSketch data structure observes multiple rows $\B_j\P$ exceeding the threshold and $\mathcal{E}_4$ be the event that $\norm{\B_i\P}_2$ exceeds the threshold but is not reported due to noise in the CountSketch data structure. 

Observe that a row $j$ is close enough to the threshold if $\norm{\B_j\P}_2\ge\frac{C\log n}{\eps^2}\widehat{F}-(C\log n)\widehat{F}$, which occurs with probability at most $\O{\frac{\eps\norm{\A_j\P}_2}{\widehat{F}}}$. 
Taking a union bound over all $n$ rows, we have $\PPr{\mathcal{E}_3}=\O{\eps}$. 

To analyze the probability of $\mathcal{E}_4$, we first condition on $\neg\mathcal{E}_1$ and $\neg\mathcal{E}_2$, so that we have $\norm{\B\P-\M}_F\le\widehat{S}$ and $\widehat{S}\le\frac{C\log n}{\eps}\widehat{F}$. 
Thus by \lemref{lem:countsketch}, 
\[\left|\norm{\B_i\P}_2-\norm{\widehat{\B_i\P}}_2\right|\le\eps\norm{\B\P_{\taileps}}_F\le\eps\norm{\B\P-\M}_F\le\eps\widehat{S}\le(C\log n)\widehat{F}.\]
Hence, $\mathcal{E}_4$ can only occur for 
\[\frac{C\log n}{\eps^2}\widehat{F}\le\norm{\B_i\P}_2\le\frac{C\log n}{\eps^2}\widehat{F}+(C\log n)\widehat{F},\]
which occurs with probability at most $\frac{\eps^4\norm{\A_i\P}_2}{(C\log n)\widehat{F}}$.

In summary if $\mathcal{E}$ occurs, then the sampler should output (noisy) row $\A_i\P$ but may fail to do so because of any of the events $\mathcal{E}_1$, $\mathcal{E}_2$, $\mathcal{E}_3$, or $\mathcal{E}_4$. 
We have $\PPr{\mathcal{E}_2\vee\mathcal{E}_3\,|\,\mathcal{E}}=\O{\eps}$ and $\PPr{\mathcal{E}_4}=\frac{\eps^4\norm{\A_i\P}_2}{(C\log n)\widehat{F}}$ so that $\PPr{\mathcal{E}_4\,|\,\mathcal{E}}=\O{\eps^2}$. 
Since $\PPr{\mathcal{E}_1}=\frac{1}{\poly(n)}$, then each $\A_i\P$ is output with probability $(1+\O{\eps})\frac{\norm{\A_i\P}_2}{\widehat{F}}$. 

Moreover by \lemref{lem:countsketch:tail}, we have that $\left|\norm{\B_i\P}_2-\norm{\widehat{\B_i\P}}_2\right|\le(C\log n)\widehat{F}$ and $\norm{\widehat{\B_i\P}}_2\ge\frac{C\log n}{\eps^2}\widehat{F}$. 
Thus, $\norm{\widehat{\B_i\P}}_2$ is a $(1+\O{\eps})$ approximation to $\norm{\B_i\P}_2$ and similarly, $t_i\norm{\widehat{\B_i\P}}_2$ is within $(1+\O{\eps})$ of $\norm{\A_i\P}_2$. 
\end{proof}	
We now provide the full guarantees of the $L_{1,2}$ sampler. 
\begin{theorem}
\thmlab{thm:l12:sampling}
Given $\eps>0$, there exists a one-pass streaming algorithm that takes rows of a matrix $\A\in\mathbb{R}^{n\times d}$ as a data stream and a matrix $\P\in\mathbb{R}^{d\times d}$ after the stream, and outputs (noisy) row $i$ of $\A\P$ with probability $\left(1\pm\O{\eps}\right)\frac{\norm{\A_i\P}_2}{\norm{\A\P}_{1,2}}+\frac{1}{\poly(n)}$. 
The algorithm uses $\O{d\,\poly\left(\frac{1}{\eps},\log n\right)}$ bits of space and succeeds with high probability.  
\end{theorem}
\begin{proof}
From \lemref{lem:l12:sampling:prob} and the fact that $\norm{\A\P}_{1,2}\le\widehat{F}\le\xi\norm{\A\P}_{1,2}$ with high probability by \lemref{lem:estimator}, then it follows that each row $\A_i\P$ is sampled with probability $\left(1+\eps\right)\frac{\norm{\A_i\P}_2}{\norm{\A\P}_{1,2}}+\frac{1}{\poly(n)}$, conditioned on the sampler succeeding. 
The probability of the sampler succeeds is $\Theta\left(\frac{\eps^2}{C\log n}\right)$, then the sampler can be repeated $\poly\left(\frac{1}{\eps},\log n\right)$ times to obtain probability of success at least $1-\frac{1}{\poly(n)}$. 
Since each instance of $\ams$, $\estimator$, and $\countsketch$ uses $\O{d\,\poly\left(\frac{1}{\eps},\log n\right)}$ bits of space, then the total space complexity follows. 
\end{proof}

\subsection{Noisy Adaptive Distance Sampling}
\seclab{sec:l12:adaptive}
Our algorithm for noisy adaptive distance sampling, given in \algref{alg:l12:noisy:adaptive}, is similar to \secref{sec:noisy:adaptive}, except it uses the $L_{1,2}$ sampling primitive of \thmref{thm:l12:sampling} instead of the $L_{2,2}$ sampler. 
\begin{algorithm}[!htb]
\caption{Noisy Adaptive Sampler}
\alglab{alg:l12:noisy:adaptive}
\begin{algorithmic}[1]
\Require{Matrix $\A\in\mathbb{R}^{n\times d}$ that arrives as a stream $\A_1,\ldots,\A_n\in\mathbb{R}^d$, parameter $k$ for number of sampled rows, constant parameter $\eps>0$.}
\Ensure{$k$ Noisy and projected rows of $\A$.}
\State{Create instances $\alg_1,\ldots,\alg_k$ of the $L_{1,2}$ sampler of \algref{alg:l12:sampler} where the number of buckets $b=\Theta\left(\frac{\log^4 n}{\eps^2}\right)$ is sufficiently large.}
\State{Let $\M$ be empty $0\times d$ matrix.}
\State{\textbf{Streaming Stage:}}
\For{each row $\A_i$}
\State{Update each sketch $\alg_1,\ldots,\alg_k$}
\EndFor
\State{\textbf{Post-processing Stage:}}
\For{$j=1$ to $j=k$}
\State{Post-processing matrix $\P\gets\I-\M^\dagger\M$.}
\State{Update $\alg_j$ with post-processing matrix $\P$.}
\State{Let $\r_j$ be the noisy row output by $\alg_j$.}
\State{Append $\r_j$ to $\M$: $\M\gets\M\circ\r_j$.}
\EndFor
\State{\Return $\M$.}
\end{algorithmic}
\end{algorithm}

We first bound the norm of the perturbation of the sampled row at each instance. 
\begin{lemma}
\lemlab{lem:l12:orthogonal:noise}
Given a matrix $\A\in\mathbb{R}^{n\times d}$ and a matrix $\P\in\mathbb{R}^{d\times d}$, as defined in Line 8 and round $i\le k$, of \algref{alg:l12:sampler}, suppose index $j\in[n]$ is sampled (in round $i$). 
Then with high probability, the sampled (noisy) row $\r_i$ satisfies $\r_i=\A_j\P+\v_e$ with 
\[\norm{\v_e\Q}_2\le\frac{\eps^3}{C\log n}\frac{\norm{\A\P\Q}_{1,2}}{\norm{\A\P}_{1,2}}\norm{\A_j\P}_2,\]
for any projection matrix $\Q\in\mathbb{R}^{d\times d}$. 
Hence, $\v_e$ is orthogonal to each noisy row $\r_y$, where $y\in[i-1]$. 
\end{lemma}
\begin{proof}
Let $\Q\in\mathbb{R}^{d\times d}$ be a projection matrix, $\B_x=\frac{\A_x}{t_x}$ be the rescaled row of $\A_x$ for each $x\in[n]$, and $\B\in\mathbb{R}^{n\times d}$ be the rescaled matrix of $\A$ so that row $x$ of $\B$ is $\B_x$ for $x\in[n]$. 
Let $\E$ be the noise in the bucket corresponding to the selected row $j$, so that the output vector is $\A_j+t_j\E$. 
Since $t_x\in[0,1]$ is selected uniformly at random for each $x\in[n]$, then for each integer $c\ge0$,
\[\PPr{\frac{\norm{\A_x\P\Q}_2}{t_x}\ge\frac{\norm{\A\P\Q}_{1,2}}{2^c}}\le\frac{2^c\norm{\A_x\P\Q}_2}{\norm{\A\P\Q}_{1,2}}.\]
Because $\B_x=\frac{\A_x}{t_x}$, then by linearity of expectation over $x\in[n]$, we can bound the expected size of each of the level sets $S_c:=\left\{x\in[n]\,:\,\frac{\norm{\A\P\Q}_{1,2}}{2^{c-1}}>\norm{\B_x\P\Q}_2\ge\frac{\norm{\A\P\Q}_{1,2}}{2^c}\right\}$ by $\Ex{|S_c|}\le\min(2^c,n)$. 
Thus $\PPr{\left|S_c\right|\le\min(2^{c+C}\log n,n)}\ge 1-\frac{1}{\poly(n)}$ by standard Chernoff bounds for appropriate constant $C$.  

We can now roughly bound the $L_{1,2}$ norm of $\B\P\Q$ by the norm of $\A\P\Q$ using a union bound over level sets $S_c$ for $0\le c\le\log n$ and upper bounding the norms of all the rows in level sets $S_c$ with $c>\log n$ by $\frac{\norm{\A\P\Q}_{1,2}}{n}$. 
That is,  
\[\PPr{\norm{\B\P\Q}_{1,2}\ge 2^C\cdot4\log^2 n\norm{\A\P\Q}_{1,2}}\le\frac{1}{\poly(n)}.\]
Hence with high probability, the total mass $\norm{\B\P\Q}_{1,2}$ distributed across the CountSketch table is $\O{\log^2 n\norm{\A\P\Q}_{1,2}}$. 

Using a CountSketch table with $b=\Theta\left(\frac{\log^4 n}{\eps^2}\right)$ buckets with sufficiently large constant to hash the rows of $\B\P\Q$, then \lemref{lem:countsketch} implies that the bucket corresponding to $\A_j$ has mass at most $\eps\norm{\A\P\Q}_F\le\eps\norm{\A\P\Q}_{1,2}$ in the subspace to which $\Q$ projects, i.e., $\norm{\E\Q}_2\le\eps\norm{\A\P\Q}_{1,2}$. 
Since row $j$ was sampled by \algref{alg:l12:sampler}, then $\norm{\B_j\P}_2\ge\frac{C\log n}{\eps^2}\norm{\A\P}_{1,2}$. 
Thus $t_j\le\frac{\eps^2}{C\log n}\frac{\norm{\A_j\P}_2}{\norm{\A\P}_{1,2}}$ since $\B_j=\frac{\A_j\P}{t_j}$, and moreover with high probability,
\[\norm{t_j\E\Q}_2\le\frac{\eps^3}{C\log n}\frac{\norm{\A\P\Q}_{1,2}}{\norm{\A\P}_{1,2}}\norm{\A_j\P}_2.\]
Since $\v_e\Q=t_j\E\Q$, then the claim follows. 
\end{proof}
We now bound the total variation distance between the distribution of sampled rows and the distribution of adaptive sampling with respect to distances to selected subspace. 
The proof is almost verbatim to \lemref{lem:adaptive:tvd}, except we now consider the probabilities with respect to the distances to the previous subspace, rather than the squared distances. 
We can still use the change of basis matrix in (\ref{eqn:cob:matrix}) to denote the perturbation in each round, where we set $\tau_i=\frac{\eps^3\sum_{a=1}^n\lambda_{a,i}}{\sum_{a=1}^n\sqrt{\sum_{b=1}^d\lambda_{a,b}^2}}$ though \lemref{lem:l12:orthogonal:noise} implies we could actually even set the scaling factor to $\frac{\eps^3}{C\log n}$ rather than just $\eps^2$. 
We can then bound $\left|\sqrt{\sum_{i=2}^d\zeta_{s,i}^2}-\sqrt{\sum_{i=2}^d\lambda_{s,i}^2}\right|$ from a bound on $\left|\sum_{i=2}^d\zeta_{s,i}^2-\sum_{i=2}^d\lambda_{s,i}^2\right|$. 
\begin{lemma}
\lemlab{lem:l12:adaptive:tvd}
Let $f(1)$ be the index of a noisy row $\r_1$ sampled in the first iteration of \algref{alg:l12:noisy:adaptive}. 
Let $\mathcal{P}_1$ be a process that projects away from $\A_{f(1)}$ and iteratively selects $k-1$ additional rows of $\A$ through adaptive sampling (with $p=1$).  
Let $\mathcal{P}_2$ be a process that projects away from $\r_1$ and iteratively selects $k-1$ additional rows of $\A$ through adaptive sampling (with $p=1$). 
Then for $\eps<\frac{1}{d}$, the total variation distance between the distributions of the $k$ indices output by $\mathcal{P}_1$ and $\mathcal{P}_2$ is $\O{k\eps}$. 
\end{lemma}
\begin{proof}
As in the proof of \lemref{lem:adaptive:tvd}, we consider the probability distributions induced by linearly independent vectors $\A_{f(1)},\ldots,\A_{f(t-1)}$ and by linearly independent vectors $\r_1,\A_{f(2)},\ldots,\A_{f(t-1)}$. 
Let $U=\{\u_1,\ldots,\u_d\}$ be an orthonormal basis for the row span of $\A$ such that $\{\u_1,\ldots,\u_s\}$ is a basis for the row span of $\{\A_{f(1)},\ldots,\A_{f(s)}\}$ for each $2\le s\le t-1$ and let $W=\{\w_1,\ldots,\w_d\}$ be an orthonormal basis for the row span of $\A$ such that $\{\w_1,\ldots,\w_s\}$ is an orthonormal basis that extends the row span of $\{\r_1,\A_{f(2)},\ldots,\A_{f(s)}\}$ for each $2\le s\le t-1$. 

\lemref{lem:l12:orthogonal:noise} then implies that $\r_1=\norm{\A_{f(1)}}_2\left(\u_1+\sum_{i=1}^d(\pm\O{\tau_i})\u_i\right)$, where $\tau_i=\frac{\eps^3\norm{\A\P_i}_{1,2}}{\norm{\A}_{1,2}}$ (in fact we could set $\tau_i$ to $\frac{\eps^3}{C\log n}\frac{\norm{\A\P_i}_{1,2}}{\norm{\A}_{1,2}}$ but we do not need this smaller value) with high probability, and $\P_i=\u_i^\dagger\u_i$ is the projection matrix onto $\u_i$. 
From the Gram-Schmidt process, the change of basis matrix $\B$ from $U$ to $W$ has the form (\ref{eqn:cob:matrix}). 

We can write each row $\A_s$ in terms of basis $U$ as $\A_s=\sum_{i=1}^d\lambda_{s,i}\u_i$ and in terms of basis $W$ as $\A_s=\sum_{i=1}^d\zeta_{s,i}\w_i$. 
Since we project away from $\A_{f(1)}$, we should have sampled $\A_s$ with probability $\frac{\norm{\A_s\Z_{t-1}}_2}{\norm{\A\Z_{t-1}}_{1,2}}=\frac{\sqrt{\sum_{i=2}^d\lambda_{s,i}^2}}{\sum_{j=1}^n\sqrt{\sum_{i=2}^d\lambda_{j,i}^2}}$ in round $t$ but instead we sample it with probability $\frac{\norm{\A_s\Y_{t-1}}_2}{\norm{\A\Y_{t-1}}_{1,2}}=\frac{\sqrt{\sum_{i=2}^d\zeta_{s,i}^2}}{\sum_{j=1}^n\sqrt{\sum_{i=2}^d\zeta_{j,i}^2}}$. 
From the change of basis matrix $\B$, $\zeta_{s,i}=(1-\O{\tau_i})\lambda_{s,i}\pm\O{\tau_i}\lambda_{s,1}\pm\sum_{j\notin \{i,1\}}\O{\tau_i\tau_j}\lambda_{s,j}$, for $i\ge 2$. 
Note that $\zeta_{s,i}$ has exactly the same form as the proof of \lemref{lem:adaptive:tvd}, since $\zeta_{s,i}$ is derived from the same change of basis matrix $\B$, albeit with different values of $\tau_i$.
It then follows by the same reasoning as (\ref{eqn:coeff:diff}) in the proof of \lemref{lem:adaptive:tvd} that for 
\begin{align*}
|\zeta_{s,i}^2-\lambda_{s,i}^2|&\le25\Big(\tau_i\lambda_{s,i}^2+\tau_i^2\lambda_{s,1}^2+\tau_i\lambda_{s,1}\lambda_{s,i}+\sum_{j,\ell\neq\{i,1\}}\tau_i^2\tau_j\tau_\ell\lambda_{s,j}\lambda_{s,\ell}\\
&+\sum_{j\neq\{i,1\}}\tau_i\tau_j\lambda_{s,i}\lambda_{s,j}+\sum_{j\neq\{i,1\}}\tau_i^2\tau_j\lambda_{s,1}\lambda_{s,j}\Big).
\end{align*}
Thus from AM-GM, we have
\begin{align*}
\tau_i\lambda_{s,1}\lambda_{s,i}&\le\eps^2\lambda_{s,i}^2+\frac{\tau_i^2}{\eps^2}\lambda_{s,1}^2\\
\tau_i^2\tau_j\tau_\ell\lambda_{s,j}\lambda_{s,\ell}&\le\tau_i^2\tau_j^2\lambda_{s,j}^2+\tau_i^2\tau_{\ell}^2\lambda_{s,\ell}^2,\\
\tau_i\tau_j\lambda_{s,i}\lambda_{s,j}&\le\eps^2\lambda_{s,i}^2+\frac{\tau_i^2\tau_j^2}{\eps^2}\lambda_{s,j}^2,\\
\tau_i^2\tau_j\lambda_{s,1}\lambda_{s,j}&\le\tau_i^2\lambda_{s,1}^2+\tau_i^2\tau_j^2\lambda_{s,j}^2
\end{align*}
so that for $\eps<\frac{1}{d}$, we have $|\zeta_{s,i}^2-\lambda_{s,i}^2|\le25\left(2\eps^2\lambda_{s,i}^2+\frac{2\tau_i^2}{\eps^2}\lambda_{s,1}^2+4\sum_{j=2}^d\frac{\tau_i^2\tau_j^2}{\eps^2}\lambda_{s,j}^2\right)$. 
Note that in comparison to \lemref{lem:adaptive:tvd}, we compare the $\tau_i\lambda_{s,1}\lambda_{s,i}$ term with $\eps^2\lambda_{s,i}^2$ rather than $\eps\lambda_{s,i}^2$. 
Therefore,
\begin{align*}
\left|\sqrt{\sum_{i=t}^d\zeta_{s,i}^2}-\sqrt{\sum_{i=t}^d\lambda_{s,i}^2}\right|&\le\sqrt{\left|\sum_{i=t}^d\zeta_{s,i}^2-\sum_{i=t}^d\lambda_{s,i}^2\right|}\\
&\le\sqrt{25\sum_{i=t}^d\left(2\eps^2\lambda_{s,i}^2+\frac{2\tau_i^2}{\eps^2}\lambda_{s,1}^2+4\sum_{j=2}^d\frac{\tau_i^2\tau_j^2}{\eps^2}\lambda_{s,j}^2\right)}\\
&\le5\sqrt{\sum_{i=t}^d 2\eps^2\lambda_{s,i}^2}+5\sqrt{\sum_{i=t}^d\frac{2\tau_i^2}{\eps^2}\lambda_{s,1}^2}+5\sqrt{4\sum_{i=t}^d\sum_{j=2}^d\frac{\tau_i^2\tau_j^2}{\eps^2}\lambda_{s,j}^2}.
\end{align*}
Since $\tau_i=\frac{\eps^3\norm{\A\P_i}_{1,2}}{\norm{\A}_{1,2}}=\frac{\eps^3\sum_{a=1}^n\lambda_{a,i}}{\sum_{a=1}^n\sqrt{\sum_{b=1}^d\lambda_{a,b}^2}}$ and $\eps<\frac{1}{d}$, then we have
\begin{align*}
\sum_{s=1}^n\sqrt{\sum_{i=t}^d\frac{\tau_i^2}{\eps^2}\lambda_{s,1}^2}&\le\left(\sum_{s=1}^n\lambda_{s,1}\right)\sqrt{\sum_{i=t}^d\frac{\tau_i^2}{\eps^2}}\le\eps^2\left(\sum_{s=1}^n\lambda_{s,1}\right)\frac{\sum_{i=t}^d\sum_{a=1}^n\lambda_{a,i}}{\sum_{a=1}^n\sqrt{\sum_{b=1}^d\lambda_{a,b}^2}}\\
&\le\eps^2\sum_{s=1}^n\sum_{i=t}^d\lambda_{s,i}\le\eps\sum_{s=1}^n\sqrt{\sum_{i=t}^d\lambda_{s,i}^2}.
\end{align*}
Similarly, since $\tau_j<\frac{1}{d}$ for each integer $j\in[2,d]$, we have 
\begin{align*}
\sum_{s=1}^n\sqrt{\sum_{i=t}^d\sum_{j=2}^d\frac{\tau_i^2\tau_j^2}{\eps^2}\lambda_{s,j}^2}&\le\sum_{s=1}^n\max_{j\in[2,d]}\sqrt{\sum_{i=t}^d\frac{\tau_i^2}{\eps^2}\lambda_{s,j}^2}\\
&\le\eps^2\left(\sum_{s=1}^n\max_{j\in[2,d]}\lambda_{s,j}\right)\frac{\sum_{i=t}^d\sum_{a=1}^n\lambda_{a,i}}{\sum_{a=1}^n\sqrt{\sum_{b=1}^d\lambda_{a,b}^2}}\\
&\le\eps^2\sum_{s=1}^n\sum_{i=t}^d\lambda_{s,i}\le\eps\sum_{s=1}^n\sqrt{\sum_{i=t}^d\lambda_{s,i}^2}.
\end{align*}
Hence, we have 
\begin{align*}
\sum_{s=1}^n \left|\sqrt{\sum_{i=t}^d\zeta_{s,i}^2}-\sqrt{\sum_{i=t}^d\lambda_{s,i}^2}\right|&\le
200\eps\sum_{s=1}^n\sqrt{\sum_{i=t}^d\lambda_{s,i}^2},
\end{align*}
so that $\norm{\A\Y_{t-1}}_{1,2}$ is once again within a $(1+200\eps)$ factor of $\norm{\A\Z_{t-1}}_{1,2}$, from which we can bound the total variation distance by $800\eps$, including the $\frac{1}{\poly(n)}$ event of failure from \lemref{lem:l12:orthogonal:noise}. 
It follows from induction that the total variation distance across $k$ rounds is at most $800k\eps$.
\end{proof}
Thus we can also approximately simulate adaptive sampling in a stream with respect to the distances to the subspace spanned by the previously sampled rows. 
\begin{theorem}
\thmlab{thm:l12:adaptive:sampler}
Given a matrix $\A\in\mathbb{R}^{n\times d}$ that arrives in a turnstile data stream, there exists a one-pass streaming algorithm \adaptiveone{} that outputs a set of $k$ indices such that the probability distribution for each set of $k$ indices has total variation distance $\eps$ of the probability distribution induced by adaptive sampling with respect to the distances to the subspace in each iteration. 
The algorithm uses $\poly\left(d,k,\frac{1}{\eps},\log n\right)$ bits of space. 
\end{theorem}
\begin{proof}
Like the proof of \thmref{thm:adaptive:sampler}, we again consider a set of processes $\mathcal{P}_1,\mathcal{P}_2,\ldots,\mathcal{P}_{k+1}$, where for each $i\in[k+1]$, $\mathcal{P}_i$ is a process that samples noisy rows from the $L_{1,2}$ sampler for the first $i-1$ rounds and actual rows from $\A$ beginning with round $i$, through adaptive sampling with $p=1$. 
Then $\mathcal{P}_1$ is the actual adaptive sampling process, while $\mathcal{P}_{k+1}$ is the noisy process of \algref{alg:l12:noisy:adaptive}. 
Then \lemref{lem:l12:adaptive:tvd} argues that the total variation distance between the output distributions of the $k$ indices sampled by $\mathcal{P}_1$ and $\mathcal{P}_2$ is at most $\O{k\eps}$. 
Moreover, the total variation distance between the output distributions of the indices sampled by $\mathcal{P}_i$ and $\mathcal{P}_{i+1}$ is at most $\O{k\eps}$ for any $i\in[k]$ since the sampling distributions of $\mathcal{P}_i$ and $\mathcal{P}_{i+1}$ is identical in the first $i$ rounds, so we can use the same argument starting at round $i$ using the input matrix $\A\Q$ rather than $\A$, where $\Q$ is the projection matrix away from the noisy rows sampled in the first $i$ rounds. 
Now by a triangle inequality argument over the $k+1$ processes $\mathcal{P}_1,\ldots,\mathcal{P}_{k+1}$, the total variation distance between the probability distribution of the $k$ indices output by \algref{alg:l12:noisy:adaptive} and the probability distribution of the $k$ indices output by adaptive sampling is at most $\O{k^2\eps}$. 
Hence the total variation distance is at most $\eps$ after the appropriate rescaling factor. 

For the space complexity, observe that we run $k$ instances of \algref{alg:l12:sampler}, each using $b=\poly\left(k,\frac{1}{\eps},\log n\right)$ buckets due to the error parameter $\frac{\eps}{dk^2}$. 
\algref{alg:l12:sampler} uses $\poly\left(k,\frac{1}{\eps},\log n\right)$ CountSketch data structures that each use $\poly\left(d,k,\frac{1}{\eps},\log n\right)$ bits of space. 
Moreover by \lemref{lem:estimator}, each instance of \estimator{} uses $d\,\polylog(n)$ bits of space. 
Hence, the total space complexity is $\poly\left(d,k,\frac{1}{\eps},\log n\right)$. 
\end{proof}
Finally, by the same argument as \corref{cor:sampler:distortion}, we have the following: 
\begin{corollary}
\corlab{cor:l12:sampler:distortion}
Suppose \algref{alg:l12:noisy:adaptive} samples noisy rows $\r_1,\ldots,\r_k$ rather than the actual rows $\A_{f(1)},\ldots,\A_{f(k)}$. 
Let $\T_k=\A_{f(1)}\circ\ldots\circ\A_{f(k)}$, $\Z_k=\I-\T_k^\dagger\T_k$, $\R_k=\r_1\circ\ldots\circ\r_k$ and $\Y_k=\I-\R_k^\dagger\R_k$. 
Then $(1-\eps)\norm{\A\Y_k}_{1,2}\le\norm{\A\Z_k}_{1,2}\le(1+\eps)\norm{\A\Y_k}_{1,2}$ with probability at least $1-\eps$. 
\end{corollary}

\end{document}